\documentclass[a4paper,11pt,USenglish]{article}

\usepackage{amsmath,amssymb,xspace}
\usepackage[utf8]{inputenc}
\usepackage[T1]{fontenc}
\usepackage{amsthm}
\usepackage{graphicx}
\usepackage{url}
\usepackage{calc}
\usepackage{hyperref}
\usepackage[left=20mm,right=20mm,bottom=23mm,top=20mm]{geometry}
\usepackage[nodisplayskipstretch]{setspace} 
\usepackage{comment}
\usepackage[boxruled,noline,noend,linesnumbered]{algorithm2e}
\usepackage{graphicx}
\usepackage[table]{xcolor}
\usepackage{tikz}
\usetikzlibrary{patterns}
\usepackage{thm-restate}
\usepackage{pifont}
\usepackage{mathdots}
\usepackage{multirow}
\usepackage{marginnote}

\usepackage{float}
\usepackage{xspace}
\usepackage{listings}
\usepackage{enumitem}
\usepackage{mathtools}
\usepackage{thm-restate}
\usepackage[outline]{contour}
\usepackage{etoolbox}
\usepackage{array}

\newcolumntype{L}[1]{>{\raggedright\let\newline\\\arraybackslash\hspace{0pt}}m{#1}}
\newcolumntype{C}[1]{>{\centering\let\newline\\\arraybackslash\hspace{0pt}}m{#1}}
\newcolumntype{R}[1]{>{\raggedleft\let\newline\\\arraybackslash\hspace{0pt}}m{#1}}

\hypersetup{
 colorlinks,
 linkcolor={blue!80!black},
 citecolor={blue!80!black},
 urlcolor={blue!80!black}
}

\tikzset{
 treenode/.style = 
 {draw,very thick,circle,minimum size=3.5mm,fill=black!30,inner sep=0pt, outer sep=0pt},
 leaf/.style = 
 {draw,thick,circle,minimum size=2.5mm,fill=black!10,inner sep=0pt, outer sep=0pt},
 hatch distance/.store in=\hatchdistance,
 hatch distance=10pt,
 hatch thickness/.store in=\hatchthickness,
 hatch thickness=2pt
}

\makeatletter
\pgfdeclarepatternformonly[\hatchdistance,\hatchthickness]{flexible hatch}
{\pgfqpoint{0pt}{0pt}}
{\pgfqpoint{\hatchdistance}{\hatchdistance}}
{\pgfpoint{\hatchdistance-1pt}{\hatchdistance-1pt}}%
{
\pgfsetcolor{\tikz@pattern@color}
\pgfsetlinewidth{\hatchthickness}
\pgfpathmoveto{\pgfqpoint{0pt}{0pt}}
\pgfpathlineto{\pgfqpoint{\hatchdistance}{\hatchdistance}}
\pgfusepath{stroke}
}
\makeatother

\definecolor{pblue}{rgb}{0.36,0.08,0.57}
\definecolor{pgreen}{rgb}{0,0.5,0}
\definecolor{pred}{rgb}{0.9,0,0}
\definecolor{pgrey}{rgb}{0.46,0.45,0.48}

\newcommand{\NMS}{\textsf{NaturalMergeSort}\xspace}
\newcommand{\MS}{\textsf{MergeSort}\xspace}
\newcommand{\HeS}{\textsf{HeapSort}\xspace}
\newcommand{\SpS}{\textsf{SplaySort}\xspace}
\newcommand{\SmS}{\textsf{SmoothSort}\xspace}
\newcommand{\MinS}{\textsf{MinimalSort}\xspace}
\newcommand{\MinSS}{\textsf{MinimalStableSort}\xspace}
\newcommand{\TS}{\textsf{TimSort}\xspace}
\newcommand{\aMS}{\textsf{$\alpha$-MergeSort}\xspace}
\newcommand{\PoS}{\textsf{PowerSort}\xspace}
\newcommand{\PeS}{\textsf{PeekSort}\xspace}
\newcommand{\ShS}{\textsf{ShiversSort}\xspace}
\newcommand{\AugSS}{\textsf{augmented ShiversSort}\xspace}
\newcommand{\oASS}[1]{\textsf{$#1$-adaptive ShiversSort}\xspace}
\newcommand{\ccASS}{\oASS{c}}
\newcommand{\dASS}{\textsf{adaptive ShiversSort\textsubscript{v2}}\xspace}
\newcommand{\CASS}{\textsf{Adaptive ShiversSort}\xspace}
\newcommand{\cASS}{\textsf{adaptive ShiversSort}\xspace}
\newcommand{\cASSm}{\textsf{$\kappa$-stack ShiversSort}\xspace}
\newcommand{\cASSM}{\textsf{$\sk$-stack ShiversSort}\xspace}
\newcommand{\LASS}{\textsf{length-adaptive ShiversSort}\xspace}

\newcommand{\QS}{\textsf{QuickSort}\xspace}
\newcommand{\mustMerge}{\textsf{mustMerge}\xspace}
\newcommand{\doMerge}{\textsf{doMerge}\xspace}
\newcommand{\succc}{\textsf{succ}}
\newcommand{\succl}{\textsf{succ$_\kappa$}}

\newcommand{\sk}{\star\kappa}
\newcommand{\sm}{\mathsf{s}_{\min}}
\newcommand{\lm}{\ell_{\min}}
\newcommand{\vv}{\hspace*{3.5mm},\hspace*{3.5mm}}
\newcommand{\ub}[3]{${\color{white}\underbrace{{\color{black}#1}}_{\text{\color{black}\begin{tabular}{c}\\[-18pt]{\makebox[0pt][c]{#2}}
\\[-4pt]{\makebox[0pt][c]{#3}}\end{tabular}}}}$}

\newcommand{\leaf}[3]{
 \node[leaf] at (#1) {};
 \node[anchor=north] at (#1-#2) {#3};
}
\newcommand{\roundnode}[3]{
 \draw[very thick,fill=white] (#1,#2) circle(0.8);
 \node[anchor=north] at (#1,#2+0.55) {#3};
}
\newcommand{\roundvalnode}[3]{
 \draw[very thick,fill=white] (#1+0.375,#2+0.8) arc (90:-90:0.8) -- 
 (#1-0.375,#2-0.8) arc (270:90:0.8) -- cycle;
 \node[anchor=north] at (#1,#2+0.65) {#3};
}
\newcommand{\ovalnode}[3]{
 \draw[very thick,fill=white] (#1+0.75,#2+0.8) arc (90:-90:0.8) -- 
 (#1-0.75,#2-0.8) arc (270:90:0.8) -- cycle;
 \node[anchor=north] at (#1,#2+0.65) {#3};
}

\newcommand{\start}{{\mathsf{start}}}
\newcommand{\final}{{\mathsf{end}}}

\newcommand{\crossmark}{{\ding{55}}}
\newcommand{\rrlap}[2]{\text{\rlap{$#1$}}\phantom{#2}}

\newcounter{leqn}
\makeatletter
\newcommand{\leqn}[1]{\addtocounter{leqn}{1}%
\def\@currentlabel{\theleqn}%
\begin{center}\hfill$\displaystyle #1$\hfill\textnormal{(\theleqn)}\end{center}}
\makeatother

\newcommand{\A}{\mathcal{A}}
\newcommand{\B}{\mathcal{B}}
\newcommand{\F}{\mathcal{F}}
\renewcommand{\H}{\mathcal{H}}
\newcommand{\M}{\mathcal{M}}
\renewcommand{\O}{\mathcal{O}}
\newcommand{\R}{\mathcal{R}}
\renewcommand{\S}{\mathcal{S}}
\newcommand{\T}{\mathcal{T}}

\newcommand{\Ms}{\mathsf{M}_{\mathrm{seq}}}

\newcommand{\oR}{\overline{R}}
\newcommand{\oS}{{\overline{\S}}}
\newcommand{\od}{\overline{d}}
\newcommand{\oh}{{\overline{h}}}
\newcommand{\ok}{\overline{k}}
\newcommand{\ol}{{\overline{\ell}}}
\newcommand{\orr}{\overline{r}}

\newcommand{\bA}{{\mathbf{A}}}
\newcommand{\bD}{{\mathbf{D}}}
\newcommand{\bI}{{\mathbf{I}}}
\newcommand{\bJ}{{\mathbf{J}}}
\newcommand{\bK}{{\mathbf{K}}}
\newcommand{\bS}{{\mathbf{S}}}

\newcommand{\obS}{{\overline{\bS}}}

\DeclareMathOperator{\rundecomp}{\texttt{runs}}

\newcommand{\true}{\textbf{true}\xspace}
\newcommand{\false}{\textbf{false}\xspace}
\newcommand{\return}{\textbf{return}\xspace}

\newcommand{\ass}{{\mathsf{ass}}}
\newcommand{\mc}{{\mathsf{mc}}}
\newcommand{\mcass}{{\mc_\ass}}
\newcommand{\mcopt}[1]{\mc_{\mathsf{opt}}^{#1}}

\newcommand{\llarge}{{\text{\textsf{large}}}}
\newcommand{\Pot}{{\mathsf{Pot}}}
\newcommand{\POT}[2]{{\Pot^{#1}_{#2}}}

\newcommand{\runend}{\mathsf{pos}}

\renewcommand{\gets}{\ensuremath{\leftarrow}}

\let\oldnl\nl
\newcommand{\drawline}{\BlankLine\renewcommand{\nl}{\let\nl\oldnl}\hrulefill\BlankLine\setcounter{AlgoLine}{0}}
\SetKwIF{If}{ElseIf}{Else}{if}{\!:}{else if}{else:}{endif}
\SetKwFor{For}{for}{\!:}{endf   }
\SetKwFor{While}{while}{\!:}{endw}
\SetKwComment{tcp}{ {$\rhd$} }{}
\SetCommentSty{textrm}
\SetInd{0.5em}{0.5em}

\theoremstyle{plain}
\newtheorem{theorem}{Theorem}
\newtheorem{proposition}[theorem]{Proposition}
\newtheorem{lemma}[theorem]{Lemma}

\newtheorem{corollary}[theorem]{Corollary}

\theoremstyle{definition}
\newtheorem{definition}[theorem]{Definition}

\title{Adaptive Shivers Sort: An Alternative Sorting Algorithm} 

\author{Vincent Jugé}

\date{Université Gustave Eiffel, LIGM (UMR 8049), CNRS, ENPC, ESIEE Paris, UPEM}

\begin{document}


\maketitle

\begin{abstract}
We present a new sorting algorithm, called \cASS, that exploits the existence of monotonic runs for sorting efficiently partially sorted data.
This algorithm is a variant of the well-known algorithm \TS, which is
the sorting algorithm used in standard libraries of programming languages
such as Python or Java (for non-primitive types).
More precisely, \cASS is a so-called $k$-aware
merge-sort algorithm, a class that was introduced by
Buss and Knop that captures ``\TS-like'' algorithms.

In this article, we prove that, although \cASS
is simple to implement and differs only slightly from \TS,
its computational cost, in number of comparisons performed,
is optimal within the class of \emph{natural} merge-sort algorithms,
up to a small additive linear term:
this makes \cASS the first $k$-aware algorithm to benefit from this property,
which is also a 33\% improvement over \TS's worst-case.
This suggests that \cASS could be a strong contender for being used instead of \TS.

Then, we investigate the optimality of $k$-aware algorithms:
we give lower and upper bounds on the best approximation factors
of such algorithms, compared to optimal stable natural
merge-sort algorithms.
In particular, we design generalisations of \cASS whose computational costs
are optimal up to arbitrarily small multiplicative factors.
\end{abstract}

\bigskip


\section{Introduction}\label{sec:intro}

The problem of sorting data has been one of the first and most extensively
studied problems in computer science, and sorting is ubiquitous,
due to its use as a sub-routine in a wealth of various algorithms.
Hence, as early as the 1940's, sorting algorithms were invented,
which enjoyed many optimality properties regarding their
complexity in time (and, more precisely, in number of comparisons 
or element moves required) as well as in memory.
Every decade or so, a new major sorting algorithm was invented,
either using a different approach to sorting or adapting specifically tuned
data structures to improve previous algorithms:
\MS~\cite{goldstine1947planning}, \QS~\cite{hoare1961algorithm},
\HeS~\cite{williams1964algorithm}, \SmS~\cite{dijkstra1982smoothsort},
\SpS~\cite{moffat1996splaysort}, \ldots

In 2002, Tim Peters, a software engineer, created a new sorting algorithm,
which was called \TS~\cite{Peters2015}. This algorithm immediately 
demonstrated its efficiency for sorting actual data, and was
adopted as the standard sorting algorithm in core libraries of
wide-spread programming languages such as Python and Java.
Hence, the prominence of such a custom-made algorithm over previously
preferred \emph{optimal} algorithms contributed to the regain of interest
in the study of sorting algorithms.

Understanding the reasons behind the success of \TS is still an ongoing task.
These reasons include the fact that \TS is well adapted to the architecture of
computers (e.g., for dealing with cache issues) and to realistic
distributions of data. In particular, a model that successfully explains
why \TS is adapted to sorting realistic data involves \emph{run
decompositions}~\cite{BaNa13,EsCaWo92}, as illustrated in Figure~\ref{fig:runs}. 
Such decompositions were already used by
Knuth \NMS~\cite{Knuth98}, which predated \TS, and adapted the traditional
\MS algorithm as follows: \NMS is based on splitting arrays into monotonic subsequences, also called \emph{runs}, and on merging these runs together.
Thus, all algorithms sharing this feature of \NMS are also called
\emph{natural} merge sorts.

\begin{figure}[ht]
\centerline{$
S=(\,\,\underbrace{12,7,6,5}_{\text{first run}},
\,\,\underbrace{5,7,14,36}_{\text{second run}},
\,\,\underbrace{3,3,5,21,21}_{\text{third run}},
\,\,\underbrace{20,8,5,1}_{\text{fourth run}}\,\,)
$}
\caption{A sequence and its \emph{run decomposition} computed by a greedy algorithm:
for each run, the first two elements determine if it is non-decreasing or decreasing, then it continues with the maximum number of consecutive elements that preserves the monotonicity.\label{fig:runs}}
\end{figure}

In addition to being a natural merge sort, \TS also includes many
optimisations, which were carefully engineered, through extensive
testing, to offer the best complexity performances.
As a result, 
the general structure of \TS can be split into three main components:
(i)~a complicated variant of an insertion sort, 
which is used to deal with \emph{small} runs
(e.g., runs of length less than 32),
(ii)~a simple policy for choosing which \emph{large} runs to merge,
(iii)~a complex sub-routine for merging these runs.
The first and third components were those which 
were the most finely tuned,
hence understanding the subtleties of why they are 
efficient and how they
could be improved seems difficult. The second component, however, is
quite simple, and therefore it offers the best opportunities for
modifying and improving \TS.

\paragraph{Context and related work}

The success of \TS has nurtured the interest in the quest for
sorting algorithms that would be adapted to arrays with few runs.
However, the \emph{ad hoc} conception of \TS made its complexity analysis
less easy than what one might have hoped, and it is only in 2015, 
a decade after \TS had been largely deployed, that Auger et al. proved that
\TS required $\O(n \log(n))$ comparisons
for sorting arrays of length $n$~\cite{AuNiPi15}.

Even worse, because of the lack of a systematic and theoretical analysis of this
algorithm, several bugs were discovered only recently 
in both Python and Java implementations of
\TS~\cite{GoRoBoBuHa15,auger2018worst}.

Meanwhile, since \TS was invented,
several natural merge sorts have been proposed,
all of which were meant to 
offer easy-to-prove complexity guarantees. Such algorithms include
\ShS, introduced by Shivers in~\cite{shivers02}, as well as Takaoka's~\MinS~\cite{Ta09}
(equivalent constructions of this algorithm were also obtained in~\cite{BaNa13}),
Buss and Knop's \aMS~\cite{BuKno18}, and
the most recent algorithms \PeS and \PoS, due to Munro and Wild~\cite{munro2018nearly}.
Alternatively, as we will mention again, algorithms for constructing
\emph{optimal binary search trees}, such as the algorithms of
Hu-Tucker~\cite{HuTucker71} and Garsia-Wachs~\cite{Garsia77}, 
can be adapted to provide natural merge sorts as well;
we call~\MinSS the merge sort derived from adapting either algorithm.

\begin{table}[ht]
\begin{center}
\begin{tabular}{|l|l|l|l|l|}
\hline
Algorithm & Time complexity & Stable & $k$-aware & Worst-case merge cost \\
\hline
\NMS & $\O(n + n \log(\rho))$ & \checkmark & \crossmark & 
$\rrlap{n \log_2(\rho)}{n \log_2(n)} + \O(n)$ \\
\hline
\TS & $\O(n + n \H)$ & \checkmark & $k=4$ & 
$\rrlap{3/2 \, n \H}{n \log_2(n)} + \O(n)$ \\
\hline
\ShS & $\O(n \log(n))$ & \checkmark & $k=2$ & 
$n \log_2(n) + \O(n)$ \\
\hline
\MinS & $\O(n + n \H)$ & \crossmark & \crossmark & 
$\rrlap{n \H}{n \log_2(n)} + \O(n)$ \\
\hline
\MinSS & $\O(n + n \H)$ & \checkmark & \crossmark & 
$\rrlap{n \H}{n \log_2(n)} + \O(n)$ \\
\hline
\aMS & $\O(n + n \H)$ & \checkmark & $k=3$ & 
$\rrlap{c_\alpha n \H}{n \log_2(n)} + \O(n)$  \\
\hline
\PoS & $\O(n + n \H)$ & \checkmark & \crossmark & 
$\rrlap{n \H}{n \log_2(n)} + \O(n)$ \\
\hline
\PeS & $\O(n + n \H)$ & \checkmark & \crossmark & 
$\rrlap{n \H}{n \log_2(n)} + \O(n)$ \\
\hline
\hline
\CASS & $\O(n + n \H)$ & \checkmark & $k=3$ & 
$\rrlap{n \H}{n \log_2(n)} + \O(n)$ \\
\hline
\end{tabular}
\end{center}
\caption{Properties of a few natural merge sorts -- The constant $c_\alpha$ is such that $1.04 < c_\alpha < 1.09$ 
The \emph{merge cost} of an algorithm is an upper bound on its number of comparisons and element moves.\label{fig:prop:sort}}
\end{table}

These algorithms share most of the nice properties of \TS,
as summarised in Table~\ref{fig:prop:sort} (columns 1--3).
For instance, except~\MinS, these are \emph{stable} algorithms,
which means that they sort
repeated elements in the same order as these elements appear in the input.
This is very important for merge sorts, because only adjacent runs will be
merged, which allows merging directly arrays instead of having to use
linked lists.
This feature is also important for merging composite types
(e.g., non-primitive types in Java), which might be sorted twice according
to distinct comparison measures.
Moreover, 
all these algorithms sort arrays of length $n$ in time $\O(n\log(n))$,
and, for all of them except \ShS, they even do it in time
$\O(n+n\log(\rho))$,
where $\rho$ is the number of runs of the array.
This is optimal in the model of sorting by comparisons~\cite{Mannila1985}, using the classical counting argument for lower bounds.

Some of these algorithms even adapt to the lengths of the runs, 
and not only to the number of runs:
if the array consists of $\rho$ runs of lengths $r_1,\ldots,r_\rho$,
these algorithms run in $\O(n + n \H)$, where $\H$ is defined as
$\H = H(r_1/n,\ldots,r_\rho/n)$ and
$H(x_1,\ldots,x_\rho) = -\sum_{i=1}^\rho x_i \log_2(x_i)$ is the
general entropy function.
Considering the number of runs and their lengths as parameters,
this finer upper bound is
again optimal in the model of sorting by comparisons~\cite{BaNa13}.

Focusing only on the time complexity, six algorithms seem on par with 
each other, and finer complexity evaluations are required to separate them.
Except \TS, it turns out that these algorithms are, in fact, described only
as policies for merging runs, the actual sub-routine used for merging runs
being left implicit.
Therefore, we settle for the following cost model.

Since naive merging algorithms approximately 
require $m+n$ element comparisons and element moves
for merging two arrays of lengths $m$ and $n$, 
and since $m+n$ element moves may be needed
in the worst case (for any values of $m$ and $n$), 
we measure below the complexity in terms of
\emph{merge cost}~\cite{AuNiPi15,BuKno18,golin1993queue,munro2018nearly}:
the cost of merging two runs of lengths $m$ and $n$ is defined as $m+n$,
and we identify the complexity of an algorithm
with the sum of the costs of the merges processed
while applying the run merge policy of this algorithm.

Of course, this identification can be legitimate only if this sum of merge costs
dominates the complexity of deciding which runs should be merged.
Fortunately, this is the case in
all of the algorithms presented in Table~\ref{fig:prop:sort}:
the latter complexity is $\O(\rho \log(\rho))$,
and therefore $\O(n \H)$ too, for the algorithms \MinS and \MinSS,
and it is $\O(n)$ for the other algorithms presented.

In this new model, every run merge policy can be identified with
a bottom-up construction algorithm for binary search trees,
an idea that was already noted and used successfully in~\cite{BaNa13,munro2018nearly}.
In particular, one can prove that the merge cost of any
natural merge sort must be at least $n \H + \O(n)$.
This makes \MinS, \MinSS, 
\PeS and \PoS the only sorting algorithms with an optimal merge cost,
as shown in the last column of Table~\ref{fig:prop:sort}.

In another direction, and since \MinSS, \PeS and \PoS are stable, 
they could be considered natural options for succeeding \TS
as standard sorting algorithm in Python or Java.
Nevertheless, and although the latter two algorithms
have implementations similar to
that of \TS, their merge policies are arguably more complicated,
as illustrated in Section~\ref{sec:description}.

Therefore, there is yet to find a natural merge sort whose structure
would be extremely close to that of \TS,
and whose merge cost would also be optimal 
up to an additive term $\O(n)$.

A first step towards this goal is using an adequate notion of
``\TS-likeness'', and therefore we look at the class of
$k$-aware sorting algorithms. 
This class of algorithms was invented
by Buss and Knop~\cite{BuKno18}, with the explicit goal of 
characterising those algorithms
whose merge policy is similar to that of \TS. More precisely,
\TS is based on discovering runs on the fly, and ``storing'' these runs
into a stack: if a run spans the $i$\textsuperscript{th} to
$j$\textsuperscript{th} entries of the array, then the stack
will contain the pair $(i,j)$.
Then, \TS merges only runs that lie on the top of the stack,
and such decisions are based only on the lengths of these top runs.

The rationale behind this process is that processing runs in such a way
should be adapted to the architecture of computers, for instance by
avoiding cache misses. Then, one says that a natural merge sort
is $k$-aware if deciding which runs should be merged is based only
on the lengths of the top $k$ runs of the stack, and if the runs merged
belong themselves to these top $k$ runs.
The 4\textsuperscript{th} column of Table~\ref{fig:prop:sort}
indicates which algorithms are $k$-aware for some $k < +\infty$,
in which case it also gives the smallest such $k$.

Focusing on $k$-aware algorithms seems all the more relevant
because some of the nice features of \TS were also due to the high degree
of tuning of the components (i) and (iii). Hence, if one does not modify
these components, and if one follows a merge policy that behaves in a way
similar to that of \TS, one may reasonably hope that those nice features of
\TS would be kept intact, even though their causes are not
exactly understood. This suggests identifying natural merge sorts with their
merge policy, and integrating the components (i) and (iii) later.

\paragraph{Contributions}

We propose a new natural merge sort, which we call \cASS.
As advertised above, we will identify this algorithm with its run merge policy.
\CASS is a blend between the algorithms \TS and \ShS;
the purpose being to borrow nice properties from both algorithms.
As a result, the merge policy of \cASS is extremely similar to that of \TS,
which means that switching
from one algorithm to the other should be essentially costless, since
it would require changing only a dozen lines in the code of Java.

\CASS is a $3$-aware algorithm, which is stable
and enjoys an optimal $n\H + \O(n)$ upper bound on its merge cost.
Hence, \cASS appears as optimal with respect to all the criteria
mentioned in Table~\ref{fig:prop:sort}; it is the first known
$k$-aware algorithm with a merge cost of $n\H + \O(n)$,
thereby answering a question left open by Buss and Knop in~\cite{BuKno18}.
Moreover, and due to its simple policy,
the running time complexity proof of \cASS is simple as well;
below, we propose a short, self-contained version of this proof.

Then, still aiming to compare \PeS, \PoS and \cASS,
we investigate their \emph{best-case} merge costs.
It turns out that the merge cost of \PoS, in every case, is bounded between
$n \H$ and $n (\H+2)$, these bounds being both close and
optimal for any stable merge sort. The merge cost of \PeS is only
bounded from above by $n (\H+3)$; this is slightly
worse than \PoS, and thus we will not discuss this
algorithm below.
Similarly, the merge cost of \cASS is 
only bounded from above by $n(\H + \Delta)$,
where $\Delta = 24/5 - \log_2(5) \approx 2.478$, which is 
also slightly worse than \PoS.
Hence, we design a variant of \cASS, called \LASS,
which is not $3$-aware,
but whose merge cost also enjoys an $n (\H+2)$ upper bound.

Finally, we further explore the question, raised by Buss and Knop in~\cite{BuKno18}, of the optimality of $k$-aware algorithms on \emph{all} arrays.
More precisely, this question can be stated as follows:
for a given integer $k$ and a real number $\varepsilon > 0$,
does there exist a $k$-aware algorithm whose merge cost is at most
$1+\varepsilon$ times the merge cost of any stable merge cost on
any array?
We prove that the answer is always negative when $k = 2$;
for all $k \geqslant 3$, the set of numbers $\varepsilon$
for which the answer is positive forms an interval with no upper bound,
and we prove that the lower bound of that interval is a positive real number,
which tends to $0$ when $k$ grows arbitrarily.

\section{\CASS and related algorithms}
\label{sec:description}

In this section, we describe
the run merge policy of the algorithm \cASS and of
related algorithms.
The merge policy of \cASS is depicted in
Algorithm~\ref{alg:ASS}. For the ease of Section~\ref{subsec:analysis-1}, 
and because it does not make any
proof harder, we shall consider \cASS as a
special case of the parameterised algorithm \ccASS:
in addition to the array to sort,
\ccASS also requires a positive integer $c$ as parameter.

\begin{algorithm}[ht]
\begin{small}
\SetArgSty{texttt}
\DontPrintSemicolon
\SetKwInOut{Input}{Input}
\Input{Array $A$ to sort, integer parameter $c$}
\KwResult{The array $A$ is sorted into a single run.
That run remains on the 
stack.}
\SetKwInput{KwData}{Note}
\KwData{We \hfill denote \hfill the \hfill height \hfill of \hfill
the \hfill stack \hfill $\S$ \hfill by \hfill $h$,\hfill 
and \hfill its \hfill $i$\textsuperscript{th} \hfill bottom-most \hfill
run \hfill (for \hfill $1 \leqslant i 
\leqslant h$) \hfill by \hfill $R_i$.
The \hfill length \hfill of \hfill $R_i$ \hfill is \hfill denoted \hfill by 
\hfill $r_i$, \hfill
and \hfill we \hfill set \hfill $\ell_i = \lfloor \log_2(r_i / c) \rfloor$. \hfill 
Whenever \hfill two \hfill consecutive
runs \hfill of \hfill $\S$ \hfill
are \hfill merged, \hfill they \hfill are \hfill replaced, \hfill in \hfill $\S$, \hfill
by \hfill the \hfill run \hfill resulting \hfill from \hfill the 
\hfill merge. \hfill In \hfill
practice, in $\S$, each run is represented by a pair of pointers to
its first and last entries.}
\BlankLine
$\rundecomp \gets$ the run decomposition of $A$\;
$\S \gets $ an empty stack\;
\While(\tcp*[f]{main loop}){\true}{
    \If
    {\textrm{$h \geqslant 3$ and $\ell_{h-2} \leqslant \max\{\ell_{h-1},\ell_h\}$}}
        {merge the runs $R_{h-2}$ and $R_{h-1}$\label{alg:ASS:merge}}
    \ElseIf{$\rundecomp \neq \emptyset$}
      {remove a run $R$ from $\rundecomp$ and
      push $R$ onto $\S$\label{alg:ASS:push}}
    \Else{break\label{algline:end_inner_loop}}
}
\While{$h \geqslant 2$}{
  merge the runs $R_{h-1}$ and $R_h$
  \label{alg:collapse}
}
\end{small}
\caption{\ccASS\label{alg:ASS}}
\end{algorithm}

In subsequent sections,
we will consider most specifically two choices for the
parameter $c$:
we may either set $c = 1$,
thereby obtaining the algorithm \cASS itself,
or $c = n+1$, where $n$ is the length of the array to be sorted,
thereby obtaining the algorithm \LASS.
Except in Section~\ref{subsec:analysis-1}, where we run a complexity analysis for
generic values of the parameter $c$ (thus encompassing both cases $c = 1$
and $c = n+1$ at once), we will only focus on
the algorithm \cASS, i.e., we will set $c = 1$.

This algorithm is based on discovering monotonic runs and
on maintaining a stack of such runs, which may be merged or pushed onto the stack according
to whether $\ell_{h-2} \leqslant \max\{\ell_{h-1},\ell_h\}$.
In particular, since this inequality
only refers to the values of $\ell_{h-2}$, $\ell_{h-1}$ and $\ell_h$,
and since only the runs $R_{h-2}$, $R_{h-1}$ and $R_h$ may be merged,
this algorithm falls within the class of $3$-aware stable sorting algorithms such as described by
Buss and Knop~\cite{BuKno18} as soon as the value of $c$ does not depend on the input. This is the case of \cASS,
but not of other variants such as \LASS.

Let us then present briefly some related algorithms.
Like \cASS, these algorithms all rely on
discovering and maintaining runs in a stack,
although their merge policies follow different rules.
In fact, each of these policies is obtained by modifying
the \emph{main loop} of \cASS.

\RestyleAlgo{boxruled}
\begin{algorithm}[ht]
\begin{small}
\SetArgSty{texttt}
\DontPrintSemicolon
\setcounter{AlgoLine}{2}
\While(\tcp*[f]{main loop}){\true}{
    \If
    {\textrm{$h \geqslant 3$ and $r_{h-2} < r_h$}}
    {merge the runs $R_{h-2}$ and $R_{h-1}$}
    \ElseIf
    {\textrm{$h \geqslant 2$ and $r_{h-1} \leqslant r_h$}}
    {merge the runs $R_{h-1}$ and $R_h$}
    \ElseIf
    {\textrm{$h \geqslant 3$ and $r_{h-2} \leqslant r_{h-1} + r_h$}}
    {merge the runs $R_{h-1}$ and $R_h$}
    \ElseIf
    {\textrm{$h \geqslant 4$ and $r_{h-3} \leqslant r_{h-2} + r_{h-1}$\label{alg:TS-missing-case-1}}}
    {merge the runs $R_{h-1}$ and $R_h$\label{alg:TS-missing-case-2}}
    \ElseIf
    {$\rundecomp \neq \emptyset$}
      {remove a run $R$ from $\rundecomp$ and
      push $R$ onto $\S$}
    \Else{break}
}
\end{small}
\caption{\TS main loop\label{alg:TS}}
\end{algorithm}

The first algorithm we present is \TS, and is due to Peters~\cite{Peters2015}. Its main loop
is presented in Algorithm~\ref{alg:TS}.
As mentioned in the introduction, this algorithm enjoys
a $\O(n + n \H)$ worst-case merge cost.
Two crucial elements in achieving this result consist in
showing that some invariant (the fact that
$r_i \geqslant r_{i+1} + r_{i+2}$ for all $i \leqslant h-2$)
is maintained, and in avoiding merging newly pushed runs
until their length is comparable with that of the runs stored
in the top of the stack. Original versions of \TS
missed the test on lines~\ref{alg:TS-missing-case-1}--\ref{alg:TS-missing-case-2}, which made the invariant
invalid and caused several
implementation bugs~\cite{GoRoBoBuHa15,auger2018worst}.
Yet, even that flawed version of the algorithm
managed to enjoy a
$\O(n + n \H)$ worst-case merge cost.
Unfortunately,
the constant hidden in the $\O$ is rather high,
as outlined by the following result.

\begin{theorem}\label{thm:TS-merge-cost}
The worst-case merge cost of \TS on inputs of length $n$ 
is bounded from above by $3/2 \, n \H + \O(n)$ and bounded from below by
$3/2 \, n \log_2(n) + \O(n)$.
\end{theorem}

\begin{proof}
The lower bound was proven in~\cite{BuKno18}.
The upper bound is proven in the full version of~\cite{auger2018worst}.
\end{proof}

Hence, and in order to lower the constant from $3/2$ to $1$,
it was important to look for other merge policies.
With this idea in mind, one hope comes from the algorithm
\ShS, which was invented by Shivers~\cite{shivers02}. It is obtained by
slightly simplifying the tests carried in the main loop
of~\cASS
and merging the runs $R_{h-1}$ and $R_h$ instead of $R_{h-2}$ and $R_{h-1}$.
More precisely, \ShS uses the following main loop:

\RestyleAlgo{plain}
\begin{algorithm}[ht]
\begin{small}
\SetArgSty{texttt}
\DontPrintSemicolon
\setcounter{AlgoLine}{2}
\While(\tcp*[f]{main loop}){\true}{
    \If
    {\textrm{$h \geqslant 2$ and $\ell_h \geqslant \ell_{h-1}$}}
        {merge the runs $R_{h-1}$ and $R_h$}
    \ElseIf
    {$\rundecomp \neq \emptyset$}
      {remove a run $R$ from $\rundecomp$ and
      push $R$ onto $\S$}
    \Else{break}
}
\end{small}
\end{algorithm}

This sorting algorithm \emph{may} merge a newly pushed
run even if its length is much larger than those of the runs
stored in the stack, and therefore it
is \emph{not} adaptive to the
number of runs nor to their lengths.
Nevertheless, it still enjoys the nice
property of having a worst-case merge cost that is optimal
up to an additive linear term, when the only complexity
parameter is $n$.

\begin{theorem}
The worst-case merge cost of \ShS on inputs of length $n$
that decompose into
$\rho$ monotonic runs is both bounded from above by
$n \log_2(n) + \O(n)$ and bounded from below by
$\omega(n \log_2(\rho))$.
\end{theorem}

This result was proven in~\cite{BuKno18,shivers02}.
Its proof, which we omit here, is very similar to our own analysis of \cASS in Section~\ref{sec:analysis} below.
A crucial element of both
proofs, as will be stated in Lemma~\ref{lem:invariant-li},
is the fact that the sequence $\ell_1,\ell_2,\ldots,\ell_{h-k}$
shall always be decreasing, for some small integer $k$:
we have $k = 1$ in the proof of~\cite{BuKno18}, and
$k = 2$ in Lemma~\ref{lem:invariant-li}.
In essence, this invariant is similar to that of \TS,
but it allows decreasing the associated constant hidden in the
$\O$ notation from $3/2$ to $1$.

The very idea of integrating features from these
two algorithms led Buss and Knop to invent the algorithm
\AugSS~\cite{BuKno18}. This algorithm is
obtained by using the following main loop:

\begin{algorithm}[ht]
\begin{small}
\SetArgSty{texttt}
\DontPrintSemicolon
\setcounter{AlgoLine}{2}
\While(\tcp*[f]{main loop}){\true}{
    \If
    {\textrm{$h \geqslant 3$, $r_h \geqslant r_{h-2}$ and $\ell_h \geqslant \ell_{h-1}$}}
    {merge the runs $R_{h-2}$ and $R_{h-1}$}
    \ElseIf
    {\textrm{$h \geqslant 2$ and $\ell_h \geqslant \ell_{h-1}$}}
    {merge the runs $R_{h-1}$ and $R_h$}
    \ElseIf
    {$\rundecomp \neq \emptyset$}
      {remove a run $R$ from $\rundecomp$ and
      push $R$ onto $\S$}
    \Else{break}
}
\end{small}
\end{algorithm}
\RestyleAlgo{boxruled}

The hope here is that both avoiding merging
newly pushed runs while they are too large and maintaining
an invariant (on the integers $\ell_i$) similar to that of
\ShS would make \AugSS very efficient.
Unfortunately, this algorithm suffers from the same design flaw
as the original version of \TS, and the desired invariant
is \emph{not} maintained. Even worse, the effects of not
maintaining this invariant are much more severe here, as
underlined by the following result.

\begin{theorem}\label{thm:complexity-AugSS}
The worst-case merge cost of \AugSS on inputs of length $n$ is $\Theta(n^2)$.
\end{theorem}

\begin{proof}
Consider some integer $k \geqslant 1$, and let
$n = 8k$ and $\rho = 2k$.
Let also $r_1,\ldots,r_\rho$ be the sequence of
run lengths defined by $r_{2i-1} = 6$ and
$r_{2i} = 2$ for all $i \leqslant k$.
Note that $r_1+\ldots+r_\rho = n$.

Now, let us apply the algorithm \AugSS on an array of $n$ integers
that splits into increasing runs $R_1,\ldots,R_\rho$ of lengths exactly
$r_1,\ldots,r_\rho$. One verifies quickly that the algorithm performs
successively the following operations:
\begin{itemize}
 \item first, push rhe runs $R_1$ and $R_2$;
 \item then, and for all $i \in \{2,\ldots,k\}$,
 push the run $R_{2i-1}$, then merge
 the runs $R_{2i-2}$ and $R_{2i-3}$,
 and push the run $R_{2i}$;
 \item finally, keep merging the last two runs on the stack
 (line~\ref{alg:collapse}): we first merge the runs $R_{\rho-1}$ and $R_\rho$,
 and then, the $(m+1)$\textsuperscript{th}
 such merge involves runs of sizes $8$ and $8 m$.
\end{itemize}
Therefore, the merge cost of \AugSS on that array is
\[\mc = \sum_{i=1}^k (r_{2i-1}+r_{2i}) + \sum_{m=1}^{k-1} 8 (m+1) =
  n^2/16 + 3 n / 2 - 8.\]

Conversely, in any (natural or not) merge sort, any element can be merged
at most $n-1$ times, and therefore the total merge cost of such a sorting
algorithm is at most $n(n-1)$.
\end{proof}

Finally, as mentioned in Section~\ref{sec:intro},
\PoS enjoys excellent complexity guarantees, with
a merge cost bounded from above by $n(\H+2)$.
However, it is not a $k$-aware algorithm for any $k$,
and its merge policy is more complicated than that of \TS.
Indeed, whether two consecutive runs $R$ and $R'$
should be merged does not depend directly on their lengths,
but on their \emph{powers}, which are defined as follows.

Let $R$ be a run resulting from the merge of several
(original) runs $R_i,\ldots,R_j$, and let
$\runend$ be the last position contained in the run $R$.
The power of $R$ is defined as the least integer $p$ such that
\[\lfloor 2^p (4 \, \runend + 1 - 2 r_j)/n \rfloor <
\lfloor 2^p (4 \, \runend + 1 + 2 r_{j+1})/n \rfloor.\]
Then, two consecutive runs $R$ and $R'$ should be merged if
$p > p'$. Hence, \PoS requires additional data structures
(e.g., a second stack storing the powers of the runs)
that are not needed in \TS nor in any $k$-aware algorithm.

In what follows, we focus on designing an algorithm that would
successfully integrate the features of both \TS and \ShS,
while keeping the simplicity of their merge policies,
thereby being $k$-aware for a small integer $k$.
This new algorithm is \cASS.

\section{Worst-case analysis of \cASS}\label{sec:analysis}

We stated, in the introduction, that \cASS
enjoys excellent worst-case upper bounds in terms of merge cost.
We prove that statement in this section,
whcih is subdivided in two independent parts.

Section~\ref{subsec:analysis-1} is devoted to proving Theorem~\ref{thm:complexity-nH+3},
which contains simple yet already excellent
upper bounds on the merge costs of \cASS and of \LASS.
Section~\ref{subsec:analysis-2} then consists in proving Theorem~\ref{thm:complexity-nH+D},
which contains an even better upper bound on the merge cost of \cASS.
However, the proof of this latter result is less intuitive and more technical than
that of Theorem~\ref{thm:complexity-nH+3}, which is why we decided to present it
only later in Section~\ref{sec:analysis}.

\subsection{A first upper bound}\label{subsec:analysis-1}

\begin{proposition}\label{pro:complexity-nH+3}
For every value of the parameter $c$, the merge cost of \ccASS is bounded from above by
$n (\H + 3 - \{\log_2(n/c)\}) - \rho-1$,
where $\{x\} = x - \lfloor x \rfloor$
denotes the fractional part of the real number $x$.
\end{proposition}

The upper bound provided in Proposition~\ref{pro:complexity-nH+3} is slightly complicated.
Hence, and in particular for 
$c = 1$ and $c = n+1$, it can readily be replaced by the following upper bounds,
which depends only on $n$ and $\H$.

\begin{theorem}\label{thm:complexity-nH+3}
The merge costs of \cASS and of \LASS are respectively
bounded from above by $n (\H + 3)$ and by $n(\H+2)$.
\end{theorem}

\begin{proof}
First, Proposition~\ref{pro:complexity-nH+3} states that \cASS has a merge cost
\[\mc \leqslant n (\H + 3 - \{\log_2(n)\}) - \rho-1 \leqslant n(\H+3).\]
It also states that \LASS has a merge cost
\[\mc \leqslant n (\H + 3 - \{\log_2(n/(n+1))\}) - \rho-1.\]
Observing that
$\{\log_2(n/(n+1))\} = \log_2(n/(n+1)) - \lfloor \log_2(n/(n+1)) \rfloor =
1 - \log_2(1+1/n) \geqslant 1 - 2/n$,
it follows that
$\mc \leqslant n (\H+2) - \rho + 1 \leqslant n(\H+2)$.
\end{proof}

In what follows, we fix the value of the parameter $c$ once and for all,
and we aim at proving Proposition~\ref{pro:complexity-nH+3}.

In order to do so, we first set up some notations.
Below, we denote by $r$ the length of a run $R$,
by $\ell$ the integer $\lfloor \log_2(r / c) \rfloor$
and by $\lambda$ the real number $\{\log_2(r / c)\} = \log_2(r / c) - \ell$.
The integer $\ell$ will be called the 
\emph{level} of the run $R$.\label{def:level}
We adapt readily these notations when the name of the run considered varies, e.g.,
we denote by $r'$ the length of the run $R'$, 
by $\ell'$ the integer $\lfloor \log_2(r' / c) \rfloor$ and
by $\lambda'$ the real number $\{\log_2(r' / c)\}$.
In particular, we will commonly denote the stack by
$(R_1,\ldots,R_h)$, where $R_k$ is the
$k$\textsuperscript{th} bottom-most run of the stack --
therefore, accessing the element $R_h$ in the stack is easy, and accessing $R_1$
is much less straightforward.
The length of $R_k$ is then denoted by $r_k$, and we set
$\ell_k = \lfloor \log_2(r_k / c) \rfloor$.

With this notation in mind, we first prove two auxiliary results
about the levels of the runs manipulated throughout the algorithm.

\begin{lemma}\label{lem:l-small-increase}
When two runs $R$ and $R'$ are merged into a single run $R''$, we have $\ell'' \leqslant \max\{\ell,\ell'\}+1$.
\end{lemma}

\begin{proof}
Without loss of generality, we assume that $r \leqslant r'$. In that case, it comes that
\[
2^{\ell''} c \leqslant r'' = r + r' \leqslant 2 r' < 2 \times 2^{\ell'+1} c = 2^{\ell'+2} c,
\]
and therefore that $\ell'' \leqslant \ell'+1$.
\end{proof}

\begin{lemma}\label{lem:invariant-li}
At any time during the main loop of the algorithm \ccASS, if
the stack of runs is $\S = (R_1,\ldots,R_h)$, we have:
\begin{equation}
\ell_1 > \ell_2 > \ldots > \ell_{h-3} > \max\{\ell_{h-2},\ell_{h-1}\}.
\label{ell-incr}
\end{equation}
If, furthermore, $\S$ results from a merge between two runs, then 
$\ell_{h-2} \geqslant \ell_{h-1}$.
\end{lemma}

\begin{proof}
The proof is done by induction.
First, if $h \leqslant 2$, there is nothing to prove:
this case occurs, in particular, when the algorithm starts.
Now, consider some stack $\S = (R_1,\ldots,R_h)$ that
satisfies~\eqref{ell-incr} and is updated into a new stack
$\oS = (\oR_1,\ldots,\oR_\oh)$, either
by merging the runs $R_{h-2}$ and $R_{h-1}$,
or by pushing the run $\oR_\oh$:
\begin{itemize}
 \item If the runs $R_{h-2}$ and $R_{h-1}$ were just merged, then
 $\oh = h-1$ and $\oR_i = R_i$ for all $i \leqslant h-3$.
 Thus, the inequalities $\ell_1 > \ell_2 > \ldots > \ell_{h-3}$
 immediately rewrite as
 $\ol_1 > \ol_2 > \ldots > \ol_{\oh-2}$.
 Meanwhile, $\oR_{\oh-1}$ results from the merge between
 $R_{h-2}$ and $R_{h-1}$, and therefore
 Lemma~\ref{lem:l-small-increase} proves that
 $\ol_{\oh-1} \leqslant \max\{\ell_{h-2},\ell_{h-1}\}+1$.
 This means that $\ol_{\oh-2} > \ol_{\oh-1}-1$ or, equivalently, that
 $\ol_{\oh-2} \geqslant \ol_{\oh-1}$.
 
 \item If the run $\oR_1$ was just pushed,
 then $\oh = h+1$ and $\oR_i = R_i$ for all $i \leqslant h$.
 Thus, the inequalities $\ell_1 > \ell_2 > \ldots > \ell_{h-2}$
 already rewrite as $\ol_1 > \ol_2 > \ldots > \ol_{\oh-3}$.
 Furthermore, since \ccASS triggered a push operation instead of a merge operation,
 it must be the case that $\ol_{\oh-3} = \ell_{h-2} > \max\{\ell_{h-1},\ell_h\} = \max\{\ol_{\oh-2},\ol_{\oh-1}\}$.
\end{itemize}
In both cases, it follows that $\oS$ also satisfies~\eqref{ell-incr},
which completes the induction.
\end{proof}

Roughly speaking, Lemma~\ref{lem:invariant-li} states that the lengths of the
runs stored in the stack increase at exponential speed
(when we start from the top of the stack),
with the possible exception of the top-most run,
whose length we have no control on.
As suggested in Section~\ref{sec:description}, 
this property was already crucial in the complexity proofs of
several algorithms such as~\ShS, \TS or \aMS.

In addition to these results, we will also need the following technical lemma.

\begin{lemma}\label{lem:2pow-x}
For all real numbers $x$ such that $0 \leqslant x \leqslant 1$,
we have $2^{1-x} \leqslant 2 - x$.
\end{lemma}

\begin{proof}
Every function of the form $x \mapsto \exp(t x)$, where $t$ is a fixed real parameter,
is convex. Therefore, the function $f : x \mapsto 2^{1-x} - (2-x)$ is convex too.
It follows, for all $x \in [0,1]$, that $f(x) \leqslant \max\{f(0),f(1)\} = 0$, which completes the proof.
\end{proof}

The proof of Proposition~\ref{pro:complexity-nH+3} consists now in a careful
estimation of the total cost of those merges performed by the algorithm.
Intuitively, this proof may be seen as a cost allocation, where each merge
between two runs $R$ and $R'$ should be paid for by some run
(which may be either $R$, $R'$, or some other run).
In practice, however, assigning the entire cost of a merge to a single run
may be too crude for our needs. Therefore, it will be convenient 
to split the cost of a merge in two parts that will be paid for by different runs.

Hence, below, we artificially split the merge between $R$ and $R'$
into two separate merge operations:
one part, for a cost of $r$, is called the \emph{merge of $R$ with $R'$},
and the other part, for a cost of $r'$, is called the
\emph{merge of $R'$ with $R$}.
Together, these operations indeed consist in merging $R$ and $R'$ with each
other. Their costs add up to $r + r'$, as expected and,
in what follows, they may be allocated to distinct cost centers.

Then, when merging the run $R$ with a run $R'$
into one bigger run $R''$, we say that
the merge of $R$ is \emph{expanding} if $\ell'' \geqslant \ell + 1$,
and is \emph{non-expanding} otherwise.
Note that, if $\ell \leqslant \ell'$, the merge of $R$ with $R'$ is necessarily expanding.
Consequently, when two runs $R$ and $R'$ are merged with each other, 
either the merge of $R$ or of $R'$ is expanding. In particular, if $\ell = \ell'$, then both merges of $R$ and of $R'$ must be expanding.
Hence, we say that the merge between $R$ and $R'$ is
\emph{intrinsically expanding} if $\ell = \ell'$.

We first show that, up to a linear term,
the announced merge cost is entirely due to expanding merges.

\begin{lemma}\label{lem:expanding-nH}
The total cost of expanding merges is at most $n (\H - \{\log_2(n/c)\}) + \Lambda$,
where $\Lambda$ is defined as $\Lambda = \sum_{i=1}^\rho r_i \lambda_i$.
\end{lemma}

\begin{proof}
While the algorithm is performed, the elements of a run $R$ of initial length $r$ may take part in at most
\[\lfloor \log_2(n/c) \rfloor - \ell =
(\log_2(n/c) - \{\log_2(n/c)\}) - (\log_2(r/c) - \lambda)
= \log_2(n/r) + \lambda - \{\log_2(n/c)\}\]
expanding merges.
Consequently, if the array is initially split into runs of lengths $r_1,\ldots,r_\rho$,
the total cost of expanding merges is at most
\[\sum_{i=1}^\rho r_i (\log_2(n/r_i) + \lambda_i - \{\log_2(n/c)\}) =
n (\H - \{\log_2(n/c)\}) + \Lambda.
\]
\end{proof}

It remains to prove that the total cost of
non-expanding merges is at most $3n - \Lambda - \rho - 1$.
This requires classifying merges based on the conditions that triggered them,
and \emph{binding} some merges to a single run, as follows.

\begin{definition}
Let $\S = (R_1,\ldots,R_h)$ be a stack of runs
such that $h \geqslant 3$ and $\ell_{h-2} \leqslant \max\{\ell_{h-1},\ell_h\}$.
By construction, upon encountering the stack $\S$ during its main loop,
the algorithm \ccASS performs a merge operation $m$ between the runs $R_{h-2}$ and $R_{h-1}$.
We say that $m$ is
(i)~a \#1-\emph{merge} if $\ell_{h-2} < \ell_{h-1}$;
(ii)~a \#2-\emph{merge} if ${\ell_{h-1} \leqslant \ell_{h-2} \leqslant \ell_h}$;
(iii)~a \#3-\emph{merge} if $\ell_h < \ell_{h-2} = \ell_{h-1}$.
Furthemore, (i)~if $m$ is a \#1-merge, we \emph{bind} it to the run $R_{h-1}$;
(ii)~if $m$ is a \#2-merge, we bind it to the run $R_h$;
(iii)~if $m$ is a \#3-merge, we do not bind it to any run.
\end{definition}

Observe that the conditions for being a \#1-merge, a \#2-merge or a \#3-merge
are mutually exclusive, and that every merge belongs to one of these three classes.
By extension, we may also refer to a \#4-push operation, so that each update
is a $\#k$-update for some $k \leqslant 4$.

Moreover, our decision to bind some merges to a given run is motivated as follows.
By construction, every \#3-merge is intrinsically expanding.
Then, if a non-expanding merge $m$ is bound to a run $R$,
we choose to allocate the cost of $m$ to $R$, which,
as we prove below, will necessarily be a run from the original array
(i.e., $R$ has not yet been merged before $m$ occurs).

\begin{lemma}\label{lem:invariant-mid-end}
No merge is immediately followed by a \#1-merge,
and no \#3-merge is immediately followed by a \#2-merge.
\end{lemma}

\begin{proof}
Let $m$ be a merge. We denote by
$\S = (R_1,\ldots,R_h)$ the stack before the merge
and by $\oS = (\oR_1,\ldots,\oR_\oh)$
the stack after the merge, so that $\oh = h-1$.
If $\oh \leqslant 2$, then $m$ cannot be followed by any merge,
hence we assume that $\oh \geqslant 3$.

First, Lemma~\ref{lem:invariant-li} already states that $\ol_{\oh-2} \geqslant \ol_{\oh-1}$,
and therefore $m$ cannot be followed by a \#1-merge.

Second, if $m$ is a \#3-merge, we must have $\ell_{h-1} = \ell_{h-2} > \ell_h$.
Since $\oR_\oh = R_h$ and $\oR_{\oh-1}$ results from merging 
the runs $R_{h-2}$ and $R_{h-1}$, it follows that
$\ol_{\oh-1} \geqslant \ell_{h-1} > \ell_h = \ol_\oh$,
which shows that $m$ cannot be followed by a \#2-merge.
\end{proof}

\begin{lemma}\label{lem:invariant-end-start}
If a push is immediately preceded by a \#3-merge,
it cannot be immediately followed by a \#1-merge.
\end{lemma}

\begin{proof}
Let $p$ be a push update preceded by a \#3-merge $m$. We denote by
$\S = (R_1,\ldots,R_h)$ the stack before $m$ occurs
and by $\oS = (\oR_1,\ldots,\oR_\oh)$
the stack after $p$ occurs, so that $\oh = h$.

By construction, the run $\oR_{\oh-2}$ results from merging the runs $R_{h-2}$ and $R_{h-1}$,
and $\oR_{\oh-1} = R_h$, so that $\ell_{h-1} \leqslant \ol_{\oh-2}$ and $\ell_h = \ol_{\oh-1}$. Since $m$ is a \#3-merge, we conclude that
$\ol_{\oh-1} = \ell_h < \ell_{h-1} \leqslant \ol_{\oh-2}$, i.e., that $p$ cannot be
immediately followed by a \#1-merge.
\end{proof}

A consequence of Lemma~\ref{lem:invariant-mid-end} is the following one.
Let $R$ be some run in the array to be sorted.
Just after $R$ has been pushed, there will be
(i)~first, zero or one \#1-merge (which is \emph{not} bound to $R$);
(ii)~then, an arbitrary number of \#2-merges (which \emph{are} bound to $R$);
(iii)~then, an arbitary number of \#3-merges;
(iv)~then, a new run $\oR$ is pushed onto the stack;
(v)~then, zero or \#1-merge (which \emph{is} bound to $R$);
and finally, other merge or push updates, none of these merges being bound to $R$.
Furthermore, if there is at least one \#3-merge at step~(iii),
then there is no \#1-merge at step~(v).
Therefore, the merges bound to $R$ form a contiguous sequence of merge updates
which we call \emph{merge sequence} of $R$.
This situation is illustrated in Figure~\ref{fig:subsequences}.

\begin{figure}[t!]
\centerline{Case~1:\hspace*{5mm}
\ub{\#4}{pushing}{run $R$}$\vv$
\ub{\#1?}{optional}{\#1-merge}$\vv$
$\overbrace{\text{\ub{\#2,\#2,\ldots,\#2}{\#2-merges}{bound to $R$}}}^%
{\text{merge sequence of $R$}}$%
$\vv$
\ub{\#3,\#3,\ldots,\#3}{\#3-merges}{}$\vv$
\ub{\#4}{pushing}{run $\oR$}$\vv$
\ub{\phantom{\#1?}}{no}{\#1-merge}$\vv\ldots$
}

\vspace*{6mm}

\centerline{Case~2:\hspace*{5mm}
\ub{\#4}{pushing}{run $R$}$\vv$
\ub{\#1?}{optional}{\#1-merge}$\vv$
$\overbrace{\text{\ub{\#2,\#2,\ldots,\#2}{\#2-merges}{bound to $R$}$\vv$
\ub{\phantom{\#3,\#3,\ldots,\#3}}{no \#3-merge}{}$\vv$
\ub{\#4}{pushing}{run $\oR$}$\vv$
\ub{\#1\phantom{?}}{\#1-merge}{bound to $R$}}}^%
{\text{merge sequence of $R$}}$%
$\vv\ldots$
}
\caption{The run $R$ is pushed. Then come one (optional) \#1-merge,
a (possibly empty) sequence of \#2-merges bound to $R$,
and a (possibly empty) sequence of \#3-merges.
Finally, a new run $\oR$ is pushed, which may be
followed by one (optional) \#1-merge bound to $R$. \newline
In case~1, the latter \#1-merge does not exist, and \#2-merges bound to $R$ form
the merge sequence of $R$; in case~2, the latter \#1-merge exists,
the sequence of \#3-merges must be empty, and the \#2-merges and \#1-merge
bound to $R$ form the merge sequence of $R$.
\label{fig:subsequences}}
\end{figure}

\begin{lemma}\label{lem:non-expanding-run-2r}
The total cost of non-expanding merges allocated to a run $R$
is at most $(2 - \lambda) r - 1$.
\end{lemma}

\begin{proof}
Let $\Ms$ be the merge sequence of $R$,
and let $\mc$ be the total cost of the non-expanding merges that belong to $\Ms$,
i.e., the total cost of non-expanding merges allocated to $R$.
We assume below that $\Ms$ contains at least one non-expanding merge,
unless what $\mc = 0 \leqslant r - 1 \leqslant (2 - \lambda) r - 1$ already.

Now, let $\S = (R_1,\ldots,R_h)$ be the stack just after $R$ has been
pushed onto the stack or, if a \#1-merge immediately follows the push of $R$,
just after that \#1-merge has been performed.
We know that $R = R_h$.
Moreover, due to Lemma~\ref{lem:invariant-li},
and since the first update performed on $\S$ cannot be a \#1-merge,
we have
\[\ell_1 > \ldots > \ell_{h-2} \geqslant \ell_{h-1}.\]
Thus, $\Ms$ consists in merging $R_{h-2}$ and $R_{h-1}$,
then merging successively the resulting run with $R_{h-3}, R_{h-4}, \ldots, R_k$
for some $k$ (we set $k = h-1$ if $\Ms$ contains no \#2-merge nor \#3-merge),
then possibly with $R = R_h$ itself.

Now, let $m$ be some non-expanding merge that belongs to $\Ms$.
If $m$ is a \#1-merge, then it must be the merge of $R = R_h$ with a smaller run,
and its cost is $r = r_h$;
if $m$ is a \#2-merge, then it must be the merge of some run $R_i$ 
(with $k \leqslant i \leqslant h-2$) with a smaller run, and its cost is $r_i$.
This proves that $\mc = \sum_{i \in X} r_i$,
where the set $X$ is defined by
\[X = \{i \colon 1 \leqslant i \leqslant h \text{ and the merge of $R_i$ is non-expanding and belongs to $\Ms$}\}.\]

Then, let $m^\ast$ be the last non-expanding merge that belongs to $\Ms$,
and let $R^\ast$ be the resulting run.
Just after $m^\ast$ is completed,
all the runs that had been merged since $\Ms$ started must now belong to the run $R^\ast$. 
These runs must include the run $R_{h-1}$,
which was the first run merged during $\Ms$ and whose merge must have been expanding.
It follows that $\sum_{i \in X} r_i \leqslant r^\ast - r_{h-1} \leqslant r^\ast - 1$.

Finally, and as mentioned before,
$m^\ast$ must be the non-expanding merge of some run $R_j$:
we have $j = h$ if $m^\ast$ is a \#1-merge, and $k \leqslant j \leqslant h-2$
if $m^\ast$ is a \#2-merge.
In both cases, it comes that $\ell_j \leqslant \ell$.
Then, since $m^\ast$ is non-expanding, it further comes that $\ell^\ast \leqslant \ell_j\leqslant \ell$.

Using Lemma~\ref{lem:2pow-x}, we conclude that
\[\mc = \sum_{i \in X} r_i \leqslant r^\ast - 1 \leqslant 2^{\ell^\ast+1} c - 1
\leqslant 2^{\ell+1} c - 1 = 2^{1-\lambda} r - 1 \leqslant (2-\lambda)r - 1.\]
\end{proof}

A similar result also holds for those merges performed after the main loop.

\begin{lemma}\label{lem:non-expanding-collapse-n}
The total cost of non-expanding merges performed in
line~\ref{alg:collapse} of Algorithm~\ref{alg:ASS} is at most $n - 1$.
\end{lemma}

\begin{proof}
Let $\S = (R_1,\ldots,R_h)$ be the stack at the end of the main loop.
Due to Lemma~\ref{lem:invariant-li} and to the fact that
no merge was triggered, we know that
\[\ell_1 > \ldots > \ell_{h-2} > \max\{\ell_{h-1},\ell_h\}.\]
Consequently, any non-expanding merge that takes place in line~\ref{alg:collapse}
must be the merge of some run $R_i$ with a smaller run.
Thus, denoting by $X$ the set
\[\{i \colon 1 \leqslant i \leqslant h
\text{ and the merge of $R_i$ is non-expanding}\},\]
the total cost of these merges is
$\mc = \sum_{i \in X} r_i$.
Since $X$ is a strict subset of $\{1,\ldots,h\}$, it follows that
$\mc \leqslant \sum_{i=1}^h r_i - 1 = n-1$.
\end{proof}

We conclude this section by gathering these results as follows.

\begin{proof}[Proof of Proposition~\ref{pro:complexity-nH+3}]
Lemma~\ref{lem:expanding-nH} states that the total cost of expanding merges
is at most ${n (\H - \{\log_2(n/c)\}) + \Lambda}$.
Then, Lemma~\ref{lem:non-expanding-run-2r} states that the total cost
of those non-expanding merges allocated to a given run $R$ is at most $(2 - \lambda) r - 1$.
Taking all runs into account, the total cost of those non-expanding merges
performed during the main loop of \ccASS is at most
$2n - \Lambda - \rho$.
Finally, Lemma~\ref{lem:non-expanding-collapse-n} states that the total cost
of those non-expanding merges performed in line~\ref{alg:collapse} is at most $n-1$.
Summing all these costs completes the proof of Proposition~\ref{pro:complexity-nH+3}.
\end{proof}

\subsection{A finer upper bound}\label{subsec:analysis-2}

Now that Theorem~\ref{thm:complexity-nH+3} has been proven,
let us present a finer upper bound on the merge cost of \ccASS.

\begin{theorem}\label{thm:complexity-nH+D}
The merge cost of \ccASS is bounded from above by $n \left(\H + \Delta\right)$, 
where $\Delta = 24/5 - \log_2(5) \approx 2.478$.
\end{theorem}

The proof we draw below relies mainly
on the ideas and results already presented
in Section~\ref{subsec:analysis-1}.
However, we cannot reuse directly our cost allocation scheme,
and we will need a notion of potential instead.
As a first step towards defining our potential,
we first introduce the notion of \emph{state} of the algorithm.

\begin{restatable}{definition}{defstate}
\label{def:state}
Consider the execution of the algorithm \ccASS
on a sequence of runs to be sorted.
At any step of the algorithm, the algorithm handles both
a stack $\S = (R_1,\ldots,R_h)$ of runs
and a sequence $\R = (R_{h+1},\ldots,R_\rho)$
of those runs that have yet to be discovered and pushed onto the stack.
We call \emph{state} of the algorithm, at that step, the
sequence $(R_1,\ldots,R_\rho)$, i.e., the concatenation of
$\S$ and $\R$.

Finally, two states of the algorithm are said to be \emph{consecutive}
if they are distinct from each other and
were separated by a single (run push or merge) operation performed
by the algorithm.
\end{restatable}

We immediately see that a push operation does not modify the state
of the algorithm. Thus, the operation that separates two consecutive states
of the algorithm is necessarily a run merge operation.
Moreover, and in order to reduce possible confusions between
sequences of runs of different natures, 
we will stick to the notation $\S$ for stacks
and $\bS$ for states of the algorithm.

We focus now on describing how \ccASS transforms
a state into another one.

\begin{restatable}{definition}{defsucc}
\label{def:successor}
Let $\R = (R_1,\ldots,R_\rho)$ be a sequence of runs
of length $\rho \geqslant 2$.
We say that an integer $x$ is \emph{dominated} in the sequence $\R$
if $1 \leqslant x \leqslant \rho-1$ and $\ell_x \leqslant \max\{\ell_{x+1},\ell_{x+2}\}$,
with the convention that $\ell_{\rho+1} = \infty$
(we may omit mentioning $\R$ when the context is clear).
This convention ensures us that $\rho-1$ is necessarily dominated.

Then, let $k$ be the smallest dominated integer.
We say that $k$ is the \emph{merge point} of the sequence $\R$.
Finally, we call \emph{successor} of $\R$, and denote by
$\succc(\mathbf{R})$, the sequence of runs
\[(R_1,\ldots,R_{k-1},\oR,R_{k+2},\ldots,R_\rho),\]
where $\oR$ is the run obtained by merging $R_k$ and $R_{k+1}$.
\end{restatable}

\begin{restatable}{proposition}{promerge}
\label{pro:first-merge}
Let $\bS$ and $\obS$ be two consecutive states encountered
during an execution of \ccASS.
Then, $\obS = \succc(\bS)$.
\end{restatable}

\begin{proof}
Let $m$ be the merge operation that transforms the state $\bS$ into $\obS$.
Let $\S = (R_1,\ldots,R_{k+2})$ be the stack just before $m$ takes place,
and let $\R = (R_{k+3},\ldots,R_\rho)$
be the sequence of those runs that are yet to be pushed onto the stack,
so that $m$ consists in merging the runs $R_k$ and $R_{k+1}$, and
that $\bS$ is the concatenation of $\S$ and $\R$.

Lemma~\ref{lem:invariant-li}
states that $\ell_1 > \ell_2 > \ldots > \ell_{k-1} > \max\{\ell_k,\ell_{k+1}\}$,
and since \ccASS performed the merge $m$, it means that
$\ell_k \leqslant \max\{\ell_{k+1},\ell_{k+2}\}$.
This implies that $k$ is the merge point of $\bS$, and therefore that
$\obS = \succc(\bS)$.
\end{proof}

In addition to the notion of state and to the notations defined in
Section~\ref{def:level} (page~\pageref{def:level}), we will frequently use
the following notation.
For every run $R$ of length $r$ and level $\ell$,
we set $r_\bullet = r / (2^\ell \, c) - 1$.
Note that, as $R$ varies over the interval $[2^\ell \, c, 2^{\ell+1} \, c)$,
the quantity $r^\bullet$ varies over the interval $[0,1)$.
Once again, we will adapt this notation when the name of $R$ varies,
e.g., writing $r_i^\bullet$ when considering the run $R_i$.

We introduce now the notions of potential that we will use in the subsequent
proofs.

\begin{definition}\label{def:potential}
Let $\Phi : [0,1] \mapsto \mathbb{R}$ be the function defined by $\Phi : x \mapsto \max\{(2-5x)/3,1/2-x,0\}$.
Then, let $\bS = (R_1,\ldots,R_\rho)$ be a state of the algorithm.
By convention, let $\ell_0 = \ell_{\rho+1} = +\infty$.
We define the \emph{potential} of a run $R_i$ in $\bS$
as the real number $\POT{\bS}{0}(R_i) = \sum_{j=1}^4 \POT{\bS}{j}(R_i)$,
where $\POT{\bS}{1}(R_i) = - \ell_i r_i$ and
\begin{align*}
\POT{\bS}{2}(R_i)
&
= \begin{cases}
\rrlap{- r_i}{2^{\ell_i} c \left(2 \Phi\!\left((r_i^\bullet + r_{i+1}^\bullet)/2\right) - \Phi(r_i^\bullet) - \Phi(r_{i+1}^\bullet)\right)} &
\text{if $\ell_{i-1} \geqslant \ell_i$ and $\ell_i < \ell_{i+1}$;} \\
0 & \text{if $\ell_{i-1} < \ell_i$ or $\ell_i \geqslant \ell_{i+1}$;}
\end{cases} \\
\POT{\bS}{3}(R_i)
&
= \begin{cases}
\rrlap{2^{\ell_i+1} c}{2^{\ell_i} c \left(2 \Phi\!\left((r_i^\bullet + r_{i+1}^\bullet)/2\right) - \Phi(r_i^\bullet) - \Phi(r_{i+1}^\bullet)\right)} & 
\text{if $\ell_{i-1} < \ell_i$;} \\
2^{\ell_i} c \, \Phi(r_i^\bullet) & \text{if $\ell_{i-1} \geqslant \ell_i$;}
\end{cases} \\
\POT{\bS}{4}(R_i)
&
= \begin{cases}
2^{\ell_i} c \left(2 \Phi\!\left((r_i^\bullet + r_{i+1}^\bullet)/2\right) - \Phi(r_i^\bullet) - \Phi(r_{i+1}^\bullet)\right) & \text{if $\ell_{i-1} > \ell_i$ and $\ell_i = \ell_{i+1}$;} \\
0 & \text{if $\ell_{i-1} \geqslant \ell_i$ or $\ell_i \neq \ell_{i+1}$.}
\end{cases} \\
\end{align*}

Finally, we call \emph{global potential} of the state $\bS$ the sum
$\Pot(\bS) = \sum_{i=1}^\rho \POT{\bS}{0}(R_i)$, i.e.,
the sum of the potentials of all the runs $R_1$ to $R_\rho$.
\end{definition}

Below, and in order to make a good use of the otherwise mysterious
function $\Phi$, we will need the following technical lemma.

\begin{lemma}\label{lem:log+phi}
The function $\Phi$ is convex and non-increasing.
Moreover, for all real numbers $x$ such that $0 \leqslant x \leqslant 1$, we have
\[1 \geqslant x + 2 \Phi(x/2) - \Phi(x) \geqslant
x + \Phi(x) \text{ and }
\log_2(1+x) + \Phi(x)/(1+x) \geqslant 3 - \Delta.\]
\end{lemma}

\begin{proof}
First, since $\Phi$ is a maximum of non-increasing affine functions,
it is convex and non-increasing.
It already follows that
$x + 2 \Phi(x/2) - \Phi(x) = \left(x + \Phi(x)\right) + 
2\left(\Phi(x/2) - \Phi(x)\right) \geqslant x + \Phi(x)$
for all $x \in [0,1]$.
It remains to prove that both functions
$f : x \mapsto \log_2(1+x) + \Phi(x)/(1+x) + \Delta - 3$ and
$g : x \mapsto 1 - x - 2 \Phi(x/2) + \Phi(x)$
are non-negative.

First, the function $\Phi$ is affine on each of the intervals
$[0,1/4]$, $[1/4,1/2]$ and $[1/2,1]$, and thus
we study $f$ on each of these intervals:
\begin{itemize}
 \item If $0 \leqslant x \leqslant 1/4$, then $\Phi(x) = (2-5x)/3$, and thus
 \[(1+x)^2 f'(x) = (1+x)\log_2(e)-7/3 \leqslant
 5/4 \log_2(e) - 7/3 \approx -0.5,\]
 which proves that $f(x) \geqslant f(1/4) = 0$.
 
 \item If $1/4 \leqslant x \leqslant 1/2$, then $\Phi(x) = 1/2-x$, and thus
 \[(1+x)^2 f'(x) = (1+x)\log_2(e)-3/2 \geqslant (1+x)\log_2(e)-3/2 \approx 0.3,\]
 which also proves that $f(x) \geqslant f(1/4) = 0$.
 
 \item If $1/2 \leqslant x \leqslant 1$, then $\Phi(x) = 0$, and thus
 \[(1+x) f'(x) = \log_2(e) > 0,\]
 which proves that
 $f(x) \geqslant f(1/2) \geqslant f(1/4) \geqslant 0$.
\end{itemize}

Second, since $\Phi$ is affine on each of the intervals
$[0,1/4]$, $[1/4,1/2]$ and $[1/2,1]$, so is $g$.
Since $g(0) = 1/3$, $g(1/4) = 1/12$ and $g(1/2) = g(1) = 0$,
we conclude that $g$ is non-negative on each of these intervals.
\end{proof}

Then, we prove that the variation of global potential
between two consecutive states separated by a merge operation $m$
is a good over-approximation of the cost of $m$.
This is the object of the two following results:
the first one proves that the variation of global potential
can be under-approximated by a \emph{local} quantity,
and the second result then proves that this is enough.

\begin{lemma}\label{lem:potential-1}
Let $\bS = (R_1,\ldots,R_\rho)$ and $\obS = (\oR_1,\ldots,\oR_{\rho-1})$
be two consecutive states of \ccASS, and let $k$
be the merge point of $\bS$.
If $k \leqslant \rho-1$, then
\[\Pot(\bS) \geqslant \Pot(\obS) + \POT{\bS}{0}(R_k) 
+ \POT{\bS}{0}(R_{k+1}) + \POT{\bS}{3}(R_{k+2}) - \POT{\obS}{0}(\oR_k) - 
\POT{\obS}{3}(\oR_{k+1}).\]
\end{lemma}

\begin{proof}
By construction, we know that $R_i = \oR_i$ for all $i \leqslant k-1$,
that $R_i = \oR_{i-1}$ for all $i \geqslant k+2$, and that the run
$\oR_k$ results from merging the runs $R_k$ and $R_{k+1}$.
It already
follows that $\POT{\bS}{0}(R_i) = \POT{\bS}{0}(\oR_i)$ for all $i \leqslant k-2$
and that $\POT{\bS}{0}(R_i) = \POT{\bS}{0}(\oR_{i-1})$ for all $i \geqslant k+3$.

Then, we also have $\POT{\bS}{j}(R_{k-1}) \geqslant \POT{\obS}{j}(\oR_{k-1})$
for all $j = 1,2,3,4$: both sides are equal to $-\ell_{k-1}$, $0$ and $2^{\ell_{k-1}} c \, \Phi(r_{k-1}^\bullet)$ when $j = 1,2,3$ respectively;
moreover, $\POT{\bS}{4}(R_{k-1}) = 0 \geqslant \POT{\obS}{4}(\oR_{k-1})$,
because $\Phi$ is convex. It follows that
$\POT{\bS}{0}(R_{k-1}) \geqslant \POT{\obS}{0}(\oR_{k-1})$.

Similarly, we have $\POT{\bS}{j}(R_{k+2}) \geqslant \POT{\obS}{j}(\oR_{k+1})$
for all $j = 1,2,4$. Indeed, both sides are equal to 
$- \ell_{k+2} r_{k+2}$ when $j = 1$;
moreover, $\POT{\bS}{2}(R_{k+2}) = 0 \geqslant \POT{\bS}{2}(\oR_{k+1})$; finally, if
$\POT{\obS}{4}(\oR_{k+1}) = 0$, then $\POT{\bS}{4}(R_{k+2}) = 0$ as well,
and thus $\POT{\bS}{4}(R_{k+2}) \geqslant \POT{\obS}{4}(\oR_{k+1})$.

Consequently, we conclude that
\begin{align*}
\Pot(\bS) - \Pot(\obS)
& = {\textstyle\sum_{i=1}^\rho} \POT{\bS}{0}(R_i) -
{\textstyle\sum_{i=1}^{\rho-1}} \POT{\obS}{0}(\oR_i) \\
& = {\textstyle\sum_{i=k-1}^{k+2}} \POT{\bS}{0}(R_i) -
{\textstyle\sum_{i=k-1}^{k+1}} \POT{\obS}{0}(\oR_i) \\
& \geqslant \POT{\bS}{0}(R_k) + \POT{\bS}{0}(R_{k+1}) + \POT{\bS}{3}(R_{k+2}) -
\POT{\obS}{0}(\oR_k) - \POT{\obS}{3}(\oR_{k+1}).
\end{align*}
\end{proof}

\begin{proposition}\label{pro:potential-2}
Let $m$ be a merge operation of cost $\mc$, and let
$\bS$ and $\obS$ be the states that $m$ separates.
It holds that $\Pot(\bS) \geqslant \Pot(\obS) + \mc$.
\end{proposition}

\begin{proof}
Let $\bS = (R_1,\ldots,R_\rho)$, and $\obS = (\oR_1,\ldots,\oR_{\rho-1})$,
and let $k$ be the merge point of $\bS$.
If $k = \rho-1$, let us append to both states a (fictitious) run $R_{\infty}$
of length $2n$, where $n = r_1+\ldots+r_\rho$.
Indeed, this does not change
the fact that $\obS = \succc(\bS)$, and only adds $\POT{\bS}{0}(R_\infty) =
\POT{\obS}{0}(R_\infty)$ to the global potentials of both states.
Then, in view of Lemma~\ref{lem:potential-1},
it suffices to prove that
$\Theta \geqslant 0$, where
\[\Theta = \POT{\bS}{0}(R_k) 
+ \POT{\bS}{0}(R_{k+1}) + \POT{\bS}{3}(R_{k+2}) - \POT{\obS}{0}(\oR_k) - 
\POT{\obS}{3}(\oR_{k+1}) - \mc.\]

Let also $r_{\max} = \max\{r_k,r_{k+1}\}$ and $r_{\min} = \min\{r_k,r_{k+1}\}$,
so that $\ell_{\max} = \max\{\ell_k,\ell_{k+1}\}$ and
$\ell_{\min} = \min\{\ell_k,\ell_{k+1}\}$.
By abuse of notation, we may denote by $R_{\max}$ the run $R_k$ if $\ell_k = \ell_{\max}$, and the run $R_{k+1}$ otherwise; the other run will be denoted by
$R_{\min}$.
With these notations, it comes that $\ell_{\max} \leqslant \ol_k \leqslant \ell_{\max}+1$ and that
$\ell_{\min} + 1 \leqslant \ol_k$.

Then, since $\ell_{k-1} > \max\{\ell_k,\ell_{k+1}\}$ and
$\ell_k \leqslant \max\{\ell_{k+1},\ell_{k+2}\}$,
we know that $\POT{\bS}{2}(R_{\max}) = 0$,
that $\POT{\bS}{4}(R_{k+1}) = 0$ and that
$\POT{\obS}{3}(\oR_k) = 2^{\ol_k}\Phi(\orr_k^\bullet)$.
Thus, we can write
$\Theta$ as the following sum:
\begin{align*}
\Theta & = (\ol_k-\ell_k - 1) r_k + (\ol_k-\ell_{k+1} - 1) r_{k+1} +
\POT{\bS}{2}(R_{\min}) + 
\POT{\bS}{3}(R_k) + \POT{\bS}{4}(R_k) +
\POT{\bS}{3}(R_{k+1}) \\
& \color{white}= (\ol_k-\ell_{\max} - 1) r_{\max}\color{black}
+ \POT{\bS}{3}(R_{k+2}) - \POT{\obS}{2}(\oR_k)
- 2^{\ol_k} c \, \Phi(\orr_k^\bullet) - \POT{\obS}{4}(\oR_k) - \POT{\obS}{3}(\oR_{k+1}).
\end{align*}
Finally, below, we will frequently use the fact that
\begin{align*}
\POT{\bS}{3}(R_k) + \POT{\bS}{4}(R_k) +
\POT{\bS}{3}(R_{k+1})
& = \rrlap{\POT{\bS}{3}(R_k) + \POT{\bS}{3}(R_{k+1})}%
{2^{\ell_k+1} c \, \Phi\left((r_k^\bullet+r_{k+1}^\bullet)/2\right)
= 2^{\ol_k} c \, \Phi(\orr_k^\bullet)} \geqslant 0 \text{ if } \ell_k \neq \ell_{k+1}; \\
& = 2^{\ell_k+1} c \, \Phi\left((r_k^\bullet+r_{k+1}^\bullet)/2\right)
= 2^{\ol_k} c \, \Phi(\orr_k^\bullet) \geqslant 0
\text{ if } \ell_k = \ell_{k+1}.
\end{align*}
With these inequalities in mind, we perform the following (rather long) disjunction of cases:

\begin{enumerate}
\item If $\ol_k < \ell_{k+2}$ and $\ell_k = \ell_{k+1}$, then
$\ol_k = \ell_{\max}+1$ and thus $\Theta = r_{\max} \geqslant 0$.

\item If $\ol_k < \ell_{k+2}$, $\ell_k \neq \ell_{k+1}$ and $\ol_k = \ell_{\max}+1$, then
\begin{align*}
\Theta
& = r_{\max} + (\ell_{\max}-\ell_{\min}) r_{\min} 
 + \POT{\bS}{3}(R_k) + \POT{\bS}{3}(R_{k+1})
 - 2^{\ol_k} c \, \Phi(\orr_k^\bullet) \\
& \geqslant r_{\max} + r_{\min} - 2^{\ol_k} c \, \Phi(\orr_k^\bullet)
 = \orr_k - 2^{\ol_k} c \, \Phi(\orr_k^\bullet) =
 2^{\ol_k} c \left(1 + \orr_k^\bullet - \Phi(\orr_k^\bullet)\right)
 \geqslant 0.
\end{align*}

\item If $\ol_k < \ell_{k+2}$ and $\ol_k = \ell_{\max}$,
then $\ell_k \neq \ell_{k+1}$ and $r_{\max}^\bullet \leqslant \orr_k^\bullet$,
and thus
\begin{align*}
\Theta
& = (\ell_{\max}-\ell_{\min}-1) r_{\min} 
 + \POT{\bS}{3}(R_k) + \POT{\bS}{3}(R_{k+1}) 
 - 2^{\ell_{\max}} c \, \Phi(\orr_k^\bullet) \\
& \geqslant \POT{\bS}{3}(R_{\max}) - 2^{\ell_{\max}} c \, \Phi(\orr_k^\bullet)
 \geqslant 2^{\ell_{\max}} c \left(\Phi(r_{\max}^\bullet) - \Phi(\orr_k^\bullet)\right)
 \geqslant 0.
\end{align*}

\item
If $\ell_{\max} < \ell_{k+2} \leqslant \ol_k$ and $\ell_k = \ell_{k+1}$, then
$\ol_k = \ell_{\max} + 1 = \ell_{k+2}$ and thus
\[
\Theta = 2^{\ell_k} c \left(3 - r_{k+1}^\bullet - 2 \Phi(r_{k+2}^\bullet)\right) -
\POT{\obS}{4}(\oR_k) \geqslant
2^{\ell_k} c \left(3 - r_{k+1}^\bullet - 2 \Phi(r_{k+2}^\bullet)\right) \geqslant 0.
\]

\item
If $\ell_{\max} < \ell_{k+2} \leqslant \ol_k$ and $\ell_k \neq \ell_{k+1}$, then
$\ol_k = \ell_{\max} + 1 = \ell_{k+2}$, and thus
\begin{align*}
\Theta
& = (\ell_{\max}-\ell_{\min}-1) r_{\min}
+ \POT{\bS}{3}(R_k) + \POT{\bS}{3}(R_{k+1}) \\
& \color{white}= (\ell_{\max}-\ell_{\min}-1) r_{\min}\color{black}
+ 2^{\ol_k} c \left(2 - \Phi(\orr_k^\bullet) - \Phi(r_{k+2}^\bullet)\right) -
\POT{\obS}{4}(\oR_k) \\
& \geqslant
2^{\ol_k} c \left(2 - \Phi(\orr_k^\bullet) - \Phi(r_{k+2}^\bullet)\right)
\geqslant 0.
\end{align*}

\item
If $\ell_k = \ell_{k+2} = \ol_k$, then $\ell_k > \ell_{k+1}$
and $r_k^\bullet \leqslant \orr_k^\bullet$, and thus
\begin{align*}
\Theta & = -r_k + (\ell_k - \ell_{k+1} - 2) r_{k+1}
+ 2^{\ell_{k+1}} c \, \Phi(r_{k+1}^\bullet) \\
& \color{white}= -r_k + (\ell_k - \ell_{k+1} - 2) r_{k+1}\color{black}
+ 2^{\ell_k} c \, \left(2 + \Phi(r_k^\bullet) - 
2 \Phi\left((\orr_k^\bullet+r_{k+2}^\bullet)/2\right)\right) \\
& \geqslant
- \orr_k
+ 2^{\ell_k} c \, \left(2 + \Phi(r_k^\bullet) - 
2 \Phi\left((\orr_k^\bullet+r_{k+2}^\bullet)/2\right)\right) \\
& \geqslant
- \orr_k
+ 2^{\ell_k} c \, \left(2 + \Phi(\orr_k^\bullet) - 
2 \Phi(\orr_k^\bullet/2)\right) =
2^{\ell_k} c \, \left(1 - \orr_k^\bullet + \Phi(\orr_k^\bullet) - 
2 \Phi(\orr_k^\bullet/2)\right) \geqslant 0.
\end{align*}

\item
If $\ell_{k+1} = \ell_{k+2} = \ol_k$, then $\ell_k < \ell_{k+1}$, and thus
\begin{align*}
\Theta &= (\ell_{k+1}-\ell_k-2) r_k - r_{k+1} 
+ \POT{\bS}{3}(R_k)
+ 2^{\ell_{k+1}} c \left(2 - \Phi(\orr_k^\bullet)\right) - \POT{\bS}{4}(\oR_k) \\
& \geqslant
- \orr_k + 2^{\ell_{k+1}} c \left(2 - \Phi(\orr_k^\bullet)\right)
= 2^{\ell_{k+1}} c \left(1 - \orr_k^\bullet - \Phi(\orr_k^\bullet)\right)
\geqslant 0.
\end{align*}

\item
If $\ell_k < \ell_{k+1}$, $\ol_k = \ell_{k+1}+1$ and $\ol_k > \ell_{k+2}$, 
then $\ell_{k+1} \geqslant \ell_{k+2}$, and thus
\[
\Theta
= (\ell_{k+1}-\ell_k - 1) r_k + \POT{\bS}{3}(R_k) + 2^{\ol_k} c \left( 1- \Phi(\orr_k^\bullet)\right)
\geqslant
2^{\ol_k} c \left(1 - \Phi(\orr_k^\bullet)\right) \geqslant 0.\]

\item
If $\ell_k < \ell_{k+1}$, $\ol_k = \ell_{k+1}$ and $\ol_k > \ell_{k+2}$, then
\begin{align*}
\Theta
& = (\ell_{k+1}-\ell_k - 2) r_k - r_{k+1} + \POT{\bS}{3}(R_k) +  2^{\ol_k} c \left(2- \Phi(\orr_k^\bullet)\right) \\
& \geqslant 
- \orr_k + 2^{\ol_k} c \left(2- \Phi(\orr_k^\bullet)\right)
= 2^{\ol_k} c \left(1 - \orr_k^\bullet - \Phi(\orr_k^\bullet)\right)
\geqslant 0.
\end{align*}

\item
If $\ell_k > \ell_{k+1}$ and $\ol_k > \ell_{k+2}$, then
$\ell_k = \ell_{k+2}$ and $\ol_k = \ell_k + 1$, and thus
\begin{align*}
\Theta 
& = (\ell_k-\ell_{k+1}-1)r_{k+1} +
\POT{\bS}{3}(R_k) + \POT{\bS}{3}(R_{k+1}) +
2^{\ell_k} c \left(2 -2 \Phi(\orr_k^\bullet) - \Phi(r_{k+2}^\bullet)\right) \\
& \geqslant
2^{\ell_k} c \left(2 -2 \Phi(\orr_k^\bullet) - \Phi(r_{k+2}^\bullet)\right)
\geqslant 0.
\end{align*}

\item
If $\ell_k = \ell_{k+1}$ and $\ol_k > \ell_{k+2}$,
then $\ol_k = \ell_k+1$ and $\ell_k = \ell_{k+1} \geqslant \ell_{k+2}$, and
thus $\Theta = 0$.
\end{enumerate}

We complete the proof by checking that one of our cases
holds whenever $\ell_k \leqslant \max\{\ell_{k+1},\ell_{k+2}\}$:
cases 1, 2 or 3 hold when $\ol_k < \ell_{k+2}$,
cases 4 or 5 hold when $\ell_{\max} < \ell_{k+2} \leqslant \ol_k$,
cases 6 or 7 hold when $\ell_{\max} = \ell_{k+2} = \ol_k$, and
cases 8, 9, 10 or 11 hold when $\ell_{k+2} < \ol_k$
\end{proof}

We conclude this section by proving Theorem~\ref{thm:complexity-nH+D}
as follows.

\begin{proof}[Proof of Theorem~\ref{thm:complexity-nH+D}]
Let $\bS = (R_1,\ldots,R_\rho)$ be the initial state, i.e.,
the run decomposition of the array to sort, and let $\obS = (\oR_\final)$ be the
last state encountered in the algorithm,
whose only run has length $r_\final = r_1+\ldots+r_\rho = n$.

For every run $R_i$, we have
\begin{align*}
\POT{\bS}{0}(R_i) & \leqslant 2^{\ell_i+1} c - \ell_i r_i 
= \left(2/(1+r_i^\bullet) + \log_2(c) - \log_2(r_i) - \log_2(1+r_i^\bullet)\right) r_i \\
& \leqslant \left(2 + \log_2(c) - \log_2(r_i)\right) r_i
\hspace{18mm} \text{by applying Lemma~\ref{lem:2pow-x} to 
$x = \log_2(1 + r_i^\bullet)$.}
\end{align*}
Therefore, it holds that
$
\Pot(\bS) \leqslant
\sum_{i=1}^\rho (2 + \log_2(c)) r_i - r_i \log_2(r_i)=
(2 + \log_2(c / n) + \H) n$.
Then, we also verify that
\begin{align*}
\Pot(\obS) & = 
2^{\ell_\final} c \, \Phi(r_\final^\bullet) - (\ell_\final+1) r_\final \\
& =
\left(\Phi(n^\bullet) / (1 + n^\bullet) - \log_2(n) + \log_2(1 + n^\bullet) + \log_2(c) - 1\right) n \\
& \geqslant
\left(3 - \Delta - \log_2(n) + \log_2(c) - 1 \right) n.
\end{align*}

Consequently, it follows from Proposition~\ref{pro:potential-2} that
the total merge cost of \mbox{\ccASS} is
$\mc \leqslant \Pot(\bS) - \Pot(\obS) \leqslant
(\H + \Delta)$, which completes the proof.
\end{proof}

\section{Best-case and worst-case merge costs}
\label{sec:best-case}

In the introduction, we mentioned that, unlike \PoS and \LASS,
the algorithm \cASS is $3$-aware.
This suggests that it might be preferred to \PoS and \LASS.
However, the worst-case merge cost of \PoS and \LASS is at most $n (\H + 2)$,
whereas the above analysis only proves that the merge
cost of \cASS is bounded from above by $n (\H + \Delta)$. Hence, and most notably
in cases where $\H$ is small, \PoS or \LASS might be significantly better options.
Furthermore, in other cases than the worst case, we have little
information on the relative costs of \cASS, \PoS and \LASS.

We address these problems as follows.
First, we derive lower bounds on the best-case merge cost of any merge policy.
Then, we prove that the worst-case merge cost of \PoS and \LASS is actually optimal
among all the stable natural merge-sort algorithms.

\subsection{Best-case merge cost}

Given a sequence $\mathbf{r} = (r_1,\ldots,r_\rho)$ of run
lengths and an array of length $n = r_1 + \ldots + r_\rho$
that splits into monotonic runs of lengths
$r_1,\ldots,r_\rho$, what is the best merge cost of
any merge policy?
An answer to this question is given by
Theorem~\ref{thm:lower-bound-merge-cost},
which is a rephrasing of results
from~\cite{BaNa13,Garsia77,HuTucker71,Mannila1985,munro2018nearly},
adapted to the context of sorting algorithms and merge costs.

This answer comes from the analysis of
\emph{merge trees}, which we describe below, and of two
algorithms: \MinS~\cite{BaNa13,Ta09}, which is \emph{not}
a stable natural merge sort, and whose merge policy is described
in Algorithm~\ref{alg:MinS};
and \MinSS, which \emph{is} a stable natural merge sort, and whose merge policy
can be computed by following either the Hu-Tucker~\cite{HuTucker71}
or Garsia-Wachs~\cite{Garsia77}
algorithms for constructing optimal binary search trees.

Note that \MinS may require merging
\emph{non-adjacent} runs, which might therefore be
less easy to implement than just merging adjacent runs.
Yet, such merges (between runs of lengths $m$ and $n$)
can still be carried out in time $m+n$, for instance by
using linked lists, and therefore the merge cost
remains an adequate measure of complexity for this algorithm.

\begin{algorithm}[h]
\begin{small}
\SetArgSty{texttt}
\DontPrintSemicolon
\SetKwInOut{Input}{Input}
\Input{Array to $A$ to sort}
\KwResult{The array $A$ is sorted into a single run.}
\BlankLine
$\rundecomp \gets $ the run decomposition of $A$\;
\While{\textrm{$\rundecomp$ contains at
least two runs}}{
  merge the two shortest runs in $\rundecomp$
}
\end{small}
\caption{\MinS\label{alg:MinS}}
\end{algorithm}

\begin{restatable}{definition}{defmergetree}
\label{def:merge-tree}
Let $\M$ be a merge policy and
$\R$ a sequence of runs.
We define the \emph{merge tree} induced by $\M$ on $\R$
(we may omit mentioning $\M$ and $\R$ if the context is clear)
as the following binary rooted tree.
Every node of the tree is identified with a run, either present
in the initial sequence or created by the merge policy.
The runs in the sequence $\R$ are the leaves of the tree, and when
two runs $R_1$ and $R_2$ are merged together in a run $\oR$,
the run $\oR$ is identified with the internal node whose
children are $R_1$ and $R_2$:
if the run $R_1$ was placed to the left of $R_2$ in the
sequence to be sorted, then $R_1$ is the left sibling of $R_2$.
\end{restatable}

Any such tree is the tree of a binary prefix encoding on a text
on the alphabet $\{1,2,\ldots,\rho\}$, which contains
$r_i$ characters $i$ for all $i \leqslant \rho$.
Furthermore, denoting by $d_i$ the depth of the leaf $R_i$,
both the length of this code and the
merge cost of the associated merge policy are equal to
$\sum_{i=1}^\rho d_i r_i$.
Hence, character encoding algorithms unsurprisingly yield
efficient sorting algorithms.

Similarly, if $\M$ is a stable merge policy,
and considering a node $R_i$ to be smaller than another node $R_j$ whenever $i < j$,
the merge tree induced by $\M$ on $\R$ is a binary search tree.
Then, assuming that there will be $r_i$ queries for finding the leaf $R_i$ in that tree,
the total cost of these queries is also equal to $\sum_{i=1}^\rho d_i r_i$,
whence the usefulness of algorithms for constructing optimal binary search trees.

In particular, from now on, we identify every merge tree $\T$ with
a collection of merges (these are the merges between the nodes $R_1$ and $R_2$
that are siblings in $\T$),
and we define the \emph{cost} of that tree as the sum of the costs of these
merges. Equivalently, the cost of $\T$ is the sum of the lengths or the runs
$R$ that are internal nodes of $\T$.

\begin{theorem}\label{thm:lower-bound-merge-cost}
The merge cost of any (stable or not) merge policy on a non-sorted array is minimised by the algorithm \MinS,
and the merge cost of any stable merge policy on a non-sorted array is minimised by
the algorithm \MinSS. Both costs are at least $\max\{n, n \H\}$.
\end{theorem}

\begin{proof}
The merge trees of the algorithms \MinS and 
\MinSS are a Huffman tree and an optimal
binary search tree.
This means that \MinS (respectively, \MinSS) 
is indeed the natural (respectively,
stable natural) merge sort with the least
merge cost, and that this merge cost is the length of a
Huffman code for a text containing $r_i$ occurrences of the
character~$i$.
Such a text has entropy $\H$, and therefore
Shannon theorem proves that the associated Huffman code
has length at least $n \H$.
Furthermore, it is clear that the merge cost of $\MinS$ must be at least
$n$, which completes the proof.
\end{proof}

\subsection{Optimality of worst-case merge costs}

The lower bound provided by 
Theorem~\ref{thm:lower-bound-merge-cost}
matches quite well
the worst-case merge costs of both \PoS and \cASS.
In particular, and independently of the sequence to be sorted,
the merge cost of \PoS (respectively, \cASS) lies between
$n \H$ and $n(\H + 2)$ (respectively, $n(\H + \Delta)$).

This shows, among others, that both
\PoS and \cASS are very close to optimal
when $\H$ is large. When $\H$ is small, however, the respective performances
of \PoS and \cASS are still worth investigating.
In particular, and given the tiny margin of freedom between these lower and upper bounds,
it becomes meaningful to check whether our upper bounds are indeed optimal.

A first result in that direction
is the following one, which implies that 
the worst-case merge costs of \PoS
and of \LASS are optimal.

\begin{proposition}\label{pro:pos-optimal}
Let $\S$ be a stable merge policy.
Assume that there exist real constants $\alpha$ and $\beta$ such that
the merge cost of $\S$ can be bounded from above by
$n (\alpha \H + \beta)$.
Then, we have $\alpha \geqslant 1$ and $\beta \geqslant 2$.
\end{proposition}

\begin{proof}
First, Theorem~\ref{thm:lower-bound-merge-cost} proves,
by considering arbitrarily large values of $\H$,
that $\alpha \geqslant 1$.
Second, let $n \geqslant 6$ be some integer,
and let $\mathbf{a}$ be some array of data that splits into
runs of lengths $2$, $n-4$ and $2$.
One checks easily that
$n \H \leqslant 4 \log_2(n) + 2$ and that
every stable merge policy has a merge cost $2n-2$ 
when sorting $\mathbf{a}$.
Hence, by considering arbitrarily large values of $n$, it follows that $\beta \geqslant 2$.
\end{proof}

While Proposition~\ref{pro:pos-optimal} indeed proves that
the worst-case merge cost of \PoS and of \LASS is optimal,
addressing the optimality of the worst-case merge cost
of \cASS requires considering another example.

\begin{proposition}\label{pro:ass-not-optimal}
Let $\beta$ be a real constant such that
the merge cost of \cASS can be bounded from above
by $n(\H + \beta)$ when $n$ is large enough.
Then, we have $\beta \geqslant \Delta$.
\end{proposition}

\begin{proof}
Let $k \geqslant 3$ be some integer, and let $m = 2^k$,
so that $m \geqslant 8$.
Then, consider an array that decomposes into
seven runs $R_1,\ldots,R_7$ with lengths
$2$, $2m-10$, $2$, $m+1$, $2$, $2m+2$ and $1$ respectively.
This array is of length
$n = 5 m = 5 \cdot 2^k$, and
\[
n \H = \sum_{i=1}^7 \log_2(n/r_i) \, r_i = n \left(\log_2(5) - 4/5 + o(1)\right).\]
Meanwhile, the merge cost of \cASS is equal to
$\mc = 20m-23 = n \left(4 + o(1)\right)$.
Since $\mc \leqslant n (\H+\beta)$, it follows that
$\beta \geqslant 4 - (\log_2(5) - 4/5) = \Delta$.
\end{proof}

Hence, the worst-case merge cost of \cASS is \emph{not} optimal.
This is not very surprising since, unlike \PoS and \LASS, 
the algorithm \cASS cannot take into account the total length of the
input until that end is indeed reached.

\section{Approximately optimal sorting algorithms}
\label{sec:optimal}

In previous sections, we have shown that \cASS is both very easy to obtain
by modifying the code of \TS and very effective.
Yet, and although \cASS is optimal up to an additive term of at most
$\Delta n$, this term may still be of importance
when considering arrays of data with small run-length entropy $\H$.
For example, in Table~\ref{table:diff:cost}, we present the run lengths
of arrays no which \cASS, \PoS and \TS have significantly different merge costs:
the overhead of each algorithm, compared to the others, can climb up to between 40\%
and 100\% (at least),
and each algorithm can be quite better or quite worse than the other two.

\begin{table}[ht]
\begin{center}
\begin{tabular}{|c|C{25.8mm}|C{25.8mm}|C{25.8mm}|C{25.8mm}|}
\hline
\multirow{2}{*}{Run lengths} & \multicolumn{4}{c|}{Merge cost} \\
\cline{2-5}
 & \textsf{adaptive} & \multirow{2}{*}{\PoS}& \multirow{2}{*}{\TS} & 
 \textsf{Minimal}\\
{\small ($n = 2^k$, $k = 2^\ell$, $\ell \geqslant 2$)} & \ShS & & & \textsf{StableSort}\\
\hline
\multirow{2}{*}{$(2n-1,n,2)$} & $6n$ & $4n+3$ & $4n+3$ & \multirow{2}{*}{$4n+3$}\\
& {\small ($50\%$ overhead)} & {\small (optimal)} & {\small (optimal)} & \\
\hline
Sequence $\mathbf{s}_{k,3}$ & $15n+6k-58$ & \cellcolor{black!10}$21n-2k-66$ &
$15n+6k-60$ & \multirow{2}{*}{$15n+6k-60$} \\
\arrayrulecolor{black!10}\cline{3-3}\arrayrulecolor{black}
{\small (see below)} & 
{\small ($0\%$ overhead)} & {\cellcolor{black!10}\small ($40\%$ overhead)} & {\small (optimal)} & \\
\hline
Sequence $\mathbf{t}_k$ & $\mathrm{cost}_k$ & $\mathrm{cost}_k$ &
\cellcolor{black!10}$8n^2+5n-3k-7$ & \multirow{2}{*}{$\mathrm{cost}_k$} \\
\arrayrulecolor{black!10}\cline{4-4}\arrayrulecolor{black}
{\small (see below)} & 
{\small (optimal)} & {\small (optimal)} & {\cellcolor{black!10}\small ($100\%$ overhead)} & \\
\hline
\multicolumn{1}{|l|}{$(2n,2,1,n-1,$}
& $9n+16$ & $12n+8$ & $12n+12$
& \multirow{2}{*}{$9n+16$}\\
\multicolumn{1}{|r|}{$1,1,n-3,3)$} & 
{\small (optimal)} & {\small ($33\%$ overhead)} & {\small ($33\%$ overhead)} & \\
\hline
\multicolumn{1}{|l|}{$(2n,n,n/2,n/4,\ldots,$}
& $12n-6$ & $8n+2k-2$ & $12n-6$
& \multirow{2}{*}{$8n+2k-2$}\\
\multicolumn{1}{|r|}{$32,16,8,4,4,2)$} & 
{\small ($50\%$ overhead)} & {\small (optimal)} & {\small ($50\%$ overhead)} & \\
\hline
\multicolumn{1}{|l|}{$(4n-1,n,n/2,n/4,\ldots,$}
& \cellcolor{black!10}$16n-8$ & $10n+2k-3$ & $12n-5$
& \multirow{2}{*}{$10n+2k-3$}\\
\arrayrulecolor{black!10}\cline{2-2}\arrayrulecolor{black}
\multicolumn{1}{|r|}{$32,16,8,4,4,2)$} & 
{\cellcolor{black!10}\small ($60\%$ overhead)} & {\small (optimal)} & {\small ($20\%$ overhead)} & \\
\hline
Sequence $\mathbf{s}_{k,2}$ & $21n-4k-34$ & $21n-2k-32$
& $15n+6k-30$ & \multirow{2}{*}{$15n+6k-30$} \\
{\small (see below)} & 
{\small ($40\%$ overhead)} & {\small ($40\%$ overhead)} & {\small (optimal)} & \\
\hline
\end{tabular}
\end{center}
\caption[]{
Merge costs for various run lengths \hspace{2mm} --- \hspace{2mm}
The overhead indicated is valid when $n \to +\infty$.
Gray cells indicate the worst performance of each algorithm.
The sequence $\mathbf{s}_{k,m}$ is the sequence \newline
\begin{minipage}{\linewidth}
~\hfill
$(3 \cdot 2^k-2,2,2,3 \cdot 2^{k-1}-4,2,2,3 \cdot 2^{k-2}-4,2,2,\ldots,
3 \cdot 2^i-4,2,2,\ldots,3 \cdot 2^m-4,2,2,3 \cdot 2^m-2)$.
\hfill~
\end{minipage}
The sequence $\mathbf{t}_k$ is the sequence \newline
\begin{minipage}{\linewidth}
~\hfill
$(2^k(2^k+1)-1,2^{k-1}(2^k+1)-1,\ldots,2(2^k+1)-1,(2^k+1)-1,2^k-1,1,1,\ldots,1)$
\hfill~
\end{minipage}
with $k+2$ terms ``$1$'' at the end, and
$\mathrm{cost}_k = 4n^2+3n+k(k+1)/2+(k+2)\ell$.
\label{table:diff:cost}}
\end{table}

Moreover, arrays with small run-length entropy
may have arbitrarily large lengths, but also
arbitrarily many monotonic runs. This is, for instance,
the case of arrays whose run lengths form the sequence
$(2,2,\ldots,2,k^2)$, 
where the $k$ first terms are integers $2$:
although this sequence contains $k+1$ terms, it is associated with
a value of $\H \approx 4 \log_2(k) / k$.

Hence, and since the parameters $n$, $\rho$ and $\H$ may vary
more or less independently of each other (up to the rather loose inequalities
$(\rho-1) \log_2(n) / n \leqslant \H \leqslant \log_2(\rho) 
\leqslant \log_2(n)$),
we aim for the
uniform approximation result captured by the following definition.

\begin{definition}\label{def:eps-optimal}
Let $\A$ be a stable natural merge sort, and let
$\varepsilon \geqslant 0$ be a real number. We say that
$\A$ is \emph{$\varepsilon$-optimal} if,
for every stable natural merge sort $\B$
and every array to be sorted,
the respective merge costs $\mc_a$ and 
$\mc_b$ of $\A$ and $\B$ satisfy the inequality
$\mc_a \leqslant (1 + \varepsilon) \mc_b$.
\end{definition}

Below, we study the $\varepsilon$-optimality of algorithms
such as \TS, \cASS, or even \PoS and \LASS.
To that aim, we first define the family of $k$-aware algorithms,
to which all these algorithms belong;
our notion subsumes and generalises slightly
the notion of \emph{awareness} of Buss and Knop~\cite{BuKno18}.
While \TS and \cASS are
$(4,3)$- and $(3,3)$-aware algorithms in the sense of
Buss and Knop, this novel notion
also captures \PoS and \LASS,
which are respectively length-$(\infty,3)$-
and length-$(3,3)$-aware algorithms.

\begin{definition}
Let $k$ and $\ell$ be elements of the set $\{0,1,2,\ldots\} \cup \{\infty\}$, with $k \geqslant \ell$. A deterministic sorting algorithm is said to be \emph{$(k,\ell)$-aware} (or simply \emph{$k$-aware} if $k = \ell$) if it sorts
arrays of data by manipulating a stack of runs (where each run is
represented by its first and last indices) and operating as follows:
\begin{itemize}
 \item the algorithm discovers, from the left to the right,
 the monotonic runs in which the array is split, and it 
 pushes these runs on the stack when discovering them;
 \item the algorithm is allowed to merge two consecutive
 runs in the $\ell$ top runs of its stack only, and its decision may be based
 only on the lengths of the top $k$ runs of the stack,
 and on whether the algorithm already discovered the entire array;
 \item if $\ell = \infty$, the algorithm may merge any two consecutive runs in its stack;
 if $k = \infty$, it is granted an infinite memory,
 and thus its decisions may be based all the push or merge operations it
 performed (and on the lengths of the runs involved in these operations).
\end{itemize}
If, furthermore, the algorithm is given access to the length of the
array and can base its decisions on this information,
then we say that it is a \emph{length-$(k,\ell)$-aware} algorithm.
\end{definition}

In particular, note that \PoS needs to remember not only the lengths
of the array and of the runs stored in its stack, but also their
powers (or, alternatively, the positions they span in the array), 
which does not fall into the scope of length-$(3,3)$-awareness.
Our new notion of awareness could be further generalised in a meaningful way
that would make \PoS a length-$(3,3)$-aware algorithm. However, our further results already
apply to all length-$(\infty,\ell)$-aware algorithms, and therefore such generalisations
are not needed within the framework of this article.

We present now results going in opposite directions, and which,
taken together, form a first step towards finding the best
approximation factor of $k$-aware algorithms.
One direction is explored in Section~\ref{sec:unapproximable},
where we prove that, once the integer $k$ is fixed, $k$-aware algorithms
cannot be $\varepsilon$-optimal for arbitrarily small values of $\varepsilon$.
The other direction is explored in Section~\ref{sec:approximable},
where we prove that, for all $\varepsilon$, there exists a $\varepsilon$-optimal
algorithm that is $k$-aware for some $k$.

\subsection{Unapproximability bounds}\label{sec:unapproximable}

Below, we present a few \emph{unapproximability} results.
We first investigate lower bounds on those numbers $\varepsilon$ such that,
for a given integer $k \geqslant 3$, there exists a $\varepsilon$-optimal
algorithm that is $k$-aware (in Proposition~\ref{pro:no-epsilon-optimal-1})
or length-$(\infty,k)$-aware (in Proposition~\ref{pro:no-epsilon-optimal-2}).
Then, we focus on the case $k = 2$,
and we prove that length-$(\infty,k)$-aware algorithms
cannot be $\varepsilon$-optimal for \emph{any} $\varepsilon$.

\begin{proposition}\label{pro:no-epsilon-optimal-1}
Let $k \geqslant 3$ be an integer,
and let
\[\theta_k = 1 / ((10k + 12) \log_2(2k+2)).\]
No $k$-aware sorting algorithm is
$\theta_k$-optimal.
\end{proposition}

\begin{proof}
Let $h = \lceil \log_2(k+2) \rceil$, and let $\rho = 2^h-1$.
We design six arrays $\bA^{x,y}$, where $x \in \{1,3\}$ and $y \in \{0,2,4\}$.
Each array will have $\rho$ runs, and we prove below that
no $k$-aware algorithm
can approach the merge cost of \MinSS
by a factor $\theta_k$ on those six arrays.

For all $x$ and $y$, let $r_1 = x+5$, $r_2 = \ldots = r_{\rho-1} = 5$,
and $r_\rho = y+5$. We build the array $\bA^{x,y}$
as any array whose run lengths
form the sequence $(r_1,r_2,\ldots,r_\rho)$.
Thus, the length of this array is $n = 5 \rho + x+y$. 

The merge tree induced by \MinSS on the array $\bA^{x,y}$
is a Huffman tree.
That tree is a perfectly balanced binary tree of height $h$,
except that its leftmost leaf (if $x > y$) or rightmost leaf (if $x < y$)
has been deleted.
Then, the merge cost of \MinSS on the array $\bA^{x,y}$ is simply
$\mcopt{} = h n - 5 - \max\{x,y\}$. In particular, note that
$h \leqslant \log_2(k+1) + 1 = \log_2(2k+2)$, so that
$\rho \leqslant 2k+1$ and that $n \leqslant 5 \rho + 7 \leqslant 10 k + 12$.
It follows that
\[\mcopt{} = h n - 5 - \max\{x,y\} < h n \leqslant 1 / \theta_k.\]

Hence, let $\A$ be some $\theta_k$-optimal algorithm.
Since $\mcopt{} \leqslant \mc_a \leqslant (1 + \theta_k) \mcopt{} < \mcopt{} + 1$,
it follows that $\mc_a = \mcopt{}$.
In particular, this means that $\A$ must sort $\mathbf{A}^{x,y}$
optimally, even in the class of not necessarily stable merge algorithms.

In particular, let us denote by $R_i$ the $i$\textsuperscript{th}
run of $\bA^{x,y}$. Then, let $d_i$ the depth of the run $R_i$
in the merge tree associated to $\A$, i.e., the number of
merges into which the elements of $R_i$ have been involved.
We know that $\mc_a = \sum_{i=1}^{\rho-1} r_i d_i$ and, by minimality of $\mc_a$,
it must be the case that $d_i \geqslant d_j$ whenever $r_i < r_j$;
otherwise, by exchanging the runs $R_i$ and $R_j$, we would strictly decrease the
merge cost of $\A$ (although this might require merging non-adjacent runs).
Now, we prove that $\A$ must merge the same pairs of runs as \MinSS.
We treat the case where $x < y$, the case $x > y$ being entirely analogous.

Let $\bD$ be the largest of the integers $d_i$, and let
$\bI$ be the number of indices $i$ such that $d_i = \bD$.
It comes that $\bD \geqslant d_i \geqslant d_1 \geqslant d_\rho$ for all
$i = 2,3,\ldots,\rho-1$.
Then, if $d_\rho \neq \bD$, let $R_i$ and $R_{i+1}$ be two runs that $\A$ merges
with each other, and such that $d_i = d_{i+1} = \bD$.
The run resulting from the merge is (strictly) larger than $R_\rho$, 
and therefore, again by minimality of $\mc_a$, we know that $d_\rho \geqslant
d_i - 1 = \bD-1$.

Finally, Kraft equality states that $1 = \sum_{i=1}^\rho 2^{-d_i}$.
Since every integer $d_i$ is equal either to $\bD$ or to $\bD-1$,
this means that
$1 = 2^{-\bD} \bI + 2^{1-\bD} (\rho-\bI)$, i.e., that
$2^\bD = 2 \rho - \bI$.
Since $1 \leqslant \bI \leqslant \rho$, it follows that
$2^{h+1} = 2 \rho + 2 > 2^\bD \geqslant \rho > 2^{h-1}$,
and therefore we conclude that $\bD = h$
and that $\bI = 2\rho - 2^{\bD} = 2^h-2 = \rho-1$.
Hence, we know that $d_1 = \ldots = d_{\rho-1} = h$ and that $d_\rho = h-1$.
Remembering that $\A$ can merge adjacent runs only,
this means that $\A$ must perform the same merges as \MinSS,
as announced above.

Let us now further assume that $\A$ is $k$-aware.
Then, when scanning some array $\bA^{x,y}$, it cannot distinguish between
the arrays $\bA^{x,x-1}$ and $\bA^{x,x+1}$ until it discovers the rightmost
run $R_\rho$. Moreover, no two adjacent runs $R_i$ and $R_{i+1}$ would ever
be merged by \MinSS in both arrays $\bA^{x,x-1}$ and $\bA^{x,x+1}$.
Thus, $\A$ must wait until discovering the run $R_\rho$ before it can perform
a single merge.

Imagine now that $y = 2$, i.e., $y = x+1$ if $x = 1$, or $y = x-1$ if $x = 3$.
When discovering the run $R_\rho$, the run $R_1$ lies at the bottom of the
stack, which contains $\rho \geqslant k+1$ runs. Therefore, $\A$ cannot
distinguish anymore between the arrays $\bA^{y-1,y}$ and $\bA^{y+1,y}$.
Once again, no two adjacent runs $R_i$ and $R_{i+1}$ would ever
be merged by \MinSS in both arrays $\bA^{y-1,y}$ and $\bA^{y+1,y}$.

This proves that $\A$ cannot make sure that it will perform
the same merges as \MinSS on all the arrays $\bA^{x,y}$, which completes the proof.
\end{proof}

\begin{proposition}\label{pro:no-epsilon-optimal-2}
Let $k \geqslant 3$ be an integer,
and let
\[\varepsilon_k = 1/2^{k+7}.\]
No length-$(\infty,k)$-aware sorting algorithm is
$\varepsilon_k$-optimal.
\end{proposition}

\begin{proof}
The proof below follows similar lines as the proof of
Proposition~\ref{pro:no-epsilon-optimal-1}, but its details
are often quite different.

We design two arrays $\bA^-$ and $\bA^+$ and prove below that
no length-$(\infty,k)$-aware algorithm
can approach the merge cost of \MinSS
by a factor $\varepsilon_k$ on those two arrays.
The array $\bA^-$ can be any array whose run lengths
form the sequence $(r_1,r_2,\ldots,r_{2k+4})$ defined by
$r_i = 2^{k+4-i}$ for
$i = 1,2,\ldots,k$; $r_{k+1} = 9$;
$r_{k+2} = 8$; $r_{k+3} = 11$; $r_{k+4} = 18$;
$r_{k+4+i} = 2^{i+4}$ for $i = 1,2,\ldots,k-2$;
$r_{2k+3} = 2^{k+2} + 2$; and
$r_{2k+4} = 2^{k+3}$.
The array $\bA^+$ is any array
whose run lengths form the sequence
$(r_1,r_2,\ldots,r_{2k+2},r_\final)$,
where $r_\final = r_{2k+3} + r_{2k+4} = 3 \times 2^{k+2} + 2$.
Hence, the length of both arrays is
$n = 9 \times 2^{k+2}$.
In what follows, we will denote by $R_i$ the $i$\textsuperscript{th}
run of $\bA^-$, and by $R_\final$ the rightmost run of $\bA^+$.

\begin{figure}[ht]
\begin{small}
\begin{center}
\begin{tikzpicture}[xscale=0.42,yscale=0.7]
\node at (5,0.3) {};
\node at (5,-12) {};

\draw[ultra thick] (10,-8.5) -- (9,-7.5) ;
\draw[ultra thick] (8,-8.5) -- (9,-7.5) -- (7.5,-6.5) -- (6.5,-5.5);
\draw[ultra thick,densely dotted] (6.5,-5.5) -- (5,-4);
\draw[ultra thick] (5,-4) -- (2,-1) -- (10,0) -- (18,-1) -- (15,-4);
\draw[ultra thick,densely dotted] (15,-4) -- (13.5,-5.5);
\draw[ultra thick] (13.5,-5.5) -- (12.5,-6.5) -- (12,-7.5);

\node[treenode] at (10,0) {};
\node[treenode] at (9,-7.5) {};
\foreach \i/\j/\m in {1/2/1,2/3/2,3/4/3,4/5/4} {
  \draw[ultra thick] (\i,-\j) -- (\i+1,1-\j);
  \node[treenode] at (\i+1,1-\j) {};
  \leaf{\i,-\j}{0.1}{$R_\m$}
}
\foreach \i/\j/\m in {6/7/-2,7/8/-1} {
  \draw[ultra thick] (\i-0.5,0.5-\j) -- (\i+0.5,1.5-\j);
  \node[treenode] at (\i+0.5,1.5-\j) {};
  \leaf{\i-0.5,0.5-\j}{0.1}{$R_{k\m}$}
}
\foreach \i/\j/\m in {8/10/,10/10/+2,12/9/+4} {
  \draw[ultra thick] (\i-0.5,0.5-\j) -- (\i,1.5-\j);
  \leaf{\i-0.5,0.5-\j}{0.1}{$R_{k\m}$}
}
\foreach \i/\j/\m in {9/10/+1,11/10/+3,13/9/+5} {
  \draw (\i-0.5,0.5-\j) -- (\i-0.5,-0.5-\j);
  \draw[ultra thick] (\i-0.5,0.5-\j) -- (\i-1,1.5-\j);
  \node[treenode] at (\i-1,1.5-\j) {};
  \leaf{\i-0.5,0.5-\j}{0.75}{$R_{k\m}$}
}
\foreach \i/\j/\m in {14/8/+6,15/7/+7} {
  \draw[ultra thick] (\i-0.5,0.5-\j) -- (\i-1.5,1.5-\j);
  \node[treenode] at (\i-1.5,1.5-\j) {};
  \leaf{\i-0.5,0.5-\j}{0.1}{$R_{k\m}$}
}
\foreach \i/\j/\m in {16/5/+1,17/4/+2,18/3/+3,19/2/+4} {
  \draw[ultra thick] (\i,-\j) -- (\i-1,1-\j);
  \node[treenode] at (\i-1,1-\j) {};
  \leaf{\i,-\j}{0.1}{$R_{2k\m}$}
}
\end{tikzpicture}
\begin{tikzpicture}[xscale=0.42,yscale=0.7]
\node at (5,0.3) {};
\node at (5,-12) {};

\draw[ultra thick] (9,-8.5) -- (8.5,-7.5) -- (6.5,-5.5);
\draw[ultra thick,densely dotted] (6.5,-5.5) -- (5,-4);
\draw[ultra thick] (5,-4) -- (2,-1) -- (9.5,0) -- (17,-1) -- (14,-4);
\draw[ultra thick,densely dotted] (14,-4) -- (12.5,-5.5);
\draw[ultra thick] (12.5,-5.5) -- (11.5,-6.5) -- (11,-7.5);

\node[treenode] at (9.5,0) {};
\foreach \i/\j/\m in {1/2/1,2/3/2,3/4/3,4/5/4} {
  \draw[ultra thick] (\i,-\j) -- (\i+1,1-\j);
  \node[treenode] at (\i+1,1-\j) {};
  \leaf{\i,-\j}{0.1}{$R_\m$}
}
\foreach \i/\j/\m in {6/7/-2,7/8/-1,8/9/} {
  \draw[ultra thick] (\i-0.5,0.5-\j) -- (\i+0.5,1.5-\j);
  \node[treenode] at (\i+0.5,1.5-\j) {};
  \leaf{\i-0.5,0.5-\j}{0.1}{$R_{k\m}$}
}
\foreach \i/\j/\m in {9/10/+1,11/9/+3} {
  \draw[ultra thick] (\i-0.5,0.5-\j) -- (\i,1.5-\j);
  \leaf{\i-0.5,0.5-\j}{0.1}{$R_{k\m}$}
}
\foreach \i/\j/\m in {10/10/+2,12/9/+4} {
  \draw (\i-0.5,0.5-\j) -- (\i-0.5,-0.5-\j);
  \draw[ultra thick] (\i-0.5,0.5-\j) -- (\i-1,1.5-\j);
  \node[treenode] at (\i-1,1.5-\j) {};
  \leaf{\i-0.5,0.5-\j}{0.75}{$R_{k\m}$}
}
\foreach \i/\j/\m in {13/8/+5,14/7/+6} {
  \draw[ultra thick] (\i-0.5,0.5-\j) -- (\i-1.5,1.5-\j);
  \node[treenode] at (\i-1.5,1.5-\j) {};
  \leaf{\i-0.5,0.5-\j}{0.1}{$R_{k\m}$}
}
\foreach \i/\j/\m in {15/5/,16/4/+1,17/3/+2} {
  \draw[ultra thick] (\i,-\j) -- (\i-1,1-\j);
  \node[treenode] at (\i-1,1-\j) {};
  \leaf{\i,-\j}{0.1}{$R_{2k\m}$}
}
\foreach \i/\j/\m in {18/2} {
  \draw[ultra thick] (\i,-\j) -- (\i-1,1-\j);
  \node[treenode] at (\i-1,1-\j) {};
  \leaf{\i,-\j}{0.1}{$R_{\final}$}
}
\end{tikzpicture}
\end{center}
\vspace{-12mm}
\end{small}
\caption{Merge trees associated to \MinSS when sorting $\bA^-$ (left)
and $\bA^+$ (right)\label{fig:epsilon-opt:huffman}}
\end{figure}

When sorting either array $\bA^+$ or $\bA^-$,
the algorithm \MinSS performs the same merges as those of a Huffman tree.
That tree looks like a bunch of grapes, as
illustrated in Figure~\ref{fig:epsilon-opt:huffman}.
Thus, the merge costs
of the algorithms \MinSS and \MinS are equal to each other.
These merge costs are $\mcopt{-} = 29 \times 2^{k+2} - 2k - 48$
on the array $\bA^-$, and
$\mcopt{+} = 26 \times 2^{k+2} - 2k - 45$ on the array $\bA^+$:
in both cases, we have $\mcopt{} < 2^{k+7} = 1 / \varepsilon_k$.

Hence, let $\A$ be some $\varepsilon_k$-optimal algorithm.
Since $\mcopt{} \leqslant \mc_a \leqslant (1 +\varepsilon_k) \mcopt{} < \mcopt{} + 1$,
it follows that $\mc_a = \mcopt{}$.
In particular, this means that $\A$ must sort $\mathbf{A}^{\pm}$
optimally, even in the class of not necessarily stable merge algorithms.
We prove now the following claims:
\begin{enumerate}[label=(\roman*)]
 \item\label{claim-1} when sorting $\bA^-$, the algorithm $\A$ must merge the runs $R_{k+2}$ and $R_{k+3}$;
 \item\label{claim-2} when sorting $\bA^+$, the algorithm $\A$
must merge the run
$R_{k+3}$ successively with the runs $R_{k+4},R_{k+5},\ldots,R_{2k+2},R_\final$.
\end{enumerate}

Let us first prove Claim~\ref{claim-1}. 
In order to do so, let us denote by $d_i$ be the depth of the run $R_i$
in the merge tree associated to $\A$ when sorting $\bA^-$.
Since $\A$ is optimal among all merge algorithms when sorting $\bA^-$,
we can prove the following statement:
if $R$ and $R'$ are two runs ever manipulated by $\A$ (either
because the are runs of the original array $\bA^-$ or because they
result from some merge operation), and if $r < r'$, then $d \geqslant d'$.
Indeed, both inequalities $r < r'$ and $d < d'$
were to hold, exchanging the runs $R$ and $R'$
would provide us with a smaller total merge cost.

Since $R_{k+1}$ and $R_{k+2}$ are the two shortest runs of the array, 
they must be the deepest runs as well, and thus $d_{k+1} = d_{k+2}$.
Then, the run $R_{k+2}$ must be merged with either $R_{k+1}$ or $R_{k+3}$.
For the sake of contradiction, assume
that $R_{k+1}$ and $R_{k+2}$ were merged with each other, and let 
$\oR$ be the run resulting from that merge.

We know that $r_{k+4} > \orr = 17 > r_{k+3}$,
and therefore that $d_{k+4} \leqslant \od = d_{k+1} - 1 \leqslant d_{k+3} \leqslant d_{k+1}$.
Then, since, $d_{k+1} - d_{k+3} \equiv \mcopt{-} \equiv 0 \pmod{2}$,
it follows that $d_{k+3} = d_{k+1}$.
Consequently, $R_{k+3}$ is also the deepest run of the array,
and since it is not merged with $R_{k+2}$, it must be merged with $R_{k+4}$.
But this is impossible since $d_{k+4} \geqslant \od > d_{k+3}$.
Hence, our assumption was incorrect, which proves that $\A$ must merge
$R_{k+2}$ with $R_{k+3}$.

Now that Claim~\ref{claim-1} is proven,
let us also prove Claim~\ref{claim-2}.
Like above, we denote by $d_i$ the depth of the run $R_i$, and by
$d_\final$ the depth of the run $R_\final$.
Once again, the runs $R_{k+1}$ and $R_{k+2}$ are the deepest run of the array,
and thus $d_{k+1} = d_{k+2}$.
However, this time, $d_{k+1} - d_{k+3} 
\equiv \mcopt{+} \equiv 1 \pmod{2}$, and therefore $d_{k+3} < d_{k+1}$.
Hence, the runs $R_{k+1}$ and $R_{k+2}$ must be merged together.

Below, let us denote by $\oR$ the parent of every run $R$.
Since $\orr_{k+1} = 17 > r_k > r_{k+1}$, it follows that
$\od_{k+1} \leqslant d_k \leqslant d_{k+3}$.
Then, since $d_{k+3} < d_{k+1} = \od_{k+1} + 1$, this even proves that
$\od_{k+1} = d_k = d_{k+3}$.
Similarly, for all $i \leqslant k-1$ or $i \geqslant k+4$,
we know that $r_i > \orr_{k+1}$, and thus that $d_i \leqslant \od_{k+1}$.
Consequently, the run $R_{k+3}$ must be
merged with either $\oR_{k+1}$ or $R_{k+4}$.

In both cases, we have $\orr_{k+3} \leqslant r_{k+3} + \max\{\orr_{k+1},r_{k+4}\} = 29$, and therefore $r_i > \orr_{k+3}$ for all
$i \leqslant k-1$ or $i \geqslant k+5$, thereby proving that
$d_i \geqslant \od_{k+3}$.
Consequently, the only runs that lie at depth $\od_{k+1}$ are $R_k$, $\oR_{k+1}$, $R_{k+3}$, and possibly $R_{k+4}$. Since there must be an even number of such
runs, it follows at once that all four runs lie at depth $\od_{k+1}$ and that
$\A$ merges $R_k$ with $\oR_{k+1}$, and $R_{k+3}$ with $R_{k+4}$.

Let us further prove that $\A$ merges $R_{k-1}$ with $\oR_k$,
and $\oR_{k+4}$ with $R_{k+5}$. Indeed, both runs $\oR_k$ and $\oR_{k+4}$
lie at depth $\od_k$, and since $\orr_k = 2^5+1 > r_{k-1} = r_{k+5} > 2^5-3 = \orr_{k+4}$, the runs $R_{k-1}$ and $R_{k+5}$ also lie at depth $\od_k$.
Hence, the run $\oR_{k+4}$ will be merged with either $\oR_k$ or $R_{k+5}$.
In both cases, we have $\overline{\orr}_{k+4} \leqslant
\orr_{k+4} + \max\{\orr_k,r_{k+5}\} < 2^6 \leqslant r_i$ for all $i \leqslant k-2$ or $i \geqslant k+6$.
Consequently, the runs $R_{k-1}$, $\oR_k$, $\oR_{k+4}$ and $R_{k+5}$
are the only runs at depth $\od_k$, and $\A$ merges
$R_{k-1}$ with $\oR_k$, and $\oR_{k+4}$ with $R_{k+5}$.

Repeating the same arguments verbatim, we prove that
$\A$ merges $R_{k-1-i}$ with $\oR_{k-i}$, and $\oR_{k+4+i}$ with $R_{k+5+i}$,
for all $i = 0,1,\ldots,k-2$.
Eventually, we end up with the runs
$R_0$, $\oR_1$, $\oR_{2k+2}$ and $R_\final$, whose lengths are
$r_0 = 2^{k+3}$, $\orr_1 = 2^{k+3}+1$, $\orr_{2k+2} = 2^{k+3}-1$ and
$r_\final = 3 \times 2^{k+2}+2$.
The only optimal way to merge these runs is to merge $R_0$ with $\oR_1$,
and $\oR_{2k+2}$ with $\oR_\final$.
This completes the proof of Claim~\ref{claim-2}.

\smallskip

Let us now further assume that $\A$ is length-$(\infty,k)$-aware.
When scanning either array $\bA^{\pm}$, 
and since both arrays $\bA^-$ and $\bA^+$ have the same length,
$\A$ cannot distinguish between
them until it discovers the run $R_{2k+3}$ or $R_\final$.
In particular, it must wait until discovering that run
before it can merge the run $R_{k+3}$.

Imagine now that $\A$ is sorting the array $\bA^+$.
When $\A$ merges the run $R_{k+3}$, the $k$
runs $R_{k+4},\ldots,R_{2k+2},R_\final$
have already been pushed onto the stack, and none of them can be
merged before $R_{k+3}$ is merged.
Thus, if the algorithm $\A$ is to merge $R_{k+3}$, it could in fact not be
length-$(\infty,k)$-aware, which completes the proof.
\end{proof}

Unsurprisingly, this result can be strengthened dramatically
in the case of length-$(\infty,2)$-aware algorithms,
which are not $\varepsilon$-optimal for any $\varepsilon$.

\begin{lemma}\label{lem:konig}
Consider the following dynamic system.
Starting with one empty stack $\S$, we successively perform operations
of the following type: either (i)~if $\S$ contains at least two elements,
we may remove the top element of the stack, or
(ii)~if we have already pushed the integers $0,1,\ldots,\ell-1$ onto $\S$
(some of which may have been removed),
we may push the integer $\ell$ on the top of the stack.

Then, for all integers $h$ and all functions
$f : \mathbb{Z}_{\geqslant 0} \mapsto \mathbb{Z}_{\geqslant 0}$, there exists
an integer $m_{f,h}$ such that, when the integer $m_{f,h}$ is pushed onto the stack,
either the stack $\S$ has been of height $h$ at some point, or
some integer $k$ has been the top element of the stack after $f(k)$ operations
of type (i).
\end{lemma}

\begin{proof}
Since, at every step, we have the choices between options~(i)
and~(ii), the set of executions of the system can be seen as an
infinite tree in which every node has two children, one per option.
Let us restrict this tree to the subtree $\mathcal{T}$
containing those (finite or infinite) executions where
the stack is always of height smaller than $h$, and where
no integer $k$ ever becomes the top element after $f(k)$ or more
operations of type (i).

First, we show that every branch (i.e., execution) of $\mathcal{T}$
is finite. Indeed, consider some infinite execution where the stack
height is always smaller than $h$. At most $h-1$ integers 
will stay forever on the stack after they have been pushed onto it,
and $0$ is one of these integers. Hence, let $k$ be the largest
such integer. Every integer ever placed just above $k$ must have
been removed from the stack at some point, hence $k$
must have been the top element of the stack infinitely many times. 
Therefore, this execution is not a branch of $\mathcal{T}$.

Finally, since the branching degree of $\mathcal{T}$ is $2$,
and since its branches are all finite,
König's lemma proves that $\mathcal{T}$ itself is finite.
Defining $m_{f,h}$ as the maximal length of $\mathcal{T}$'s
branches completes the proof.
\end{proof}

\begin{proposition}\label{pro:never-epsilon-optimal}
Let $\A$ be a length-$(\infty,2)$-aware stable merge sort algorithm.
The worst-case merge cost of $\A$ is bounded from below by
$\omega(n (\H+1))$.
In particular, $\A$ is not $\varepsilon$-optimal 
for any real number~$\varepsilon$.
\end{proposition}

\begin{proof}
For the sake of contradiction, 
we assume in the entire proof that there exist an integer $\mathbf{z}$
and a length-$(\infty,2)$-aware algorithm $\A$
whose merge cost is bounded from above by $\mathbf{z} n (\H+1)$.
Then, for all integers $k \leqslant \ell$, let
$\bA^{k,\ell}$ be an array that decomposes into $k+1$ runs 
$R_0,R_1,\ldots,R_k$ of
respective lengths $2^{\ell-1},2^{\ell-2},\ldots,2^{\ell-k}$ and
$2^\ell+2^{\ell-k}$: the array $\bA^{k,\ell}$ has length $2^{\ell+1}$.

The entropy of any array $\bA^{k,\ell}$ is defined by
\[\H^{k,\ell} = \sum_{i=2}^{k+1} \frac{1}{2^i}\log_2(2^i) -
\frac{1+2^{-k}}{2}\log_2\left(\frac{1+2^{-k}}{2}\right) 
< \sum_{i \geqslant 2} \frac{i}{2^i} + 1 = \frac{5}{2},\]
and therefore the cost of those merges used by $\A$ for sorting
$\bA^{k,\ell}$ is smaller than $7 \mathbf{z} 2^\ell$.
Moreover, if the stack of $\A$ is of height $h$
when the last run of $\bA^{k,\ell}$ is discovered,
then that run will take part in $h$ merges,
for a total cost of $2^\ell h$ at least.
It follows that $h < 7 \mathbf{z}$.

However, once the integer $\ell$ is fixed, and
provided that it is executed on some array $\bA^{k,\ell}$,
the algorithm $\A$ cannot distinguish 
between the arrays $\bA^{0,\ell},\ldots,\bA^{\ell,\ell}$
until it discovers their last run (or the second last run of $\bA^{\ell,\ell}$).
In particular, if, when treating some array $\bA^{k,\ell}$,
and just before pushing its $i$\textsuperscript{th} run (with $i \leqslant k$),
the stack of $\A$ turns out to be of height $h \geqslant 7 \mathbf{z}$,
then $\A$ might as well discover that it was, in fact, treating the array
$\bA^{i,\ell}$, contradicting the previous paragraph.
Therefore, the stack of $\A$ may \emph{never} be of height $7 \mathbf{z}$ or more.

At the same time, if the elements of some run $R_j$ take part in
$7 \mathbf{z} 2^{j+1}$ merges, then of course these merges have a total cost
of $7 \mathbf{z} 2^\ell$ at least. Hence, the elements of
every run $R_j$ can take part in at most $f(j)$ merges, where $f(j) = 7 \mathbf{z} 2^{j+1}$.
Furthermore, our stack follows exactly the dynamics described in
Lemma~\ref{lem:konig}, where choosing the option (i)
means that we merge the top two elements of the stack:
if the element $j$ becomes the new top element after such an operation,
this means that we have just merged the elements of $R_j$ (among others).
It follows from Lemma~\ref{lem:konig} that $\ell \leqslant m_{f,7 \mathbf{z}}$.

Consequently, and when $\ell$ is large enough, we know that there must exist
arrays on which the merge cost of $\A$ is at least $\mathbf{z} n (\H+1)$.
In particular, since \PoS would have sorted $\A$ for a cost of $n(\H + 2)$,
it follows that $\A$ is not $(\mathbf{z}/2-1)$-optimal.
\end{proof}

\subsection{Approximability bounds}\label{sec:approximable}

We focus now on one \emph{approximability} result,
which states that there exists
some $k$-algorithm that matches some approximabilty bound.
Furthermore, this result has a constructive proof,
which consists in providing an effective algorithm
and proving that this algorithm indeed matches the approximability bound.

\begin{theorem}\label{thm:eta-optimal}
Let $k \geqslant 8$ be an integer, and let
\[\eta_k = (\Delta+7) / \log_2((k-3)/4).\]
There exists a $k$-aware
sorting algorithm that is $\eta_k$-optimal.
\end{theorem}

Note that, if $k \geqslant 3$, then
Theorems~\ref{thm:complexity-nH+D} and~\ref{thm:lower-bound-merge-cost}
already prove that \cASS, which is a $3$-aware (and thus a $k$-aware) algorithm,
is also $\Delta$-optimal. In particular, numerical
computations show that $\Delta \leqslant \eta_k$ when $k \leqslant 59$,
which already proves Theorem~\ref{thm:eta-optimal} in that case.
Furthermore, and although the constants
$\theta_k$ and $\eta_k$ become arbitrarily small
when $k \to \infty$, there is still an exponential gap between these constants.

Our proof consists in showing that the parameterised algorithm~\cASSm below,
which is visibly $(2\kappa+2)$-aware, is
$\eta_{2\kappa+3}$-optimal.

\begin{algorithm}[b!]
\begin{small}
\DontPrintSemicolon
\SetArgSty{texttt}
\SetKwInOut{Input}{Input}
\SetKwInput{KwData}{Note}
\SetKwProg{Pn}{Function}{:}{}
\SetKwFunction{MustMerge}{\mustMerge}
\SetKwFunction{DoMerge}{\doMerge}
\Input{\hspace*{1.95mm}An array to $A$ to sort, integer parameter $\kappa$}
\KwResult{The array $A$ is sorted into a single run, which remains on the 
stack.}
\BlankLine
\KwData{We \hfill denote \hfill the \hfill height \hfill of \hfill the \hfill stack \hfill $\S$ \hfill by \hfill $h$, \hfill and \hfill its \hfill $i$\textsuperscript{th} \hfill bottom-most \hfill run \hfill (for \hfill $1 \hfill \leqslant \hfill i \hfill \leqslant \hfill h$) \hfill by
$R_i$. \hfill The \hfill length \hfill of \hfill $R_i$ \hfill is \hfill denoted \hfill by \hfill $r_i$, \hfill and \hfill we \hfill set \hfill $\ell_i \hfill = \hfill \lfloor \hfill \log_2(r_i) \hfill \rfloor$: \hfill the \hfill parameter \hfill $c$ \hfill of
\cASS \hfill is \hfill implicitly \hfill set \hfill to \hfill $c \hfill = \hfill 1$. \hfill Whenever \hfill two \hfill successive \hfill runs \hfill of \hfill $\S$ \hfill are
merged, \hfill they \hfill are \hfill replaced, \hfill in \hfill $\S$, \hfill by \hfill the \hfill run \hfill resulting \hfill from \hfill the \hfill merge. \hfill In \hfill practice, \hfill in \hfill $\S$, \hfill each
run is represented by a pair of pointers to its first and last entries.}
\BlankLine
$\rundecomp \gets $ the run decomposition of $A$\label{alg:m:init}\;
$\S \gets $ an empty stack\;
\While(\tcp*[f]{main loop}){\true}{
 \If{\textrm{$h = 2\kappa+1$, $\rundecomp \neq \emptyset$
  and $\mustMerge$}\label{alg:ASSm:0}}
  {$\doMerge$\label{alg:ASSm:1}}
 \ElseIf{\textrm{$h \geqslant 2\kappa+2$ and $\ell_{h-2} \leqslant \max\{\ell_{h-1},\ell_h\}$}}
  {merge the runs $R_{h-2}$ and $R_{h-1}$\label{alg:ASSm:2}}
 \ElseIf{$\rundecomp \neq \emptyset$}
  {remove a run $R$ from $\rundecomp$ and push $R$ onto $\S$\label{alg:ASSm:3}}
 \ElseIf{$h \geqslant 2\kappa+2$}
  {merge the runs $R_{h-1}$ and $R_h$\label{alg:ASSm:4}}
 \ElseIf(\tcp*[f]{phase 2}){$h \geqslant 2$}
  {perform the first merge prescribed by \MinSS\label{alg:ASSm:brute-force-merge}}
 \Else
  {break}
}
\drawline
\setcounter{AlgoLine}{15}
\Pn(\tcp*[f]{called only with $h = 2\kappa+1$}\label{alg:lowstack}){\MustMerge}{
 $r_\llarge \gets \lfloor (r_1+r_2+\ldots+r_{2\kappa-1}) / \kappa \rfloor$\;
 $\ell_\llarge \gets \lfloor \log_2(r_\llarge)\rfloor$ \;
 \For{$i \gets 1,2,\ldots,2\kappa-1$\label{alg:m:cs2}}{
  \If{\textrm{$\ell_i \leqslant \max\{\ell_{i+1},\ell_{i+2}\}$ and
  $\max\{\ell_i,\ell_{i+1}\} \leqslant \ell_\llarge$}}
   {\return{} \true\label{alg:mustMerge:true}}
 }
 \return{} \false\label{alg:mustMerge:false}
}
\drawline
\setcounter{AlgoLine}{22}
\Pn(\tcp*[f]{called only with $h = 2\kappa+1$}\label{alg:lowstack}){\DoMerge}{
 $r_\llarge \gets \lfloor (r_1+r_2+\ldots+r_{2\kappa-1}) / \kappa \rfloor$\;
 $\ell_\llarge \gets \lfloor \log_2(r_\llarge)\rfloor$ \;
 \For{$i \gets 1,2,\ldots,2\kappa-1$\label{alg:m:cs2}}{
  \If{\textrm{$\ell_i \leqslant \max\{\ell_{i+1},\ell_{i+2}\}$ and
  $\max\{\ell_i,\ell_{i+1}\} \leqslant \ell_\llarge$}}
   {merge the runs $R_i$ and $R_{i+1}$\;\label{alg:doMerge:merge}
    break}
 }
}
\end{small}
\caption{\cASSm and its auxiliary functions\label{alg:ASSm}}
\end{algorithm}

To do this, let us already observe that, 
as soon as $\rundecomp = \emptyset$ and $h \leqslant 2\kappa+1$,
the algorthm \cASSm will keep performing the merges
of line~\ref{alg:ASSm:brute-force-merge} until $h = 1$, thereby following the
merge policy of \MinSS itself.
Therefore, we split \cASSm in two phases:
\emph{phase~1} consists in the (merge or push)
operations performed while $\rundecomp \neq \emptyset$ or 
$h \geqslant 2\kappa+2$,
in lines~\ref{alg:ASSm:1} to~\ref{alg:ASSm:4} (or~\ref{alg:doMerge:merge})
of Algorithm~\ref{alg:ASSm}, whereas
\emph{phase~2} consists in the merge operations performed
while $h \leqslant 2\kappa+2$ and $\rundecomp = \emptyset$, in 
line~\ref{alg:ASSm:brute-force-merge}.

Unsurprisingly, the dynamics of these two phases will be quite different:
phase~1 keeps the flavour of \cASS,
whereas phase~2 is just a brutal application of \MinSS.
Consequently, below, we will both need to study the merge costs of each phase
and to find a suitable characterisation of the stack $\S$ once phase~1 is terminated
and phase~2 begins.

Due to its similarity to phase~1, let us further focus on the dynamical properties of
\cASS. We will then draw a parallel between these properties and
those of \cASSm, thereby obtaining results about the latter algorithm.
In order to do so, we will reuse some key tools already introduced in
previous sections: the notions of \emph{state}, \emph{merge point}, \emph{successor}
and \emph{merge tree}.
Since we will build further notions and results on top of these objects,
and in order to make that construction easier to follow,
let us first recall some definitions and results.

\defstate*

\defsucc*

\promerge*

\defmergetree*

Proposition~\ref{pro:first-merge} clearly marks \emph{states} as a crucial object for
studying the dynamics of \cASS, since they seem to capture the right amount of information
manipulated by the algorithm at a given point in time.
On the other hand, \emph{merge trees} are generic objects, which may be constructed
for every (stable) merge policy, and alluw us to capture at once the entire list
of merges performed by the algorithm.
A natural further step is then to bind both of these objects:
this is the aim of \emph{borders}.

\begin{definition}\label{def:border}
Let $\R = (R_1,\ldots,R_\rho)$
be a sequence of runs, and let
$\T$ be the merge tree induced by \cASS on $\R$
(we may omit mentioning \cASS if the context is clear,
thereby abusively referring to the merge tree induced by $\R$).
A \emph{border} of $\T$ is a maximal set $\B$
of pairwise incomparable (for the strict ancestor relation) nodes of $\T$, i.e., 
a maximal set $\B$ such that no
strict ancestor of a node in $\B$ belongs to $\B$.
We call the set of ancestors of nodes in $\B$, including nodes in $\B$
themselves, the \emph{ancestor sub-tree} of $\B$ in $\T$.
\end{definition}

Every border or, more generally, 
every set of pairwise incomparable nodes of a merge tree $\T$,
is naturally endowed with a left-to-right ordering.
Hence, we will often identify a border $\B$
with the unique sequence of runs $(R'_1,\ldots,R'_k)$
such that each $R'_i$ lies to the left of $R'_{i+1}$
and such that $\{R'_1,\ldots,R'_k\} = \B$.
In particular, every state encountered during an execution of \cASS is
a border of $\R$.
Building on this identification, let us focus on the merge
tree induced by $\B$.

\begin{theorem}\label{thm:stable-border}
Let $\T$ be the merge tree induced by
a sequence of runs $\R$, and let $\R'$ be a border of $\T$.
The merge tree induced by $\R'$ coincides with the ancestor sub-tree of 
$\R'$ in $\T$.
\end{theorem}

\begin{proof}
Let $(R_1,\ldots,R_\rho)$
be our sequence of runs $\R$, and let
$\rho = |\R|$ and $s = |\R'|$.
Then, if $\rho \geqslant 2$, let $k$ be the merge point of $\R$,
and let $\oR$ be the run obtained by merging $R_k$ and $R_{k+1}$,
i.e., their parent in the tree $\T$;
if $\rho = 1$, we just set $k = 0$.
Since $k \leqslant \rho$, we prove Theorem~\ref{thm:stable-border}
by induction on the triple $(\rho,\rho - s, k)$.

First, if $\rho = s$, then
$\R = \R'$, and thus
the result is immediate.

Second, let us assume that $\rho = s+1$.
In that case, there exists
an integer $i \leqslant \rho-1$ such that
$\R' = \R \setminus \{R_i,R_{i+1}\} \cup
\{R'\}$, where $R'$ is the run obtained by
merging $R_i$ and $R_{i+1}$, i.e., their parent in $\T$.
Observe that either $i = k$ or $i \geqslant k+2$:
\begin{itemize}
\item If $i = k$, Proposition~\ref{pro:first-merge} proves that
merging the runs $R_i$ and $R_{i+1}$ into one single run $R'$ is
the first merge operation performed while applying \cASS on the sequence $\R$.
Thus, for all $a \geqslant 1$, the $a$\textsuperscript{th} merge performed
while executing \cASS on $\B$
is also the $(a+1)$\textsuperscript{th} merge performed while
executing \cASS on $\R$.
This completes the proof in that case.

\item If $i \geqslant k+2$, let also
$\overline{\R} = \R \setminus 
\{R_k,R_{k+1}\} \cup \{\oR\}$ and
$\overline{\R}' = \R \setminus 
\{R_k,R_{k+1},R_i,R_{i+1}\} \cup \{\oR,R'\}$
be borders of $\T$.
Let $\T'$, $\overline{\T}$ and 
$\overline{\T}'$ be the
merge trees induced by 
$\R'$, $\overline{\R}$ and 
$\overline{\R}'$.
We will show that they are all ancestor sub-trees of $\T$,
as illustrated in Figure~\ref{fig:sub-trees}.

The induction hypothesis proves that
$\overline{\T}$ is the ancestor sub-tree of $\overline{\R}$ in $\T$.
Hence, the runs $R_i$ and $R_{i+1}$ are siblings in both
trees $\T$ and $\overline{\T}$. Consequently,
$\overline{\R}'$ is also a border of $\overline{\T}$, and therefore
$\overline{\T}'$ is the ancestor sub-tree of $\overline{\R}'$ in $\overline{\T}$
(and thus in $\T$ as well).

Then, Proposition~\ref{pro:first-merge} proves that
the runs $R_k$ and $R_{k+1}$ are siblings in $\T'$.
It follows that $\overline{\R}'$ is a border of $\T'$ and,
by induction hypothesis, that $\overline{\T}'$ is the
ancestor sub-tree of $\overline{\R}'$ in $\T'$.

Hence, $\overline{\T}'$ is a sub-tree of both trees $\T$ and $\T'$.
Then, since the only nodes of $\T'$ that do not belong to $\T'$ are
siblings in $\T$ (these are the leaves $R_i$ and $R_{i+1}$),
it means that $\T'$ is itself a sub-tree of $\T$, as desired.
\end{itemize}

\begin{figure}[ht]
\begin{center}
\begin{tikzpicture}[scale=0.5]
\fill[fill=black!10]
(-1.4,-1.4) -- (-1.4,1.4) -- (11.5,7) -- (24.4,1.4) -- (24.4,-1.4) -- 
(10.4,-1.4) -- (10.4,1.4) -- (3.6,1.4) -- (3.6,-1.4) -- cycle;
\fill[pattern=flexible hatch,hatch distance=12pt,hatch thickness=2pt,pattern color=black!35]
(-1.4,-1.4) -- (-1.4,1.4) -- (11.5,7) -- (24.4,1.4) -- (24.4,-1.4) -- 
(19.4,-1.4) -- (19.4,1.4) -- (12.6,1.4) -- (12.6,-1.4) -- cycle;
\draw[ultra thick,draw=black]
(-1.4,-1.4) -- (-1.4,1.4) -- (11.5,7) -- (24.4,1.4) -- (24.4,-1.4) -- cycle;
\draw[ultra thick] (5,0) -- (7,2.8) -- (9,0);
\draw[ultra thick] (14,0) -- (16,2.8) -- (18,0);
\draw[very thick,fill=white] (7,2.8) circle (1);
\node[anchor=north] at (7,3.4) {$\oR$};
\draw[very thick,fill=white] (16,2.8) circle (1);
\node[anchor=north] at (16,3.4) {$R'$};
\draw[very thick,fill=white] (0,0) circle (1);
\node[anchor=north] at (0,0.6) {$R_1$};
\node at (2.5,0) {\contour{black}{\textbf{\ldots}}};
\draw[very thick,fill=white] (5,0) circle (1);
\node[anchor=north] at (5,0.6) {$R_k$};
\draw[very thick,fill=white] (9,0) circle (1);
\node[anchor=north] at (9,0.6) {$R_{k+1}$};
\node at (11.5,0) {\contour{black}{\textbf{\ldots}}};
\draw[very thick,fill=white] (14,0) circle (1);
\node[anchor=north] at (14,0.6) {$R_p$};
\draw[very thick,fill=white] (18,0) circle (1);
\node[anchor=north] at (18,0.6) {$R_{p+1}$};
\node at (20.5,0) {\contour{black}{\textbf{\ldots}}};
\draw[very thick,fill=white] (23,0) circle (1);
\node[anchor=north] at (23,0.6) {$R_\rho$};

\draw[very thick] (-0.5,-4.1) -- (-0.5,-3.3) -- (1,-2.5) -- 
(2.5,-3.3) -- (2.5,-4.1) -- cycle;
\node[anchor=south] at (1,-5.3) {$\T$};

\draw[very thick,fill=black!10] (6.5,-4.1) -- (6.5,-3.3) -- (8,-2.5) -- 
(9.5,-3.3) -- (9.5,-4.1) -- (7.8,-4.1) -- (7.8,-3.6) -- (7.1,-3.6) -- (7.1,-4.1) -- cycle;
\node[anchor=south] at (8,-5.3) {$\overline{\T}$};

\draw[very thick,pattern=flexible hatch,hatch distance=6pt,hatch thickness=1pt,
pattern color=black!35]
(13.5,-4.1) -- (13.5,-3.3) -- (15,-2.5) --  (16.5,-3.3) -- (16.5,-4.1) --
(15.9,-4.1) -- (15.9,-3.6) -- (15.2,-3.6) -- (15.2,-4.1) -- cycle;
\node[anchor=south] at (15,-5.3) {$\T'$};

\draw[very thick,fill=black!10]
(20.5,-4.1) -- (20.5,-3.3) -- (22,-2.5) -- 
(23.5,-3.3) -- (23.5,-4.1) -- (22.9,-4.1) -- (22.9,-3.6) -- (22.2,-3.6) -- (22.2,-4.1) --
(21.8,-4.1) -- (21.8,-3.6) -- (21.1,-3.6) -- (21.1,-4.1) -- cycle;
\draw[very thick,pattern=flexible hatch,hatch distance=6pt,hatch thickness=1pt,
pattern color=black!35]
(20.5,-4.1) -- (20.5,-3.3) -- (22,-2.5) -- 
(23.5,-3.3) -- (23.5,-4.1) -- (22.9,-4.1) -- (22.9,-3.6) -- (22.2,-3.6) -- (22.2,-4.1) --
(21.8,-4.1) -- (21.8,-3.6) -- (21.1,-3.6) -- (21.1,-4.1) -- cycle;
\node[anchor=south] at (22,-5.3) {$\overline{\T}'$};
\end{tikzpicture}
\end{center}
\vspace{-\baselineskip}
\caption{Embedding the trees $\T'$, $\overline{\T}$ and $\overline{\T}' = \T' \cap \overline{\T}'$ as sub-trees
of $\T$.\label{fig:sub-trees}}
\end{figure}

Finally, we assume that $\rho \geqslant s+2$.
Let $R$ be an internal node of $\T$
that belongs to $\R'$, and let $R_\oplus$ and
$R_\ominus$ be its children.
The set $\R'' = \R' \setminus \{R\} \cup
\{R_\oplus,R_\oplus\}$ is a border of $\T$ of
cardinality $s+1$.
Let also $\T'$ and $\T''$ be the merge trees
respectively induced by $\R'$ and $\R''$.
Then, $\T''$ is the ancestor sub-tree of
$\R''$ in $\T$, and thus
$\T'$ is the
ancestor sub-tree of $\R'$ in both $\T''$
and $\T$.
\end{proof}

Theorem~\ref{thm:stable-border} provides us the following
remarkable \emph{stability} property
of \cASS.

\begin{corollary}
\label{cor:stable-border}
Let $\A$ be some stable natural merge sort algorithm, which consists in
(i)~applying some (arbitrary) sequence of merges that would have been performed by \cASS,
and (ii)~applying \cASS on the resulting state.
Then, $\A$ simply performs the same merges as \cASS, although in a possibly different order.
\end{corollary}

Given that these notions were particularly
well suited to the study of the dynamics of \cASS,
we shall now adapt these tools,
and the results that had led to Proposition~\ref{pro:first-merge},
to the analogous case of \cASSm.

However, since phase~1 of \cASSm is meant to be a variant of \cASS where
small runs only would be merged, let us first present
the notion of what these runs are, as well as
how we adapt the notion of \emph{successor}.

\begin{definition}\label{def:k-successor}
Let $\R = (R_1,\ldots,R_\rho)$ be a sequence of runs of 
length $\rho \geqslant 2\kappa+1$.
Recall that an integer $x$ is \emph{dominated} in the sequence $\R$ if
$1 \leqslant x \leqslant \rho-1$ and $\ell_x \leqslant \max\{\ell_{x+1},\ell_{x+2}\}$,
with the convention that $\ell_{\rho+1} = \infty$.

Let us also define the integers
\[r_\llarge = 
\lfloor (r_1+r_2+\ldots+r_{2\kappa-1}) / \kappa \rfloor\]
and $\ell_\llarge = 
\lfloor \log_2(r_\llarge) \rfloor$.
We say that an integer $x$ is \emph{low} in the sequence $\R$ if
$1 \leqslant x \leqslant \rho-1$ and if $\max\{\ell_x,\ell_{x+1}\} \leqslant \ell_\llarge$
(we may omit mentioning $\R$ when the context is clear).

Finally, let $k$ be the least integer
that is both low and dominated, if such an integer exists;
it not, we simply set $k = 0$.
The integer $k$ is called the \emph{$\kappa$-merge point} of $\R$.
Furthermore, if $k \neq 0$, we call \emph{$\kappa$-successor} of $\R$, and denote by
$\succl(\R)$, the sequence of runs
\[(R_1,\ldots,R_{k-1},R',R_{k+2},\ldots,R_\rho),\]
where $R'$ is the run obtained by merging $R_k$ and $R_{k+1}$;
if $k = 0$, the $\kappa$-successor of $\R$ is not defined.
\end{definition}

\begin{lemma}\label{lem:few-large}
Let $\R = (R_1,\ldots,R_\rho)$ be a sequence of runs of length $\rho \geqslant 2\kappa-1$.
There exists at least one integer $x \leqslant 2\kappa-1$ that is low in $\R$.
\end{lemma}

\begin{proof}
 Let $s = r_1+r_2+\ldots+r_{2\kappa-1}$, and consider the set
 \[X = \{x \colon 1 \leqslant x \leqslant 2 \kappa-1 \text{ and } \ell_x > \ell_\llarge\}.\]
 For all $x \in X$, it comes that $\log_2(r_x) > \log_2(r_\llarge)$, i.e., that
 $r_x > r_\llarge$, and therefore that $r_x > s / \kappa$.
 It follows that \[s \geqslant \sum_{x \in X} r_x > s |X| / \kappa,\]
 which means that $|X| \leqslant \kappa - 1$.
 Since every large integer belongs to the set $X \cup \{x \colon x + 1 \in X\}$,
 there can be at most $2(\kappa-1)$ large integers. This completes the proof.
\end{proof}

\begin{lemma}\label{lem:low-then-low}
Let $\S = (R_1,\ldots,R_h)$ be a stack of height $h = 2\kappa+1$ encountered
while \cASSm calls the function \mustMerge, and
let $\bS = (R_1,\ldots,R_\rho)$ be the associated state.
If there exists a low integer $x \leqslant 2\kappa-1$
such that $x+1$ is not low, then
that call to \mustMerge returns \true.
\end{lemma}

\begin{proof}
If $x$ is low and $x+1$ is not low, then
$\ell_x \leqslant \max\{\ell_x,\ell_{x+1}\} <
\ell_\llarge \leqslant \max\{\ell_{x+1},\ell_{x+2}\}$.
Therefore, $x$ is both low and dominated, and \mustMerge necessarily
returns \true.
\end{proof}

\begin{lemma}\label{lem:2k-low}
Let $\S = (R_1,\ldots,R_h)$ be a stack encountered
during phase~1 of \cASSm, let $\bS = (R_1,\ldots,R_\rho)$ be the associated state,
and let $k$ be the $\kappa$-merge point of $\bS$.
If $h \geqslant 2\kappa+2$, then 
(i)~the integer $k$ is not equal to any of the integers
$1,2,\ldots,2\kappa-1$;
(ii)~the integer $2\kappa-1$ is low in $\bS$; and~(iii) it holds that:
\begin{equation}
\ell_{2\kappa-1} > \ell_{2\kappa} > \ldots > \ell_{h-3} > \max\{\ell_{h-2},\ell_{h-1}\}.
\label{ell-incr-2}
\end{equation}
\end{lemma}

\begin{proof}
The proof is done by induction.
First, if $h \leqslant 2\kappa+1$, there is nothing to prove:
this case occurs, in particular, when the algorithm starts.
Now, consider some stack $\S = (R_1,\ldots,R_h)$ that
is updated into a new stack
$\oS = (\oR_1,\ldots,\oR_\oh)$ of height $\oh \geqslant 2\kappa+2$, either
by merging the runs $R_{h-2}$ and $R_{h-1}$
or by pushing the run $\oR_\oh$,
and let $\obS$ be the state associated to $\oS$.
Let also $k$ and $\ok$ be the respective $\kappa$-merge points of $\bS$ and $\obS$,
and assume that $\S$ satisfies Lemma~\ref{lem:2k-low}:
\begin{itemize}
 \item If $h = 2\kappa+1$, then $\oh = 2\kappa+2$ and $\oS$
 was obtained by pushing $\oR_\oh$, so that $\obS = \bS$.
 This run push must have been preceded by a call to the function \mustMerge,
 which consisted in checking that $k$ is not equal to any of the integers
 $1,2,\ldots,2\kappa-1$ before returning the value \false.
 
 Since Lemma~\ref{lem:few-large} states
 that there exists a low integer $x \leqslant 2\kappa-1$,
 Lemma~\ref{lem:low-then-low} then proves that $2\kappa-1$ is low.
 Then, since $k \neq 2\kappa-1$, 
 it must be that $2\kappa-1$ is not dominated, i.e., that
 $\ol_{2\kappa-1} > \max\{\ol_{2\kappa},\ol_{2\kappa+1}\}$,
 which means that $\oS$ satisfies~\eqref{ell-incr-2}.
 
 \item If $h \geqslant 2\kappa+2$ and $\oS$
 was obtained by pushing $\oR_\oh$,
 then $\oh = h+1$ and $\obS = \bS$.
 It already follows that $2\kappa-1$ is low,
 that $\ok = k \notin \{1,2,\ldots,2\kappa-1\}$
 and that $\ol_{2\kappa-1} > \ol_{2\kappa} > \ldots > \ol_{\oh-3}$.
 Finally, since \cASSm triggered a push operation instead of 
 a merge operation,
 it must be the case that $\ell_{h-2} > \max\{\ell_{h-1},\ell_h\}$,
 i.e., that $\ol_{\oh-3} > \max\{\ol_{\oh-2},\ol_{\oh-1}\}$.
 This proves that $\oS$ satisfies~\eqref{ell-incr-2}.

 \item If $\oS$ was obtained by merging the runs
 $R_{h-2}$ and $R_{h-1}$, then
 $\oh = h-1$ and $\oR_i = R_i$ for all $i \leqslant \oh-2$,
 and in particular for all $i \leqslant 2\kappa+1$.
 It follows that~(i) each integer $x \leqslant \oh-4$
 is dominated in $\obS$ if and only if $x$ is dominated in $\bS$,
 and that~(ii) $\ell_\llarge = \ol_\llarge$, thereby proving that
 each integer $x \leqslant \oh-3$
 is low in $\obS$ if and only if $x$ is low in $\bS$
 Since $\oh \geqslant 2\kappa+2$, it already follows that
 $\ok \notin \{1,2,\ldots,2\kappa-2\}$ 
 and that $2\kappa-1$ is low in $\obS$.
 
 Finally, the inequalities 
 $\ell_{2\kappa-1} > \ell_{2\kappa} > \ldots > \ell_{h-3}$
 immediately rewrite as
 $\ol_{2\kappa-1} > \ol_{2\kappa} > \ldots > \ol_{\oh-2}$.
 Then, since the run $\oR_{\oh-1}$ results from the merge between
 $R_{h-2}$ and $R_{h-1}$,
 Lemma~\ref{lem:l-small-increase} proves that
 $\ol_{\oh-1} \leqslant \max\{\ell_{h-2},\ell_{h-1}\}+1$.
 This means that $\ol_{\oh-2} > \ol_{\oh-1}-1$ or, equivalently, that
 $\ol_{\oh-2} \geqslant \ol_{\oh-1}$.
 We conclude that $\oS$ satisfies~\eqref{ell-incr-2},
 from which it also follows
 that $2\kappa-1$ is not dominated in $\obS$,
 thereby proving that $\ok \neq 2\kappa-1$.
\end{itemize}
\end{proof}

\begin{proposition}\label{pro:first-merge-2}
Let $\bS$ and $\obS$ be two consecutive states encountered
during an execution of the phase~1 of \cASSm.
Then, $\obS = \succl(\bS)$.
\end{proposition}

\begin{proof}
Let $m$ be the merge operation that transforms $\bS$ into $\obS$.
Let $\S = (R_1,\ldots,R_h)$ be the stack just before $m$ is performed.
Let $k$ be the $\kappa$-merge point of $\bS$,
and let $R_i$ and $R_{i+1}$ be the runs that are merged when $m$ is performed. We shall prove that $i = k$.

If $h = 2\kappa+1$, then $m$ was triggered by a call to \mustMerge.
During that call, the function \mustMerge scans the integers $1,2,\ldots,2\kappa-1$
one by one until it finds a low dominated integer is identified.
That integer is equal to $i$, and by construction it is also equal to $k$.

Then, if $h \geqslant 2\kappa+2$, we know that $i = h-2$.
Furthermore, in that case, Lemma~\ref{lem:2k-low} states that
either $k = 0$ or $k \geqslant 2\kappa$, and that the 
inequalities~\eqref{ell-incr-2} are satisfied.
In particular, $k$ cannot be equal to any of the integers $2\kappa,\ldots,h-3$, which are not even dominated,
whereas $h-2$ is both low and dominated.
This means that $k = h-2$, which completes the proof.
\end{proof}

Due to the characterisation provided by Proposition~\ref{pro:first-merge-2},
we focus now on proving that every merge performed in phase~1
of \cASSm would also have been performed by \cASS.

\begin{lemma}\label{lem:first-kappa-merge}
Let $\R = (R_1,\ldots,R_\rho)$ be a sequence of runs
with $\kappa$-merge point $k \neq 0$.
At some point when sorting that sequence, the algorithm \cASS will
merge the runs $R_k$ and $R_{k+1}$.
\end{lemma}

\begin{proof}
If $k = 1$, then $k$ is the least dominated integer,
and thus \cASS starts by merging $R_k$ and $R_{k+1}$.
Hence, we assume that $k \geqslant 2$.
If $k-1$ is dominated, then
$\ell_{k-1} \leqslant \max\{\ell_k,\ell_{k+1}\} \leqslant \ell_\llarge$,
and thus $\max\{\ell_{k-1},\ell_k\} \leqslant \ell_\llarge$,
proving that $k-1$ is also low.
Since $k$ is the minimal low dominated integer, it follows that $k-1$
cannot be dominated.

In particular, an immediate induction on the number of
operations performed by \cASS proves that,
in every state encountered until \cASS ever merges
one of the runs $R_k$ or $R_{k+1}$, the run $R_k$ will always be preceded
by a run $R$ such that $\ell > \max\{\ell_k,\ell_{k+1}\}$.
Indeed, the only way to modify that run is to merge it with
other runs, which cannot decrease its level.

Now, let $m$ be the first merge performed by \cASS that involves some
of the runs $R_k,R_{k+1},\ldots,R_\rho$. 
We shall prove that $m$ consists in
merging the runs $R_k$ and $R_{k+1}$.
Let also $\obS = (\oR_1,\ldots,\oR_t,R_k,R_{k+1},\ldots,R_\rho)$
be the state just before $m$ takes place.
First, $m$ cannot merge any of the runs $\oR_1,\ldots,\oR_{t-1}$.
Second, since $\ol_t > \max\{\ell_k,\ell_{k+1}\}$, it cannot merge
$\oR_t$ either. Consequently, and since
$\ell_k \leqslant \max\{\ell_{k+1},\ell_{k+2}\}$, it must merge
$R_k$ and $R_{k+1}$.
\end{proof}

\begin{proposition}\label{pro:phase-1-cASS}
Every merge performed during phase~1 of \cASSm on a sequence
of runs $\R$ would also
have been performed by \cASS on the sequence $\R$.
\end{proposition}

\begin{proof}
Let $\bS_0,\ldots,\bS_p$ be the sequence of states
encountered in phase~1 of \cASSm, with $\bS_0 = \R$.
For all $i \leqslant p-1$, 
since $\bS_{i+1} = \succl(\bS_i)$,
Lemma~\ref{lem:first-kappa-merge} proves that
$\bS_{i+1}$ is a border of the merge tree generated by $\bS_i$ for \cASS.
It follows at once that every state $\bS_i$ is a border of
the tree $\T$ generated by $\R$ for \cASS, which completes the proof.
\end{proof}

As a nice consequence of Proposition~\ref{pro:phase-1-cASS},
we can already observe that using \cASSm instead of \cASS comes with no
additional merge cost.

\begin{corollary}\label{cor:cassm-better-than-cass}
The merge cost of \cASSm cannot exceed the cost of \cASS.
\end{corollary}

\begin{proof}
If phase~2 of \cASSm consisted in running \cASS instead of \MinSS,
the algorithms \cASS and \cASSm would perform the same list of merges,
possibly in a different order. Since running \MinSS cannot increase the
total merge cost of those merges performed in phase~2, the result follows.
\end{proof}

In order to go further in our proof, let us look for a nice
characterisation of the state $\bS$ obtained when phase~2 starts.
Such a characterisation relies on the notions of \emph{high} nodes
in merge tree, and of \emph{least high border}.

\begin{definition}\label{pro:visible-node}
Let $\R = (R_1,\ldots,R_\rho)$ be a sequence of runs,
and let $n = r_1+\ldots+r_\rho$ and
$\ell_\llarge^\star = \lfloor \log_2(n / \kappa) \rfloor$.
Let also $\T$ be a merge tree induced by $\R$, and let
$R$ be a node of $\T$, of level $\ell$.
We say that $R$ is \emph{immense} if $\ell \geqslant \ell_\llarge^\star+2$,
that $R$ is \emph{very high} if $\ell \geqslant \ell_\llarge^\star+1$,
and that $R$ is \emph{high} if $R$ or its sibling (in case
$R$ is not the root of $\T$) is very high.
\end{definition}

\begin{definition}\label{def:least-visible-border}
Let $\T$ be a merge tree induced by the sequence $\R$,
and let $\T'$ be the sub-tree that consists of the high nodes of $\T$.
The border of $\T$ that consists of the leaves of $\T'$
is called the \emph{least high border} of $\T$.
\end{definition}

\begin{proposition}\label{pro:last-state-phase-1}
Let $\bS$ be the last state encountered in phase~1 of \cASSm
when sorting the sequence $\R$,
and let $\T$ be the merge tree induced by \cASS on $\R$.
The least high border of $\T$ is a border of the ancestor sub-tree
of $\bS$ in $\T$.
\end{proposition}

\begin{proof}
Let $\B$ be the least high border of $\T$, and let us assume,
for the sake of contradiction, that
some run $R$ of $\bS$ is a strict ancestor of a node in $\B$.
By definition, both children of $R$ are high nodes of $\T$.
Hence, one of them, say $R'$, is a run of level $\ell' > \ell_\llarge^\star$,
where we recall that $\ell_\llarge^\star = \lfloor \log_2(n / \kappa) \rfloor$
and $n$ is the sum of the lengths of the runs in every border of $\T$.

Now, consider the state $\obS = (\oR_1,\ldots,\oR_\rho)$ just before \cASSm proceeded to merging $R'$.
By construction, we know that
$\orr_\llarge = \lfloor (\orr_1+\ldots+\orr_{2\kappa-1})/\kappa \rfloor
\leqslant (\orr_1+\ldots+\orr_\rho)/\kappa = n / \kappa$, and thus that
$\ol_\llarge = \lfloor \log_2(\orr_\llarge) \rfloor
\leqslant \ell_\llarge^\star$.
It follows that $\ell' > \ol_\llarge$, contradicting the fact that
\cASS might have merged $R'$. This completes the proof.
\end{proof}

As a consequence of Proposition~\ref{pro:last-state-phase-1},
we introduce the algorithm \cASSM,
which is the variant of \cASSm presented in Algorithm~\ref{alg:ASSM}.
The details of how the runs $R$ and $R'$ are chosen
in lines~\ref{alg:ASSM:0} and~\ref{alg:ASSM:2} 
is deliberately left unspecified here;
anyways, the merge cost of \cASSM does not depend on which choices
are performed then.

\begin{algorithm}[th!]
\begin{small}
\DontPrintSemicolon
\SetArgSty{texttt}
\SetKwInOut{Input}{Input}
\SetKwInput{KwData}{Note}
\SetKwProg{Pn}{Function}{:}{}
\SetKwFunction{MustMerge}{\mustMerge}
\SetKwFunction{DoMerge}{\doMerge}
\Input{\hspace*{1.95mm}An array to $A$ to sort, integer parameter $\kappa$}
\KwResult{The array $A$ is sorted into a single run, which remains on the 
stack.}
\BlankLine
\KwData{We \hfill view \hfill $\rundecomp$ \hfill as \hfill a \hfill 
mutable \hfill sequence \hfill of \hfill runs. \hfill Whenever \hfill two \hfill
successive \hfill runs \hfill of \hfill $\rundecomp$ \hfill are \hfill merged,
they \hfill are \hfill replaced, \hfill in \hfill $\rundecomp$, \hfill 
by \hfill the \hfill run \hfill resulting \hfill from \hfill the \hfill merge.
In \hfill practice, \hfill in
$\rundecomp$, each
run is represented by a pair of points to its first and last entries.}
\BlankLine
$\rundecomp \gets$ the run decomposition of $A$\label{alg:M:init}\;
$\T \gets$ merge tree induced by \cASS on $\rundecomp$\;
$\bS \gets$ last state encountered in phase~1 of \cASSm\;
$\B \gets$ least high border of $\T$\;
\While{\textrm{two nodes $R$ and $R'$ of $\rundecomp$ are
 siblings in $\T$ and strictly descend from some node in
  $\bS$\label{alg:ASSM:0}}}
  {merge the runs $R$ and $R'$\label{alg:ASSM:1}}
\While{\textrm{two nodes $R$ and $R'$ of $\rundecomp$ are
 siblings in $\T$ and strictly descend from some node in
  $\B$\label{alg:ASSM:2}}}
  {merge the runs $R$ and $R'$\label{alg:ASSM:3}}
apply \MinSS on the resulting sequence $\rundecomp$%
\label{alg:ASSM:brute-force-merge}
\end{small}
\caption{\cASSM\label{alg:ASSM}}
\end{algorithm}

\begin{lemma}\label{lem:star-worse-than-kappa}
Let $\mc$ and $\mc^\star$ be the
respective merge costs of \cASSm and \cASSM when sorting
a sequence $\R$ of runs. We have $\mc \leqslant \mc^\star$.
\end{lemma}

\begin{proof}
Thanks to Proposition~\ref{pro:phase-1-cASS},
both algorithms \cASSm and \cASSM start by performing
merges prescribed by \cASS (possibly in different orders) until they obtain
the state $\bS$. Then, \cASSm sorts $\bS$ optimally, by applying \MinSS,
while \cASSM may perform sub-optimal merges.
It follows immediately that $\mc \leqslant \mc^\star$.
\end{proof}

Now, let us study the overhead due to using \cASSM instead of \MinSS.
This overhead has two causes: (i)~using run merges prescribed by \cASS
to transform the initial sequence $\R$ into the border $\B$ could be
sub-optimal, and (ii)~using $\B$ as an intermediary state
could be sub-optimal too.

For the sake of readability when evaluating the impacts
of these two causes, let us introduce now some notations.
In what follows, we consider a given sequence of runs
$\R = (R_1,\ldots,R_\rho)$, that will be fixed once and for all,
and therefore left implicit whenever possible.

As hinted when we defined the notions of \emph{dominated}
and \emph{low} integers, we wish to 
identify every run $R_i$ with the integer $i$.
More precisely, and for each sub-interval $X = \{x,\ldots,y\}$
of the set $\{1,\ldots,\rho\}$,
we denote by $\R_X$ the sequence $(R_x,\ldots,R_y)$ and by
$R_X$ the run obtained by merging these runs $R_x,\ldots,R_y$.
Similarly, if $\bI = (I_1,\ldots,I_s)$
is a partition of $X$ into intervals, 
we denote by $\R_\bI$ the sequence $(R_{I_1},\ldots,R_{I_s})$.
Finally, we set $r_X = \sum_{k \in X} r_k$,
$\ell_X = \lfloor \log_2(r_X) \rfloor$ and
\[\H_X = - \sum_{i=x}^y r_i / r_X
\log_2(r_i /r_X).\]

\begin{definition}\label{def:diff:partition}
Let $X$ be a sub-interval of the set $\{1,\ldots,\rho\}$.
Let also $\bI = (I_1,\ldots,I_s)$ and $\bJ = (J_1,\ldots,J_t)$
be two partitions of $X$ into intervals,
such that $\bJ$ refines $\bI$,
i.e., that every interval $J_j$ is contained in some interval~$I_i$.
For all $i \leqslant s$, we set 
$\mathsf{d}_i = |\{j \colon J_j \subseteq I_i\}|$,
i.e., $\mathsf{d}_i$ is the number of intervals $I_i$ in which $J_j$
is subdivided. We say that $\bJ$ is a \emph{$\mathsf{d}$-refinement} of
$\bI$ if $\mathsf{d}_i \leqslant \mathsf{d}$ for all $i \leqslant s$.
We also call \emph{distortion} between $\bI$ and $\bJ$
the integer
\[\delta(\bI,\bJ) = \sum_{i=1}^s r_{I_i} 
\mathbf{1}_{\mathsf{d}_i \neq 1},\]
i.e., the total length of those runs $R_{I_i}$ that do not belong
to the sequence $\R_\bJ$
(or, equivalently, the total length of those runs $R_{J_j}$ that do not belong
to the sequence $\R_\bI$).

Finally, we denote by $\mc(\bJ \to \bI)$ the minimal cost of a stable merge
policy that
transforms the sequence $\R_\bJ$ into the sequence $\R_\bI$
by successively merging adjacent runs and,
if $\R_\bI$ is a border of the merge tree induced by $\cASS$ on $\R_\bJ$,
we denote by $\mcass(\bJ \to \bI)$ the total cost of those merges
prescribed by $\cASS$ and that allow transforming
the sequence $\R_\bJ$ into the sequence $\R_\bI$.
\end{definition}

Below, we also denote by $\start$ the finest partition
$\{\{1\},\ldots,\{\rho\}\}$ and by $\final$ the coarsest partition
$\{\{1,\ldots,\rho\}\}$. Then, $\R_\start$ is simply the sequence $\R$
(or, equivalently, $\R_{\{1,\ldots,\rho\}}$),
whereas $\R_\final$ is the sequence whose only run is the sorted array itself.

\begin{lemma}\label{lem:singular-cost-cASS}
Let $\T$ be the merge tree induced by \cASS on the sequence $\R$, and
let $R_X$ be a node of the tree $\T$, where $X$ is a sub-interval of
$\{1,\ldots,\rho\}$.
The sub-tree of $\T$ rooted at $R_X$
coincides with the merge tree induced by \cASS on the sequence~$\R_X$.
\end{lemma}

\begin{proof}
Let $d$ be the depth of the node $R_X$ in the tree $\T$,
and let $\T'$ be the sub-tree of $\T$ rooted at $R_X$.
We prove Lemma~\ref{lem:singular-cost-cASS} by induction on the pair $(d,\rho)$.
First, if $d = 0$, then $R_X$ is the root of $\T$, and the result clearly holds.

Second, if $d = 1$, let $R_Y$ be the sibling of $R_X$.
If $Y = \{1\}$ or $Y = \{\rho\}$,
replacing the run~$R_Y$ by an arbitrarily large run
would not alter the dynamics of \cASS, and thus
the conclusion of Lemma~\ref{lem:singular-cost-cASS}
comes immediately. However, if $Y$ contains at least two elements,
let $\B$ be the border of $\T$ that consists of the leaves of $\T'$
and of the node $R_Y$: this border contains at most
$\rho-1$ nodes. Theorem~\ref{thm:stable-border} states that the merge tree
$\T''$ induced by \cASS on the sequence $\B$ coincides
with the ancestor sub-tree of $\B$,
and thus $\T'$ is also the sub-tree of $\T''$ rooted at $R_X$.
Hence, the induction hypothesis proves the desired result in that case too.

Finally, if $d \geqslant 2$, let $R_Z$ be the parent of $R_X$,
and let $\T''$ be the sub-tree of $\T$ rooted at $R_Z$. The induction
hypothesis proves first that $\T''$ coincides with the merge tree induced by
\cASS on the sequence $\R_Z$, and then that $\T'$ coincides with the merge
tree induced by \cASS on the sequence $\R_X$.
\end{proof}

\begin{proposition}\label{pro:overcost-cASS}
Let $\bI$ be a partition
of the set $\{1,\ldots,\rho\}$ into intervals
such that $\R_\bI$ is a border
of the merge tree induced by $\cASS$ on $\R$.
We have
\[\mcass(\start \to \bI) \leqslant
\mc(\start \to \bI) + \Delta \, \delta(\bI,\start),\]
where we recall that $\Delta = 24/5 - \log_2(5)$.
\end{proposition}

\begin{proof}
Transforming $\R$ into $\R_\bI$ amounts to transforming,
for each interval $I \in \bI$,
the sequence $\R_I$ into the run $R_I$,
and therefore Theorem~\ref{thm:lower-bound-merge-cost} states that
the merge cost of that transformation is at least~$r_I \H_I$
when $|I| \geqslant 2$.

Then, let $\T$ be the merge tree induced by \cASS on the sequence $\R$
and, for each interval $I \in \bI$, let $\T_I$ be the sub-tree of $\T$
rooted at $R_I$.
Those merges prescribed by \cASS and that allow transforming
$\R$ into $\R_\bI$ are the merges between runs $R$ and $R'$
that are siblings and belong to some sub-tree $\T_I$, and therefore
the total cost of these merges is equal to the sum of the costs
of these sub-trees.
Lemma~\ref{lem:singular-cost-cASS} then states that
each sub-tree $\T_I$
coincides with the merge tree induced by \cASS on the sequence~$\R_I$,
and thus Theorem~\ref{thm:complexity-nH+D} prove that
the cost of $\T_I$ is at most $r_I (\H_I + \Delta)$.

We conclude that
\[\mcass(\start \to \bI) \leqslant 
\sum_{I \in \bI} r_I (\H_I + \Delta) 
\mathbf{1}_{|I| \geqslant 2} \leqslant 
\mc(\start \to \bI) + \Delta \, \delta(\bI,\start).\]
\end{proof}

In the special case where the sequence
$\R_\bI$ is equal to the least high border $\B$,
Proposition~\ref{pro:overcost-cASS} will be a powerful tool to
evaluate the overhead due to transforming
the initial sequence $\R$ into $\R_\bI$
by using merges presecribed by \cASS instead of using an
optimal merge policy.
We must now focus on the overhead due to choosing that sequence $\R_\bI$
as an intermediate state. We do so by finding another partition $\bJ$
that is \emph{close enough} to $\bI$, in a sense
that we will precise later, and such that \MinSS may use the
sequence $\R_\bJ$ as an intermediate state when sorting the sequence $\R$.

\begin{lemma}\label{lem:refinement-1}
Let $\bI$ and $\bJ$
be two partitions of the set $\{1,\ldots,\rho\}$ into intervals,
such that $\bJ$ refines~$\bI$.
We have $\mc(\start \to \bJ) \leqslant \mc(\start \to \bI)$ and
$\mc(\bJ \to \final) \leqslant \mc(\bI \to \final)$.
\end{lemma}

\begin{proof}
Let $\bI = (I_1,\ldots,I_n)$
and $\bJ = (J_1,\ldots,J_m)$ be our two partitions
of the set $\{1,\ldots,\rho\}$.
We prove Lemma~\ref{lem:refinement-1} by induction on $m$.

First, if $m \leqslant n$, then $\bI = \bJ$, and the result is immediate.
Then, if $m \geqslant n+2$, let $\bK$
be a partition that lies
between the partitions $\bI$ and $\bJ$
(i.e., $\bJ$ strictly refines $\bK$,
which strictly refines $\bI$).
The induction hypothesis ensures that
$\mc(\start \to \bJ) \leqslant \mc(\start \to \bK) 
\leqslant \mc(\start \to \bI)$ and
$\mc(\bJ \to \final) \leqslant \mc(\bK \to \final)
\leqslant \mc(\bI \to \final)$.

Finally, let us assume that $m = n+1$.
We first prove that $\mc(\bJ \to \final) \leqslant \mc(\bI \to \final)$.
Without loss of generality, let us assume that $\bJ = \start$.
In that case, let $I_i$ be the only interval such that $\mathsf{d}_i = 2$:
it is the disjoint union of the intervals $J_i = \{i\}$ and $J_{i+1} = \{i+1\}$.
Then, let $\T$ be the merge tree generated by \MinSS on the sequence $\R$.
Still without loss of generality, and denoting by
$\delta_i$ and $\delta_{i+1}$ be the respective
depths of the leaves $R_i$ and $R_{i+1}$ in $\T$, we assume that 
$\delta_i \geqslant \delta_{i+1}$.

Then, let us consider the sequence of merges
performed by \MinSS when sorting the sequence $\R$:
these are the merges between two runs $R_X$ and $R_Y$ that are siblings
in the tree $\T$.
We modify that sequence as follows.
First, the (unique) merge between the run $R_i$ and another run is deleted.
Second, and for every run $R_X$ ever involved in a merge:
\begin{itemize}
\item if $X$ is of the form $\{i+1,\ldots,x\}$ (with $x \geqslant i+1$)
we replace $R_X$ by the run $R_{\{i,\ldots,x\}}$;
\item if $X$ is of the form $\{x,\ldots,i\}$ (with $x \leqslant i-1$)
we replace $R_X$ by the run $R_{\{x,\ldots,i-1\}}$;
\item otherwise, we keep the run $R_X$ as is.
\end{itemize}
This situation is illustrated in Figure~\ref{fig:modify-tree},
where merges of the the former sequence are gathered
in the tree $\T$ (on the left) and the merges of
the latter (modified) sequence
are gathered in the tree $\T'$ (on the right).

\begin{figure}[ht]
\begin{center}
\begin{tikzpicture}[scale=0.5]
\draw[ultra thick,draw=black]
(6,-7.5) -- (3,-2.5) -- (0,0) -- (-3,-2.5) -- (-4.5,-5)
(3,-7.5) -- (4.5,-5)
(1.5,-5) -- (3,-2.5)
(0,-7.5) -- (-3,-2.5)
(-4.5,-10) -- (-1.5,-5)
(-1.5,-10) -- (-3,-7.5);
\ovalnode{0}{0}{$R_{\{1,\ldots,7\}}$}
\ovalnode{3}{-2.5}{$R_{\{5,6,7\}}$}
\roundvalnode{4.5}{-5}{$R_{\{6,7\}}$}
\roundnode{6}{-7.5}{$R_7$}
\roundnode{3}{-7.5}{$R_6$}
\roundnode{1.5}{-5}{$R_5$}
\ovalnode{-3}{-2.5}{$R_{\{1,\ldots,4\}}$}
\ovalnode{-1.5}{-5}{$R_{\{2,3,4\}}$}
\roundnode{0}{-7.5}{$R_4$}
\roundvalnode{-3}{-7.5}{$R_{\{2,3\}}$}
\roundnode{-1.5}{-10}{$R_3$}
\roundnode{-4.5}{-10}{$R_2$}
\roundnode{-4.5}{-5}{$R_1$}
\draw[ultra thick,draw=black,->,>=stealth] (7.25,-5) -- (9.5,-5);
\draw[ultra thick,draw=black]
(21.5,-7.5) -- (18.5,-2.5) -- (15.5,0) -- (12.5,-2.5) -- (11,-5)
(18.5,-7.5) -- (20,-5)
(17,-5) -- (18.5,-2.5)
(15.5,-7.5) -- (12.5,-2.5)
(12.5,-7.5) -- (14,-5);
\ovalnode{15.5}{0}{$R_{\{1,\ldots,7\}}$}
\ovalnode{18.5}{-2.5}{$R_{\{4,\ldots,7\}}$}
\roundvalnode{20}{-5}{$R_{\{6,7\}}$}
\roundnode{21.5}{-7.5}{$R_7$}
\roundnode{18.5}{-7.5}{$R_6$}
\roundvalnode{17}{-5}{$R_{\{4,5\}}$}
\ovalnode{12.5}{-2.5}{$R_{\{1,2,3\}}$}
\roundvalnode{14}{-5}{$R_{\{2,3\}}$}
\roundnode{15.5}{-7.5}{$R_3$}
\roundnode{12.5}{-7.5}{$R_2$}
\roundnode{11}{-5}{$R_1$}

\node[anchor=south] at (0,1.25) {Tree $\T$};
\node[anchor=south] at (15.5,1.25) {Tree $\T'$};
\end{tikzpicture}
\end{center}
\vspace{-\baselineskip}
\caption{Transforming the tree $\T$ into a new tree $\T'$ of smaller merge cost
(in case $i = 4$)\label{fig:modify-tree}}
\end{figure}

By following this new sequence of merges,
we transform the sequence $\R_\bI$ into the sequence $\R_\final$,
for a total cost $\mc$ such that $\mc(\bJ \to \final) \leqslant \mc 
\leqslant \mc(\bI \to \final)$.

Second, we prove that $\mc(\start \to \bJ) \leqslant \mc(\start \to \bI)$.
Again, and without loss of generality, we assume that $\bI = \final$,
and thus that
$\bJ = \{\{1,\ldots,z\},\{z+1,\ldots,\rho\}\}$
for some integer $z$.
Then, let us consider the sequence of merges
performed by \MinSS when sorting the sequence $\R$.
We modify that sequence as follows:
\begin{itemize}
\item the (unique) merge between runs $R_X$ and $R_Y$ such that
$z \in X$ and $z+1 \in Y$ is deleted;
\item every merge between runs $R_X$ and $R_Y$ such that
$\{z,z+1\} \subseteq X$ is replaced by a merge between
the runs $R_{\{x \in X \colon x \geqslant z+1\}}$ and $R_Y$;
\item every merge between runs $R_X$ and $R_Y$ such that
$\{z,z+1\} \subseteq Y$ is replaced by a merge between
the runs $R_X$ and $R_{\{y \in Y \colon y \leqslant z\}}$;
\item other merges are not modified.
\end{itemize}
This situation is illustrated in Figure~\ref{fig:modify-tree},
where merges of the the former sequence are gathered
in the tree $\T$ (on the left) and the merges of
the latter (modified) sequence
are gathered in the forest $\F$ (on the right),
which consists of two trees:
one that gathers the runs $R_X$ such that $X \subseteq \{1,\ldots,z\}$
and one that gathers the runs $R_X$ such that $X \subseteq \{z+1,\ldots,\rho\}$.

\begin{figure}[ht]
\begin{center}
\begin{tikzpicture}[scale=0.5]
\draw[ultra thick,draw=black]
(-6,-7.5) -- (-3,-2.5) -- (0,0) -- (3,-2.5) -- (4.5,-5)
(-3,-7.5) -- (-4.5,-5)
(-1.5,-5) -- (-3,-2.5)
(0,-7.5) -- (3,-2.5)
(4.5,-10) -- (1.5,-5)
(1.5,-10) -- (3,-7.5);
\ovalnode{0}{0}{$R_{\{1,\ldots,7\}}$}
\ovalnode{-3}{-2.5}{$R_{\{1,2,3\}}$}
\roundvalnode{-4.5}{-5}{$R_{\{1,2\}}$}
\roundnode{-6}{-7.5}{$R_1$}
\roundnode{-3}{-7.5}{$R_2$}
\roundnode{-1.5}{-5}{$R_3$}
\ovalnode{3}{-2.5}{$R_{\{4,\ldots,7\}}$}
\ovalnode{1.5}{-5}{$R_{\{4,5,6\}}$}
\roundnode{0}{-7.5}{$R_4$}
\roundvalnode{3}{-7.5}{$R_{\{5,6\}}$}
\roundnode{1.5}{-10}{$R_5$}
\roundnode{4.5}{-10}{$R_6$}
\roundnode{4.5}{-5}{$R_7$}
\draw[ultra thick,draw=black,->,>=stealth] (6.25,-5) -- (8.5,-5);
\draw[ultra thick,draw=black]
(9.5,-7.5) -- (12.5,-2.5) -- (15.5,0) -- (18.5,-2.5) -- (20,-5)
(12.5,-7.5) -- (11,-5)
(14,-5) -- (12.5,-2.5)
(17,-5) -- (18.5,-2.5)
(17,-10) -- (18.5,-7.5) -- (20,-10);
\ovalnode{15.5}{0}{$R_{\{1,\ldots,5\}}$}
\ovalnode{12.5}{-2.5}{$R_{\{1,2,3\}}$}
\roundvalnode{11}{-5}{$R_{\{1,2\}}$}
\roundnode{9.5}{-7.5}{$R_1$}
\roundnode{12.5}{-7.5}{$R_2$}
\roundnode{14}{-5}{$R_3$}
\roundvalnode{18.5}{-2.5}{$R_{\{4,5\}}$}
\roundnode{17}{-5}{$R_4$}
\roundnode{20}{-5}{$R_5$}
\roundvalnode{18.5}{-7.5}{$R_{\{6,7\}}$}
\roundnode{17}{-10}{$R_6$}
\roundnode{20}{-10}{$R_7$}

\node[anchor=south] at (0,1.25) {Tree $\T$};
\node[anchor=south] at (15.5,1.25) {Forest $\F$};
\end{tikzpicture}
\end{center}
\vspace{-\baselineskip}
\caption{Transforming the tree $\T$ into a new forest $\F$ of smaller
merge cost (in case $z = 5$)\label{fig:modify-tree-2}}
\end{figure}

By following this new sequence of merges,
we transform the initial sequence $\R$ into the sequence $\R_\bJ$,
for a total cost $\mc$ such that $\mc(\start \to \bJ) \leqslant \mc 
\leqslant \mc(\start \to \bI)$.
\end{proof}

\begin{lemma}\label{lem:refinement-2}
Let $\bI$ and $\bJ$
be two partitions of the set $\{1,\ldots,\rho\}$ into intervals,
such that $\bJ$ refines~$\bI$.
If $\bJ$ is a $\mathsf{d}$-refinement of $\bI$, then
$\mc(\bJ \to \bI) \leqslant \lceil \log_2(\mathsf{d}) \rceil \delta(\bI,\bJ)$.
\end{lemma}

\begin{proof}
Let $\bI = (I_1,\ldots,I_n)$
and $\bJ = (J_1,\ldots,J_m)$ be our two partitions
of the set $\{1,\ldots,\rho\}$.
Consider some integer $I_i$ such that $\mathsf{d}_i \geqslant 2$,
and let $j$ be the integer such that
$I_i$ is the disjoint union of $J_{j+1},J_{j+2},\ldots,J_{j+\mathsf{d}_i}$.
We can merge $R_{J_{j+1}},\ldots,R_{J_{j+\mathsf{d}_i}}$ into one unique run
$R_{I_i}$ by using a balanced binary merge tree of height
$\lceil \log_2(\mathsf{d}_i) \rceil$. The total merge cost of these operations
is thus bounded above by $\lceil \log_2(\mathsf{d}) \rceil \, r_{I_i}$.
Proceeding in this way for every interval $I_i$ such that
$\mathsf{d}_i \geqslant 2$, we transform the sequence
$\R_{\bJ}$ into $\R_{\bI}$
for a cost of $\lceil \log_2(\mathsf{d}) \rceil \, \delta(\bI,\bJ)$
or less.
\end{proof}

\begin{lemma}\label{lem:least-high-nodes}
Let $\T$ be a merge tree induced by the sequence $\R$,
and let $\B$ be the least high border of~$\T$.
It holds that (i)~each very high run of $\R$ also belongs to $\B$, and
(ii)~each immense run of $\B$ also belongs to $\R$
\end{lemma}

\begin{proof}
Claim~(i) immediately follows from the definition of the
least high border. Then, if an immense run $R$
has two children $R_\ominus$ and $R_\oplus$ in $\T$,
Lemma~\ref{lem:l-small-increase}
proves that $\max\{\ell_\ominus,\ell_\oplus\} \geqslant \ell-1 \geqslant
\ell_\llarge^\star+1$. This means that $R_\ominus$ and $R_\oplus$ are
high nodes, and thus that $R$ cannot belong to the border $\B$, thereby
proving claim~(ii).
\end{proof}

\begin{lemma}\label{lem:last-phase-2}
Let $\T$ be the merge tree induced by \cASS on the sequence $\R$,
and let $\B$ be the least high border of $\T$.
Among any two consecutive runs in $\B$,
at least one is very high.
\end{lemma}

\begin{proof}
Let us assume that there exist two consecutive runs $R$ and $R'$
of $\B$ that are not very high. 
Applying the algorithm \cASS on the border $\B$,
let $\bS$ be the state obtained just before one of the runs
$R$ or $R'$ is merged.
The run $R$ cannot be merged with $R'$, unless what neither run
would be high. Thus, it must be merged with its predecessor $\oR$ in 
the state $\bS$.
It follows that $\ol \leqslant \max\{\ell,\ell'\} \leqslant
\ell^\star_\llarge$, and thus that neither $\oR$ nor $R$ is very high.
This contradicts the fact that both nodes should be high, thereby
disproving our initial assumption.
\end{proof}

\begin{lemma}\label{no-small}
Let $R$ be some run obtained while applying \MinSS
on the sequence $\R$ (i.e., $R$ belongs to the merge tree induced by \MinSS
on $\R$),
and let $\overline{R}_1,\ldots,\overline{R}_k$
be the runs with which \MinSS successively merges the elements of $R$.
For all $i \geqslant 3$, it holds that~$\overline{r}_i \geqslant r$.
\end{lemma}

\begin{proof}
Let us first assume that $\overline{r}_3 < r$.
Then, the total cost of the merges of $R$ with
$\overline{R}_1$, $\overline{R}_2$ and $\overline{R}_3$ is
$\mc = 3(r + \overline{r}_1) + 2 \overline{r}_2 + \overline{r}_3$.
However, if we had used a balanced binary tree of height 2
for merging these four runs, (i.e., merging first
the two leftmost runs, then the two rightmost runs,
and finally the two resulting runs),
each run would have participated to $2$ merges only,
for a total cost of
$\mc' = 2 (r + \overline{r}_1 + \overline{r}_2 + \overline{r}_3) =
\mc - (r + \overline{r}_1 - \overline{r}_3) < \mc$.
This contradicts the optimality of our merge policy,
which proves that $\overline{r}_3 \geqslant r$.

The same reasoning, applied to the run obtained
by merging $R$ and $\overline{R}_1,\ldots,\overline{R}_i$,
shows that $\overline{r}_{i+3} \geqslant r + \overline{r}_1 + \ldots + \overline{r}_i \geqslant r$ for all $i \leqslant k-3$,
which completes the proof.
\end{proof}

\begin{corollary}\label{cor:nice-border}
Let $\T$ be the merge tree induced by \MinSS on the sequence $\R$,
and let $\B$ be the least high border of $\T$.
Among any five consecutive runs in $\B$, at least one is very high.
\end{corollary}

\begin{proof}
Let $\T'$ be the ancestor sub-tree of $\B$ in $\T$,
and let us assume that there exist five consecutive runs
$R_{i-2}$, $R_{i-1}$, $R_i$, $R_{i+1}$ and $R_{i+2}$ of $\B$ 
that are not very high.
Let $R'$ be the parent of $R_i$, and let
$\langle R' \rangle$ be the sub-tree of $\T'$ rooted at $R'$.
Among any two sibling leaves of $\langle R' \rangle$ of maximal depth,
one of them, say $\oR$, must be very high.

Since $\ol_i \leqslant \ell_\llarge^\star < \ol$,
and thus $\orr_i < \orr$, Lemma~\ref{no-small} proves that
$R_i$ is either the sibling or the uncle of $\oR$.
This means that $\langle R' \rangle$ has either two or three leaves and,
since these leaves that consecutive runs of $\B$,
they must all belong to the set $\{R_{i-2},R_{i-1},R_i,R_{i+1},R_{i+2}\}$.
Thus, one of these runs was in fact very high,
which invalidates our initial assumption and completes the proof.
\end{proof}

\begin{proposition}\label{pro:S-close-B}
Let $\T_\ass$ and $\T_{\min}$ be the merge trees
respectively induced by \cASS and \MinSS on the sequence $\R$.
Then, let $\bI$ and $\bJ$ be the partitions of $\{1,\ldots,\rho\}$ such that
$\R_\bI$ is the least high border of $\T_\ass$ and
$\R_\bJ$ is the least high border of $\T_{\min}$.
Finally, let $\bK$ be the coarsest partition of $\{1,\ldots,\rho\}$
(i.e., the partition in the least possible number of intervals)
that refines both $\bI$ and $\bJ$.
This partition is a $11$-refinement of $\bI$ and a $7$-refinement of $\bJ$.
\end{proposition}

\begin{proof}
First, assume that some interval $I$, belonging to the partition $\bI$,
contains at least $12$ sub-intervals $K_1,\ldots,K_{12}$
of the partition $\bK$.
By definition of $\bK$, each of the intervals $K_2,\ldots,K_{11}$
belongs to $\bJ$.
Corollary~\ref{cor:nice-border} proves that one of
runs $R_{K_2},\ldots,R_{K_6}$, say $R_{K_x}$, is very high.
Similarly, one of the runs $R_{K_7},\ldots,R_{K_{11}}$, say $R_{K_y}$,
is very high too. It follows that
\[r_I \geqslant \sum_{k=1}^{12} r_{K_k} \geqslant r_{K_x} + r_{K_y}
 \geqslant 2 \times 2^{\ell^\star_\llarge+1} = 2^{\ell^\star_\llarge+2},
\]
which means that $R_I$ is immense.
Hence, Lemma~\ref{lem:least-high-nodes} proves that
$R_I$ also belongs to $\R$ and to $\R_\bJ$, which is a contradiction.

Likewise, assume some interval $J$ of the partition $\bJ$ contains at least
$8$ sub-intervals $K_1,\ldots,K_8$ of the partition $\bK$,
then all sub-intervals $K_2,\ldots,K_7$ also belong to $\bI$.
Lemma~\ref{lem:last-phase-2} proves that one of
runs $R_{K_2}$, $R_{K_3}$ or $R_{K_4}$, say $R_{K_x}$, is very high.
Similarly, one of the runs $R_{K_5}$, $R_{K_6}$ or $R_{K_7}$,
say $R_{K_y}$, is very high too. But then
\[r_J \geqslant \sum_{k=1}^{8} r_{K_k} \geqslant r_{K_x} + r_{K_y}
 \geqslant 2 \times 2^{\ell^\star_\llarge+1} = 2^{\ell^\star_\llarge+2},
\]
which means that $R_J$ is immense, again leading to a contradiction.
\end{proof}

\begin{theorem}\label{thm:cASSm-optimal}
The algorithm \cASSm is $\eta_{2\kappa+3}$-optimal.
\end{theorem}

\begin{proof}
Let $\R = (R_1,\ldots,R_\rho)$ 
be a sequence of runs to sort, of total length $n$ and entropy $\H$.
Let $X = \{i \colon R_i$ is not very high$\}$,
and let $n^\star = \sum_{i \in X}r_X$ be the sum of the lengths
of those runs $R_i$ that are not very high.
For all $i \in X$, we know that
$\lfloor \log_2(r_i) \rfloor = \ell_i \leqslant
\ell^\star_\llarge \leqslant \log_2(n/\kappa)$, and thus that
$r_i \leqslant 2 n / \kappa$.
It follows that
\[n\H = \sum_{i=1}^\rho r_i \log_2(n/r_i)
\geqslant \sum_{i \in X} r_i \log_2(n / r_i)
\geqslant \sum_{i \in X} r_i \log_2(\kappa/2) = n^\ast \log_2(\kappa/2).\]

Then, let $\bI$, $\bJ$ and $\bK$ be the partitions of $\{1,\ldots,\rho\}$
mentioned in Proposition~\ref{pro:S-close-B}.
Lemma~\ref{lem:least-high-nodes} proves that each very high run $R_i$
belongs both to $\R_\bI$ and to $\R_\bJ$, and thus
the disortions between the partitions $\bI$, $\bJ$, $\bK$ and $\start$
satisfy the inequalities
$\delta(\bI,\bK) \leqslant n^\star$, $\delta(\bJ,\bK) \leqslant n^\star$
and $\delta(\bI,\start) \leqslant n^\star$.
Thus, it follows from Lemmas~\ref{lem:refinement-1},~\ref{lem:refinement-2}
and Proposition~\ref{pro:S-close-B} that
\begin{align*}
\mc(\start \to \bI)
& \leqslant \mc(\start \to \bK) + \mc(\bK \to \bI)
\leqslant \mc(\start \to \bJ) + 4 n^\star \text{ and} \\
\mc(\bI \to \final)
& \leqslant \mc(\bK \to \final)
\leqslant \mc(\bK \to \bJ) + \mc(\bJ \to \final)
\leqslant 3 n^\star + \mc(\bJ \to \final)
\end{align*}

Finally, let us denote by 
$\mc_\kappa$, $\mc_{\sk}$ and $\mc_{\min}$ the respective merge
costs of \cASSm, \cASSM and \MinSS when sorting the sequence $\R$.
Theorem~\ref{thm:lower-bound-merge-cost} proves that
$\mc_{\min} \geqslant n\H \geqslant n^\ast \log_2(\kappa/2)$, 
and we conclude
that
\begin{align*}
\mc_\kappa & \leqslant \mc_{\sk} = \mcass(\start \to \bI) + \mc(\bI \to \final)
& \text{by Lemma~\ref{lem:star-worse-than-kappa}} \\
& \leqslant \mc(\start \to \bI) + \Delta n^\star + 
\mc(\bI \to \final)
& \text{by Proposition~\ref{pro:overcost-cASS}} \\
& \leqslant \mc(\start \to \bJ) + \mc(\bJ \to \final) + (\Delta+7) n^\star
= \mc_{\min} + (\Delta+7) n^\star \\
& \leqslant (1 + (\Delta+7)/ \log_2(\kappa/2)) \mc_{\min} = (1+\eta_{2\kappa+3})	 \mc_{\min},
\end{align*}
which means that \cASSm is $\eta_{2\kappa+3}$-optimal.
\end{proof}

In conclusion, the algorithm \cASSm is $(2\kappa+2)$-aware and
$\eta_{2\kappa+2}$-optimal.
Hence, it is also $(2\kappa+3)$-aware and $\eta_{2\kappa+3}$-optimal,
which proves Theorem~\ref{thm:eta-optimal}.

\section{Implementation details and simplifications}

One of the reasons for introducing the algorithm \cASS is that
it might be a good substitute to \TS, being both more efficient
in the worst case and simple to implement by modifying the code used for \TS.
However, past versions of \TS in languages such as Python or Java
suffered from implementation bugs~\cite{auger2018worst,GoRoBoBuHa15}, 
which involved the time complexity and,
most importantly, the space complexity of \TS.

Indeed, in both languages,
the stack $\S$ used in~\TS is simulated by a fixed-size array.
Allocating enough memory to that array is therefore a crucial step.
This task requires bounding by above the size that $\S$ may ever take
during the execution of \TS, before even starting to sort the array $A$
(i.e., when the only thing known about $A$ is its length).
However, that step was incorrectly performed, which led to
the bugs mentioned above. Consequently, and in order to avoid
similar problems in the future, we study this problem below.

Another critical point, if one were to replace \TS by \cASS,
consists in making sure that as few code lines as possible be modified,
and that switching between the two algorithms be straightforward.
This point was already taken care of since, as mentioned in the introduction,
the only part that distinguishes \cASS from \TS is the merge policy
used for choosing which \emph{large} runs to merge.
Below, we complete this task and actually provide
the few code lines that should be used in order to use \cASS instead of \TS.

\subsection{Stack size}\label{sec:stack-size}

In Section~\ref{sec:analysis}, we focused on the \emph{time} complexity of
\cASS, which is obviously an important parameter. However,
for the reasons mentioned just above, evaluating precisely the
\emph{space} complexity of \cASS is also important.
Thus, we first provide upper bounds on the stack size that might be required
while sorting an array of size $n$.

Since, as mentioned in Section~\ref{sec:intro}, \TS is also based on dealing
with small runs with an \emph{ad hoc} sub-routine,
we also take into account the minimal size $\sm$ that
characterizes runs \emph{large enough} to be considered by our merge policy
(except the last run, which may be of any size).
We also denote by $\lm = \lfloor \log_2(\sm) \rfloor$ the
\emph{minimal level} of such runs. 

\begin{proposition}\label{pro:stack-size}
At any time while sorting an array
of size $n \geqslant \sm$ by considering runs of size $\sm$ or more
(except, possibly, the rightmost run),
the stack size required by the algorithm
\cASS is at most $\lceil \log_2(n) \rceil + 1 - \lm$.
\end{proposition}

\begin{proof}
Let $\S = (R_1,\ldots,R_h)$ be some stack encountered
while executing \cASS. 
An immediate consequence of Lemma~\ref{lem:invariant-li}
is that $\ell_1 > \ldots > \ell_{h-2}$.
It follows that
\[\ell_i \geqslant \ell_{h-2} + (h-2-i) \geqslant \lm + (h-2-i)\]
for all $i \leqslant h-2$, and thus that
$r_i \geqslant 2^{\ell_i} \geqslant 2^{h-2-i+\lm}$ as well.
Consequently,
\[n \geqslant r_1 + \ldots + r_h \geqslant 
\left(\sum_{i=1}^{h-2} r_i\right) + \sm + 1 \geqslant 
\left(\sum_{i=1}^{h-2} 2^{h-2-i+\lm}\right) + 2^{\lm} + 1 = 2^{h-2+\lm} + 1,\]
and therefore $h < \log_2(n) + 2 - \lm$.
\end{proof}

As an immediate consequence, 
and since $\sm \geqslant 1$, i.e., $\lm \geqslant 0$,
we already know that
no stack of size larger than $\lceil \log_2(n) \rceil + 1$
will ever be required.
Then, in practice, it remains to check whether the
stack sizes currently used in the implementations
of \TS in both languages Python and Java would be sufficient for sorting arrays
of any size.
In fact, it might even be possible to use smaller
arrays than those currently in use,
but whether making such a change would be worth the effort is not clear.
Finally, we only focus here on arrays of \emph{meaningful} length,
which means that Python and Java cannot deal with arrays of arbitrary sizes;
we make that point clearer below.

\begin{corollary}\label{cor:stack-size-java}
Whenever sorting an array of \emph{meaningful} length,
the size of the array used to implement \TS's stack in both languages
Python and Java is large enough to also implement the stack of \cASS and \dASS.
\end{corollary}

\begin{proof}
Let $n$ be the length of the array to be sorted.
In Python, \TS's stack is simulated by an array of size
$h_{\max} = 85$ and contains only runs of size at least $32$~\cite{Peters2015b},
i.e., $\lm = 5$.
Thus, Proposition~\ref{pro:stack-size} proves that this array is large enough
whenever $n \leqslant 2^{89}$:
this is more than could ever be handled by an actual computer, and therefore
large enough for all reasonable purposes.

In Java, the situation is slightly different, because
we only have $\lm = 4$,
and \TS's stack is simulated by an array whose size depends on $n$ 
and is quite smaller~\cite{Bloch2013}.
More precisely, this size is the integer $h_{\max}$ defined by
\[h_{\max} = \begin{cases}
5 & \text{if } n \leqslant 119, \\
10 & \text{if } 120 \leqslant n \leqslant 1541, \\
24 & \text{if } 1542 \leqslant n \leqslant 119150, \\
49 & \text{if } 119151 \leqslant n \leqslant 2^{31}-5.\end{cases}\]
Finally, Java fails to handle and sort arrays of length $n \geqslant 2^{31}-4$,
which makes these cases irrelevant for our purposes.
We complete the proof by checking that 
$h_{\max} \geqslant \lceil \log_2(n) \rceil + 1 - \lm$
in all cases where $n \leqslant 2^{31}-5$.
\end{proof}

\subsection{Switching from \TS to \cASS in Python and Java}\label{sec:switching}

The structure of the merge policies of \TS and \cASS are remarkably similar.
However, a crucial difference is that, instead of comparing directly
the lengths of the runs involved, \cASS requires comparing their levels.
More precisely, a key step is to check efficiently whether
$\ell_{h-2} \leqslant \max\{\ell_h,\ell_{h-1}\}$,
which might be bothersome if implemented carelessly.
Fortunately, given three integers
$r_i$, $r_{i+1}$ and $r_{i+2}$,
checking whether $\ell_i \leqslant \max\{\ell_{i+1},\ell_{i+2}\}$
is made very easy by the use of boolean integer operations.
This is the object of the following two-line algorithm.

\begin{algorithm}[hb]
\begin{small}
\SetArgSty{texttt}
\DontPrintSemicolon
\SetKwInOut{Input}{Input}
\Input{Integers $r_i$, $r_{i+1}$ and $r_{i+2}$}
\KwResult{\true if $\ell_i \leqslant \max\{\ell_{i+1},\ell_{i+2}\}$, 
and \false otherwise.}
\SetKwInput{KwData}{Note}
\KwData{We use bit-wise \textbf{and}, \textbf{or} and
\textbf{not} binary operations on integers.}
\BlankLine
$x \gets r_{i+1}~\mathbf{or}~r_{i+2}$\;
\textbf{return} $x > (r_i~\mathbf{and}~(\mathbf{not}~x))$
\end{small}
\caption{Checking whether $\ell_i \leqslant \max\{\ell_{i+1},\ell_{i+2}\}$
\label{alg:check-L}}
\end{algorithm}

\begin{proposition}\label{pro:check-L}
When given positive integers $r_i$, $r_{i+1}$ and $r_{i+2}$ as input,
Algorithm~\ref{alg:check-L} returns \true if $\ell_i \leqslant \max\{\ell_{i+1},\ell_{i+2}\}$, 
and \false otherwise.
\end{proposition}

\begin{proof}
Consider the integers $\ell$, $x'$ and $\ell'$ respectively defined by
$\ell = \lfloor \log_2(x) \rfloor$, $x' = (r_i~\mathbf{and}~(\mathbf{not}~x))$,
and $\ell' = \lfloor \log_2(x') \rfloor$.
First, the relation
$\ell = \max\{\ell_{i+1},\ell_{i+2}\}$
holds regardless of whether $r_{i+1} < r_{i+2}$ or not.
Therefore, it remains to prove that $\ell_i \leqslant \ell$ if and only if
$x' < x$.
Then, and for all $k \geqslant 0$, the $k$\textsuperscript{th} bits of
$x$ and $x'$ cannot both be non-zero. It follows that $\ell \neq \ell'$.
Thus, we complete the proof by distinguishing two cases:
\begin{itemize}
 \item If $\ell_i \leqslant \ell$, then we already have 
 $\ell' \leqslant \ell_i$, by definition of $x'$, and therefore $\ell' \leqslant \ell$.
 Since $\ell \neq \ell'$, it follows that $\ell' < \ell$, and thus that $x' < x$.
 \item If $\ell_i > \ell$, then the $(\ell+1)$\textsuperscript{th} bit of both
 $r_i$ and $(\mathbf{not}~x)$ is non-zero, and therefore $\ell' \geqslant \ell$.
 Since $\ell \neq \ell'$, it follows that $\ell' > \ell$, and thus that $x' > x$.
\end{itemize}
\end{proof}

Consequently, and as promised, switching from \TS to \cASS would be extremely
easy in practice. For instance, in Python,
it would suffice to change the lines 1923 to 1934
of the implementation of \TS~\cite{Peters2015b} by the following 8 lines:

\begin{small}
\begin{lstlisting}
Py_ssize_t n = ms->n - 3;
if (n >= 0) {
    int x = p[n+1].len | p[n+2].len;
    if (x > (p[n].len & ~x)) {
        if (merge_at(ms, n) < 0)
            return -1;
    }
}
\end{lstlisting}
\end{small}

Similarly, in Java, it would suffice to change the lines 440 to 447
of the implementation of \TS~\cite{Bloch2013} by the following 5 lines:

\begin{small}
\begin{lstlisting}
int n = stackSize - 3;
int x = runLen[n+1] | runLen[n+2];
if (n < 0 || x <= (runLen[n] & ~x)) {
    break;
}
\end{lstlisting}
\end{small}

\end{document}